\documentclass{amsart}

\usepackage{amsmath}
\usepackage{amssymb,bm}
\usepackage{amsthm}
\usepackage{latexsym,hyperref}
\usepackage{enumitem,nccmath}
\usepackage{stmaryrd}

\usepackage{cancel}

\usepackage{color}
\usepackage[usenames,dvipsnames]{xcolor}
\usepackage{etoolbox}

\DeclareMathAlphabet{\mathbbm}{U}{bbm}{m}{n}

\DeclareMathAlphabet{\bbi}{U}{bbm}{m}{sl}

\hypersetup{%
    pdfborder = {0 0 0},
    colorlinks,
    citecolor=red!50!black,
    filecolor=Darkgreen,
    linkcolor=blue!40!black,
    urlcolor=cyan!50!black!90
}

\numberwithin{equation}{section}

\newtheorem{thrm}{Theorem}[section]
\newtheorem{prop}[thrm]{Proposition}
\newtheorem{crl}[thrm]{Corollary}
\newtheorem{lemma}[thrm]{Lemma}

\newtheorem{defn}[thrm]{Definition}

\theoremstyle{remark}
\newtheorem{rmk}[thrm]{Remark}
\newtheorem{exam}[thrm]{Example}

\newcommand{\nc}{\newcommand}

\DeclareMathOperator{\sgn}{sgn}
\DeclareMathOperator{\tr}{tr}

\nc{\Res}[1]{\underset{\;{#1}\;}{\rm Res}}

\nc{\End}{\mathrm{End}}
\nc{\Ext}{\mathrm{Ext}}
\nc{\Hom}{\mathrm{Hom}}
\nc{\Ima}{\mathrm{Image}}
\nc{\Ind}{\mathrm{Ind}}
\nc{\Ker}{\mathrm{Ker}}
\nc{\RHom}{\mathrm{RHom}}
\nc{\Sym}{\mathrm{Sym}}

\nc\bb{\mathbb}
\nc\mf{\mathfrak}
\nc\ms{\mathsf}
\nc\mc{\mathcal}
\nc\mr{\mathscr}
\nc\mt[1]{{\tt #1}}

\nc\ol{\overline}

\nc{\mfg}{\mf{g}}
\nc{\mfh}{\mf{h}}
\nc{\mfsl}{\mf{sl}}
\nc{\mfgl}{\mf{gl}}
\nc{\mfso}{\mf{so}}
\nc{\mfsp}{\mf{sp}}

\nc{\sca}{\mathscr{a}}
\nc{\scb}{\mathscr{b}}
\nc{\scc}{\mathscr{c}}
\nc{\scd}{\mathscr{d}}

\nc{\nn}{\nonumber}   
\nc{\el}{\nonumber\\} 

\nc{\equ}[1]{\begin{equation}#1\end{equation}}
\nc{\eqa}[1]{\begin{equation}\begin{alignedat}{50}#1\end{alignedat}\end{equation}}
\nc{\eqn}[1]{\begin{equation*}\begin{alignedat}{50}#1\end{alignedat}\end{equation*}}
\nc{\eqg}[1]{\begin{equation}\begin{gathered}#1\end{gathered}\end{equation}}

\nc{\ali}[1]{\begin{alignat}{50}#1\end{alignat}}
\nc{\als}[1]{\begin{subequations}\begin{alignat}{50}#1\end{alignat}\end{subequations}}
\nc{\aln}[1]{\begin{alignat*}{50}#1\end{alignat*}}

\nc{\gat}[1]{\begin{gather}#1\end{gather}}
\nc{\gas}[1]{\begin{subequations}\begin{gather}#1\end{gather}\end{subequations}}
\nc{\gan}[1]{\begin{gather*}#1\end{gather*}}

\nc{\mcA}{\mc{A}}
\nc{\mcB}{\mc{B}}
\nc{\mcU}{\mc{U}}
\nc{\mfU}{\mf{U}}
\nc{\mcP}{\mc{P}}
\nc{\mcQ}{\mc{Q}}
\nc{\mcX}{\mc{X}}
\nc{\mcZ}{\mc{Z}}
\nc{\mcT}{\mc{T}}
\nc{\mcG}{\mc{G}}
\nc{\msS}{\ms{S}}
\nc{\mss}{\ms{s}}
\nc{\msc}{\ms{c}}
\nc{\msd}{\ms{d}}
\nc{\msv}{\ms{v}}
\nc{\msq}{\ms{q}}
\nc{\msw}{\ms{w}}
\nc{\mcS}{\mc{S}}
\nc{\mcI}{\mc{I}}

\nc{\C}{\mathbb{C}}
\nc{\N}{\mathbb{N}}
\nc{\Z}{\mathbb{Z}}

\nc{\ot}{\otimes}
\nc{\op}{\oplus}

\nc{\lan}{\langle}
\nc{\ran}{\rangle}

\nc{\qu}{\quad}
\nc{\qq}{\qquad}

\nc\Tr{{\rm tr}}

\nc{\al}{\alpha}
\nc{\del}{\delta}
\nc{\eps}{\epsilon}
\nc{\veps}{\varepsilon}
\nc{\ga}{\gamma}
\nc{\Ga}{\Gamma}
\nc{\ka}{\kappa}
\nc{\la}{\lambda}
\nc{\La}{\Lambda}
\nc{\om}{\omega}
\nc{\Om}{\Omega}
\nc{\si}{\sigma}
\nc{\Si}{\Sigma}
\nc{\bsi}{\boldsymbol\sigma}
\nc{\bSi}{\boldsymbol\Sigma}
\nc{\Ups}{\upsilon}
\nc{\vphi}{\varphi}

\nc{\btau}{\boldsymbol\tau}
\nc{\bdel}{\boldsymbol\delta}

\nc{\id}{\mathrm{id}}
\nc{\gr}{\mathrm{gr}}

\nc{\lrh}{\leftrightharpoons}
\nc{\iso}{\stackrel{\sim}{\longrightarrow}}
\nc{\liso}{\stackrel{\sim}{\longleftarrow}}
\nc{\wh}{\widehat}
\nc{\wt}{\widetilde}
\nc{\tl}{\tilde}
\nc{\lra}{\longrightarrow}
\nc{\ra}{\rightarrow}
\nc{\into}{\hookrightarrow}
\nc{\onto}{\twoheadrightarrow}

\nc{\qdet}{{\rm qdet\,}}
\nc{\sdet}{{\rm sdet\,}}
\nc{\sign}{{\rm sign}}
\nc{\inv}{{\rm inv}}

\nc{\tran}{{\text{\tiny $\ms T$}}}

\nc{\F}[2]{F'^{\rho}_{#1#2}}

\nc{\mysum}{\textstyle\sum}
\nc{\mysuml}{\textstyle\sum\limits}

\nc{\TX}{X_\rho(\mfg_{2n},\mfg_{2n}^\theta)^{tw}}
\nc{\TXO}{X_\rho(\mfso_{2n},\mfso_{2n}^\theta)^{tw}}
\nc{\TXS}{X_\rho(\mfsp_{2n},\mfsp_{2n}^\theta)^{tw}}

\nc{\Bnr}{\mc{B}_\rho(n,r)}
\nc{\tBnr}{ \mc{B}^{\,\rm ex}_\rho(n,r)}

\usepackage{color}
\usepackage[usenames,dvipsnames]{xcolor}

\nc{\hilb}{\mathscr{H}}
\nc{\hilbb}{\wt{\mathscr{H}}}
\nc{\vac}{\Omega}

\usepackage[scr=boondoxo]{mathalfa}
\nc{\key}{{\mathscr{k}}}
\nc{\ley}{{\mathscr{l}}}
\nc{\lla}{\la\hspace{-1.5mm}\la}

\nc{\at}{\tilde{a}}
\nc{\tu}{\tilde{u}}
\nc{\tv}{\tilde{v}}
\nc{\tw}{\tilde{w}}

\nc{\even}[2]{\left\{ #2 \right\}^{#1}}

\makeatletter
\newcommand{\dwt}[1]{{%
  \mathpalette\double@widetilde{#1}%
}}
\newcommand{\double@widetilde}[2]{%
  \sbox\z@{$\m@th#1\widetilde{#2}$}%
  \ht\z@=.9\ht\z@
  \widetilde{\box\z@}%
}
\makeatother

\renewcommand{\,}{\kern 0.07em} 

\begin{document}

\title[Nested algebraic Bethe ansatz for orthogonal and symplectic open spin chains]{Nested algebraic Bethe ansatz \\ for orthogonal and symplectic open spin chains}

\begin{abstract} 
We present a nested algebraic Bethe ansatz for one-dimensional open $\mfso_{2n}$- and $\mfsp_{2n}$-symmetric spin chains with diagonal boundary conditions and described by the extended twisted Yangian. We use a generalization of the Bethe ansatz introduced by De Vega and Karowski which allows us to relate the spectral problem of a $\mfso_{2n}$- or $\mfsp_{2n}$-symmetric open spin chain to that of a $\mfgl_{n}$-symmetric open spin chain. We explicitly derive the structure of Bethe vectors and the nested Bethe equations. 
\end{abstract}

\author{Allan Gerrard}
\address{University of York, Department of Mathematics, York, YO10 5DD, UK.}
\email{ajg569@york.ac.uk}

\author{Vidas Regelskis}
\address{University of Hertfordshire, School of Physics, Astronomy and Mathematics, 
Hatfield AL10 9AB, UK and
Vilnius University, Institute of Theoretical Physics and Astronomy, Saul\.etekio av.~3, Vilnius 10257, Lithuania.}
\email{vidas.regelskis@gmail.com}



\maketitle

\setcounter{tocdepth}{1}
\tableofcontents

\setlength{\parskip}{1ex}

%

\section{Introduction}

The Bethe ansatz is a large collection of methods to find the spectrum and common eigenvectors of commuting families of operators (transfer matrices) occurring in the theory of quantum integrable models. 
It was Faddeev's Leningrad school of mathematical physics which reformulated the spectral problem of quantum integrable models into a question of representation theory of certain associative algebras with quadratic relations, now generally known as quantum groups \cite{FST,FT}. 
More precisely, the spaces of states of such models, called quantum spaces, are associated to tensor products of irreducible representations of these quantum groups. 
The commuting operators are then images of elements in the commutative subalgebra, known as Bethe subalgebra, on the quantum space.
By acting with appropriate algebra elements on the vacuum vectors one then constructs the so-called Bethe vectors that depend on sets of complex parameters. In the case when these parameters satisfy certain algebraic equations, known as Bethe equations, the corresponding Bethe vectors become eigenvectors of the commuting operators. 
In this form, the Bethe ansatz is called the algebraic Bethe ansatz.
The general conjecture is that the constructed eigenvectors form a basis in the space of states of the model, see reviews \cite{PRS3,Sl2}.

 The algebraic Bethe ansatz has been very fruitful in the study of $\mfgl_N$-symmetric integrable models \cite{KuRs,BeRa1,BeRa2,PRS1,PRS2,GMR}. Finding eigenvectors and their eigenvalues provides the necessary first step in the study of scalar products and norms \cite{HLPRS1,HLPRS2,Ko}, correlation functions and form factors \cite{IzKo,KKMST1,KKMST2,KMST,KMT,Sl1}.
The study of the $\mfso_N$- and $\mfsp_N$-symmetric models so far has been less productive. One of the obstacles is that the $R$-matrix in this case is not quite of a six-vertex type, which is the key property used in the study of the $\mfgl_N$-symmetric models. 
Another obstacle is that not every irreducible highest weight $\mfso_N$- or $\mfsp_N$-representation can be lifted to a representation of the corresponding quantum group, such as Yangian or quantum loop algebra. Moreover, the lifting itself is often not straightforward and requires use of the fusion procedure or some other method, such as a spinor or oscillator algebra realization \cite{AMR,GRW4}.
Consequently, the study of one-dimensional $\mfso_N$- or $\mfsp_N$-symmetric spin chains has mostly been restricted to the cases when the quantum space of the model is a tensor product of fundamental representations (``fundamental models''). 
The nested algebraic Bethe ansatz for fundamental periodic spin chains in the orthogonal case was addressed by De Vega and Karowski \cite{DVK} (see also \cite{Rs2}) and in the symplectic case by Reshetikhin \cite{Rs1}. The latter paper uses a combination of analytic and algebraic methods to study periodic spin chains with more general quantum spaces. 
The fundamental open spin chains in the symplectic case were addressed by Guang-Liang, Kang-Jie and Rui-Hong \cite{GKR} and, more recently, in the orthogonal case by Gombor and Palla \cite{Go,GoPa}. The algebraic Bethe ansatz for the ortho-symplectic (supersymmetric) closed spin chain was studied by Martins and Ramos \cite{MaRa}.
The analytical Bethe ansatz for orthogonal, symplectic and ortho-symplectic open spin chains was studied by Arnaudon et al.\ in \cite{AACDFR1, AACDFR2}.

In the present paper we study the spectral problem of $\mfso_N$- and $\mfsp_N$-symmetric open spin chains with more general quantum spaces and certain diagonal boundary conditions. The $R$-matrix of the spin chain is that introduced by Zamolodchikov and Zamolodchikov \cite{ZaZa} (see also \cite{BKWK}) for $\mfso_{N}$  and by Kulish and Sklyanin \cite{KuSk} for $\mfsp_{N}$. 
We focus on the case when $N=2n$. We choose the left boundary to be a trivial diagonal boundary. The right boundary is chosen to be a diagonal boundary corresponding to symmetric pairs of types CI, DIII, CII, DI and CD0 in terms of the Cartan's classification of symmetric spaces (the precise details are given in Section~\ref{sec:pairs}). 
Our approach relies on the decomposition $\C^{2n} \cong \C^2 \ot \C^n$, which allows us to rewrite the $R$-matrix as an $\End(\C^n\ot\C^n)$-valued six-vertex matrix and thus apply conventional algebraic Bethe ansatz methods, subject to necessary modifications, to solve the spectral problem of the chain. The space of states is given by a tensor product of symmetric irreducible $\mfso_{2n}$-representations or by a tensor product of skew-symmetric irreducible $\mfsp_{2n}$-representations. 
We use fusion procedure to extend these $\mfso_{2n}$- and $\mfsp_{2n}$-representations to those of the Extended Yangians $X(\mfso_{2n})$ and $X(\mfsp_{2n})$ studied in \cite{AMR}. 
The monodromy matrix of the chain satisfies the defining relations of the extended Twisted Yangian $X(\mfg_{2n},\mfg_{2n}^\rho)^{tw}$ studied by the second named author in \cite{GR}, introduced in Drinfel'd J presentation in \cite{DMS}.
The key idea is to use the symmetry relation of $X(\mfg_{2n},\mfg_{2n}^\rho)^{tw}$, which allows us to rewrite the exchange relations involving the $D$ operator of the transfer matrix in terms of the $A$ operator, thus effectively reducing the problem to that of the $\mfgl_n$-symmetric open spin chains studied by Belliard and Ragoucy in \cite{BeRa2}. 
Our main results are Theorems \ref{T:spn-spec} and \ref{T:son-spec} stating eigenvalues of symplectic and orthogonal Bethe vectors and the Bethe equations. Here by a symplectic Bethe vector we mean a Bethe vector for a $\mfsp_{2n}$-symmetric open spin chain. An orthogonal Bethe vector means a Bethe vector for a $\mfso_{2n}$-symmetric open spin chain. 
The eigenvalues of Bethe vectors and Bethe equations for fundamental chains in this paper agrees with the results obtained in \cite{GoPa} and \cite{AACDFR1}. The differences amount to a factor of $\frac{g(v)}{2v-2\ka-\rho}$ which we have introduced in the definition of the transfer matrix (Definition \ref{D:tm}, see also \eqref{SS}), and additional factors which appear in the definitions of the reflection $K$-matrices (given by Lemma~\ref{L:1-dim}).
We also present in Section~\ref{sec:NABA-RE} a detailed survey of the algebraic Bethe ansatz for a $\mfgl_n$-symmetric open spin chain studied in \cite{BeRa2}, since it is a prerequisite to our approach to $\mfsp_{2n}$- and $\mfso_{2n}$-symmetric open spin chains. 
Our main objectives in this survey are Theorems \ref{T:Bnr-spec} and \ref{T:tf-bnr}. Theorem~\ref{T:Bnr-spec} rephrases the relevant to the current paper results obtained in Section 6 of \cite{BeRa2}. Theorem \ref{T:tf-bnr} states a trace formula for Bethe vectors of a $\mfgl_n$-symmetric open spin chain. This formula may be viewed as a special case of the supertrace formula given by Theorem 7.1 in \cite{BeRa1}, which presented only an outline of the proof. In the current paper we provide a detailed proof of the trace formula under consideration. 
However, we were unable to find (reasonably simple) trace formulas for Bethe vectors of the $\mfsp_{2n}$- and $\mfso_{2n}$-symmetric open spin chains, and hence have limited ourselves to providing examples of the explicit form of Bethe vector with a small number of excitations. These are given in Examples \ref{E:BV-spn} and \ref{E:BV-son}.

We note for the reader that the approach presented in this paper may be used for any irreducible representations of $\mfsp_{2n}$ and $\mfso_{2n}$ that can be extended to representations of $X(\mfsp_{2n})$ and $X(\mfso_{2n})$. In Section \ref{sec:spin} we have demonstrated this in case of the so-called $SO_{2n}/(U_n \times U_n)$ and $SP_{2n}/(U_n \times U_n)$ magnets, in the periodic cases studied in \cite{Rs1}.

The paper is organized as follows. In Section 2 we introduce the notation used in the paper and provide details of the symmetric pairs that describe boundary conditions of the open spin chains.
In Section 3 we set up the algebraic description of the spin chains. We recall the definition of the orthogonal and symplectic extended Yangians and twisted Yangians, and relevant details of their representation theory. We then present the fusion procedure. We also recall the necessary details of the Molev-Ragoucy reflection algebra and present the six-vertex block-decomposition of the extended twisted Yangian.
In Section 4 we present the nested algebraic Bethe ansatz for an open $\mfgl_n$-symmetric open spin chain first addressed by Belliard and Ragoucy in \cite{BeRa1}.
Section 6 contains the main results of this work. We first derive technical identities that provide key steps of the algebraic Bethe ansatz. We introduce a creation operator of (multiple-)excitations and describe its algebraic properties: we derive the exchange relations for the creation operator and the monodromy matrix that lead to the so-called wanted and unwanted terms. We then present the nested algebraic Bethe ansatz. We provide the complete set of Bethe equations and examples of the Bethe vectors. 
We end this section with a brief discussion of the $SO_{2n}/(U_n \times U_n)$ and $SP_{2n}/(U_n \times U_n)$ magnets and the nearest-neighbour Hamiltonian operator for the fundamental open spin chain.

{\it Acknowledgements.} The authors thank Samuel Belliard, Eric Ragoucy, Niall MacKay and Curtis Wendlandt, for useful discussions and comments. A.G. was supported by an EPSRC PhD studenship. V.R. was supported by the European Social Fund, grant number 09.3.3-LMT-K-712-02-0017.
The authors gratefully acknowledge the financial support.

%

\section{Preliminaries and definitions}


\subsection{Lie algebras} \label{sec:Lie}

Fix $n\in\N$. Let $\mfgl_{2n}$ denote the general linear Lie algebra and let $E_{ij}$ with $1\le i,j \le 2n$ be the standard basis elements of $\mfgl_{2n}$ satisfying
\[
[E_{ij}, E_{kl}] = \del_{jk} E_{il} - \del_{il} E_{kj}.
\]
The orthogonal Lie algebra $\mfso_{2n}$ or the symplectic Lie algebra $\mfsp_{2n}$ can be realized as a Lie subalgebra of $\mfgl_{2n}$ as follows. For any $1 \le i,j\le 2n$ set $\theta_{ij}=\theta_i\theta_j$ with $\theta_i=1$ in the orthogonal case and $\theta_i=\del_{i>n}-\del_{i\le n}$ in the symplectic case.
Introduce elements $F_{ij}= E_{ij} - \theta_{ij} E_{2n-j+1,2n-i+1}$ satisfying the relations 
\gat{ 
\label{[F,F]}
[F_{ij},F_{kl}] = \del_{jk} F_{il} - \del_{il} F_{kj} + \theta_{ij} \del_{j,2n-l+1}F_{k,2n-i+1} - \theta_{ij}\del_{i,2n-k+1} F_{2n-j+1,l} ,
\\
\label{F+F=0}
F_{ij} + \theta_{ij} F_{2n-j+1,2n-i+1}=0,
}
which in fact are the defining relations of the Lie algebra $\mfso_{2n}$ or $\mfsp_{2n}$. Namely, we may identify $\mfso_{2n}$  or $\mfsp_{2n}$ with $\mathrm{span}_{\C} \{ F_{ij} : 1\le i,j\le 2n \}$ and we will use $\mfh_{2n}=\mathrm{span}_{\C} \{ F_{ii} : 1 \le i \le n \}$ as a Cartan subalgebra. It will be convenient to denote both Lie algebras $\mfso_{2n}$ and $\mfsp_{2n}$ simply by $\mfg_{2n}$.

For any $n$-tuple $\la=(\la_1,\dots,\la_n)\in \C^n$ we will denote by $V(\la)$ the irreducible highest weight representation of the Lie algebra $\mfg_{2n}$. In particular, $V(\la)$ is generated by a non-zero vector $\xi$ such that
\aln{
F_{ij}\,\xi &= 0 &&\qu\text{for}\qu &&1\le i<j \le 2n &&\qu\text{and} \\
F_{ii}\,\xi &= \la_i\,\xi &&\qu\text{for}\qu &&1\le i \le n . 
\intertext{The representation $V(\la)$ is finite-dimensional if and only if}
\la_i - \la_{i+1} &\in \Z_+ &&\qu\text{for}\qu &&i=1,\ldots,n-1 &&\qu\text{and} \\
\la_{n-1}+\la_n &\in \Z_+ &&\qu\text{if}\qu &&\mfg_{2n}=\mfso_{2n}, \\
\la_n &\in \Z_+ &&\qu\text{if}\qu &&\mfg_{2n}=\mfsp_{2n} .
}

The subalgebra of $\mfg_{2n}$ generated by the elements $F_{ij}$ with $1\le i,j\le n$ is isomorphic to the Lie algebra~$\mfgl_n$.
We will be interested in the following restriction of $V(\la)$:
\equ{
\ol{V}(\la) = \{ v\in V(\la) \,:\, F_{i,n+j}\,v = 0 \;\text{ for }\; 1\le i,j \le n  \} . \label{bV(la)}
}
The vector space $\ol{V}(\la)$ is an irreducible representation of $\mfgl_n\subset \mfg_{2n}$. It is finite-dimensional if $V(\la)$ is.

Given a Lie algebra $\mfg$ its universal enveloping algebra will be denoted by $U(\mfg)$.


\subsection{Matrix operators} \label{sec:notation}

We need to introduce some operators on $\C^N\ot \C^N$, where the tensor product $\ot$ is defined over the field of complex numbers, $\ot = \ot_\C$, and $N=n$ or $N=2n$ (it will always be clear from the context which $N$ is used). Let $e_{ij}\in\End(\C^N)$ be the standard matrix units and let $e_i$ be the standard basis vectors of $\C^N$. Then $P$ will denote the permutation operator on $\C^N \ot \C^N$ and we set $Q=P^{t_1}=P^{t_2}$ where the transpose $t$ is defined by $(e_{ij})^t = \theta_{ij}\, e_{N-j+1,N-i+1}$; explicitly,
\eq{ \label{PQ}
P = \sum_{i,j=1}^{N} e_{ij} \ot e_{ji}, \qq 
Q = \sum_{i,j=1}^{N} \theta_{ij} e_{ij} \ot e_{N-i+1,N-j+1} .
}

Let $I$ denote the identity matrix on $\C^N\ot \C^N$ or $\C^N$ (again, it will always be clear from the context which $I$ is used). Then $P^2=I$, $PQ=QP=\pm Q$, $Q^2=N Q$, which will be useful below. Here, and further in this paper, the upper sign in $\pm$ and $\mp$ corresponds to the orthogonal (or ``$+$'') case and the lower sign to the symplectic (or ``$-$'') case.  
Also note that $P (e_{ij} \ot I) = (I \ot e_{ij}) P$. Taking the transpose of this, we obtain a pair of relations for $Q$ and any $M\in\End(\C^N)$:
\equ{ \label{Q_rel}
Q\, (M \ot I) =  Q\, (I \ot M^t) , \qq  (M \ot I)\, Q  = (I \ot M^t)\, Q .
}
Note that the transposition $t$ can we equivalently written as
\[
M^t = J M^\tran J
\]
where $\tran$ denotes the usual transposition of matrices, i.e., $e^\tran_{ij}= e_{ji}$, and $J=\sum_{i=1}^{N} e_{i,N-i+1}$.

For a matrix $X$ with entries $x_{ij}$ in an associative (or Lie) algebra $A$ we write
\[
X_s = \sum_{i,j=1}^{N} \underbrace{ I \ot \cdots \ot I}_{s-1} \ot \;e_{ij} \ot I \ot \cdots \ot I \ot x_{ij} \in \End(\C^N)^{\ot k} \ot A .
\]
Where appropriate, we will use the notation $[X]_{ij}$ to denote matrix elements $x_{ij}$ of a matrix operator $X$. 
Products of matrix operators will be ordered using the following rules:
\equ{ \label{orderprod} 
\prod_{i=1}^s X_i = X_1 \, X_2 \cdots X_s  \qq \text{and} \qq \prod_{i=s}^1 X_i = X_s \, X_{s-1} \cdots X_1.  
}
Here $k \geq 2$ and $1\le s\le k$; it will always be clear from the context what $k$ is.

We will denote the generating matrix of $\mfgl_{2n}$ by $E=\sum_{1\le i \le j\le 2n}e_{ij} \ot E_{ij}$ and the generating matrix of $\mfg_{2n}$ by $F= \sum_{i,j=1}^{2n} e_{ij} \ot F_{ij}$.


\subsection{Symmetric pairs}  \label{sec:pairs}

The symmetric pairs that we are interested in are of the form $(\mfg_{2n},\mfg_{2n}^{\rho})$, where $\rho$ is an involution of $\mfg_{2n}$ and $\mfg_{2n}^\rho$ denotes the $\rho$-fixed subalgebra of $\mfg_{2n}$. 
The involution $\rho$ is given by $\rho(F) = G F G^{-1}$ for a particular matrix $G\in \End(\C^{2n})$; we will use the matrices $G$ in agreement with those in Section 2.2 of \cite{GRW1}. 
This allows us to view $\mfg_{2n}^\rho$ as the subalgebra of $\mfg_{2n}$ generated by the elements $F^\rho_{ij} = F_{ij} + (G F G^{-1})_{ij}$. Its generating matrix is given by $F^\rho = F + G F G^{-1}$.
We also recall the further refinement of Cartan's classification of symmetric spaces introduced in {\it cit.~loc.} that reflects the explicit form of $G$ listed below and differences in the study of representation theory of twisted Yangians.

Let $p$ and $q$ be such that $p\ge q>0$ and $p+q=2n$. In the list below, for each Cartan type, we indicate the corresponding symmetric pair and give our choice of matrix $G$:

\begin{itemize} [itemsep=0.75ex]

\item CI\hspace{.55cm}: $(\mfg_{2n},\mfg_{2n}^\rho)=(\mfsp_{2n},\mfgl_{n})$ and $G=\sum_{i=1}^{n} (e_{ii} - e_{n+i,n+i})$. 

\item DIII\hspace{.273cm}: $(\mfg_{2n},\mfg_{2n}^\rho)=(\mfso_{2n},\mfgl_{n})$ and $G=\sum_{i=1}^{n} (e_{ii} - e_{n+i,n+i})$. 

\item CII\hspace{.42cm}: $(\mfg_{2n},\mfg_{2n}^\rho)=(\mfsp_{2n},\mfsp_p\op\mfsp_q)$ such that both $p$ and $q$ are even and $p\ge q$. The matrix $G$ is
\equ{
G= \sum_{i=\frac{q}{2}+1}^{2n-\frac{q}{2}} e_{ii} - \sum_{i=1}^{\frac{q}{2}} (e_{ii} + e_{2n-i+1,2n-i+1}) . \label{G:CII}
}
In this case the subalgebra of $\mfg_{2n}^{\rho}$ spanned by $F_{ij}$ with $\frac{q}{2}+1\le i,j\le 2n-\frac{q}{2}$ is isomorphic to $\mfsp_p$ and the subalgebra of $\mfg_{2n}^{\rho}$ spanned by $F_{ij}$ with $1\le i,j\le \frac{q}{2}$ and $2n-\frac{q}{2}+1\le i,j\le 2n$ is isomorphic to $\mfsp_q$.

\item DI\hspace{.54cm}: $(\mfg_{2n},\mfg_{2n}^\rho)=(\mfso_{2n},\mfso_p\op\mfso_q)$ such that both $p$ and $q$ are even and $p\ge q$. We choose $G$ to be the same as for CII case, i.e.\@ given by \eqref{G:CII}. Hence the subalgebras $\mfso_p$ and $\mfso_q$ of $\mfg^\rho_{2n}$ are defined analogously.

\item CD0\hspace{.24cm}: $(\mfg_{2n},\mfg_{2n}^{\rho}) = (\mfg_{2n},\mfg_{2n})$ and $G=I$. 

\end{itemize}

Note that we have excluded the DI case, when both $p$ and $q$ are odd (called DI(b) in \cite{GRW1}). In this case the matrix $G$ can not be chosen to be diagonal. Also note that the last case, CD0, can be viewed as a limiting case of types CII and DI, when $p=2n$ and $q=0$.


\section{Setting up symmetries and representations of the spin chain}


\subsection{The Yangian $X(\mfg_{2n})$} \label{sec:X}

We briefly recall necessary details of the Extended Yangian $X(\mfg_{2n})$ and its representation theory, adhering closely to \cite{AMR}. We will drop ``Extended'' part of the name to ease the notation.  We then use the fusion procedure of \cite{IMO} and follow arguments presented in Sections 6.4 and 6.5 of \cite{Mo3} (see also Section 2 in \cite{MoMu}) to extend symmetric representations of $\mfso_{2n}$ and skew-symmetric representations of $\mfsp_{2n}$ to representations of $X(\mfg_{2n})$. They are examples of the so-called Kirillov-Reshetikhin modules of $X(\mfg_{2n})$ \cite{KrRs}. A multiple tensor product of such representations will serve as the bulk quantum space of the open spin chains.

Introduce a rational function acting on $\C^{2n}\ot \C^{2n}$
\equ{
R(u) = I - \frac1u\,P + \frac1{u-\ka}\,Q,  \qu\text{where}\qu \ka = n \mp 1, \label{R(u)}
}
called {\it Zamolodchikov's $R$-matrix} \cite{KuSk,ZaZa}. It satisfies the unitarity and cross-unitarity relations 
\[
R(u)\, R(-u) = R(u)\, R^t(u+\ka) = (1-u^{-2})\cdot I
\] 
and is a solution of the quantum Yang-Baxter equation,
\equ{
R_{12}(u-v)\,R_{13}(u-z)\,R_{23}(v-z) = R_{23}(v-z)\,R_{13}(u-z)\,R_{12}(u-v). \label{YBE}
}

We introduce elements $t_{ij}^{(r)}$ with $1 \le i,j \le {2n}$ and $r\ge 0$ such that $t^{(0)}_{ij}= \del_{ij}$. Combining these into formal power series $t_{ij}(u) = \sum_{r\ge 0} t_{ij}^{(r)} u^{-r}$, we can then form the generating matrix $T(u)= \sum_{i,j=1}^{2n} e_{ij} \ot t_{ij}(u)$.

\begin{defn}
The Yangian $X(\mfg_{2n})$ is the unital associative $\C$-algebra generated by elements $t_{ij}^{(r)}$ with $1 \le i,j \le {2n}$ and $r\in\Z_{\ge 0}$ satisfying the relation
\equ{ 
R(u-v)\,T_1(u)\,T_2(v) = T_2(v)\,T_1(u)\,R(u-v) . \label{Y:RTT} 
}
The Hopf algebra structure of $X(\mfg_{2n})$ is given by 
\equ{ \label{Hopf:Y}
\Delta: t_{ij}(u)\mapsto \sum_{k=1}^{2n} t_{ik}(u)\ot t_{kj}(u), \qq S: T(u)\mapsto T^{-1}(u),\qquad \veps: T(u)\mapsto I. 
}
\end{defn}

We now collect several useful facts about the algebra $X(\mfg_{2n})$. The matrix $T(u)$ satisfies the symmetry (cross-unitarity) relation
\equ{
T(u)\,T^t(u+\ka)\, = T^t(u+\ka)\, T(u)\, = z(u)\,I, \label{tsymm}
}
where $z(u)$ is a formal series in $u^{-1}$ with coefficients central in $X(\mfg_{2n})$.
Let $c\in\C$ and $f(u)\in\C[[u^{-1}]]$. The {\it shift} and {\it twist} automorphisms of $X(\mfg_{2n})$ are defined by, respectively,
\equ{
\tau_c \;:\; T(u) \mapsto T(u-c), \qq \mu_f \;:\; T(u) \mapsto f(u)\,T(u). \label{Aut1}
}
We will make us of the following anti-automorphisms of $X(\mfg_{2n})$:
\equ{
sign \;:\; T(u) \mapsto T(-u) , \qq tran \;:\; T(u) \mapsto T^{\,\tran}(u), \qq rev \;:\; T(u) \mapsto T^t(u) . \label{Aut2}
}

Next, we recall the definition of the lowest weight representation of $X(\mfg_{2n})$. 

\begin{defn}
A representation $V$ of $X(\mfg_{2n})$ is called a lowest weight representation if there exists a non-zero vector $\eta\in V$ such that $V=X(\mfg_{2n})\,\eta$ and
\equ{ 
t_{ij}(u)\,\eta = 0 \qu\text{for}\qu 1 \le j<i \le 2n 
\qu\text{and}\qu 
t_{ii}(u)\,\eta=\la_i(u)\,\eta \qu\text{for}\qu 1 \leq i \leq 2n,
}
where $\la_i(u)$ is a formal power series in $u^{-1}$ with a constant term equal to $1$. The vector $\eta$ is called the lowest vector of $V$ and the $2n$-tuple $\la(u)=(\la_1(u),\ldots,\la_{2n}(u))$ is called the lowest weight of $V$. 
\end{defn}

The Yangian $X(\mfg_{2n})$ contains the universal enveloping algebra $U(\mfg_{2n})$ as a Hopf subalgebra. An embedding  $U(\mfg_{2n})\into X(\mfg_{2n})$ is given by 
\equ{
F_{ij} \mapsto \tau^{(1)}_{ij} := \tfrac12(t_{ij}^{(1)}-\theta_{ij}t^{(1)}_{2n-j+1,2n-i+1}) \label{tau_ij}
}
for all $1\le i,j\le 2n$. We will identify $U(\mfg_{2n})$ with its image in $X(\mfg_{2n})$ under this embedding. 
However, in contrast to the Yangian $Y(\mfgl_{2n})$ of the Lie algebra $\mfgl_{2n}$, there is no surjective homomorphism from the Yangian $X(\mfg_{2n})$ onto the algebra $U(\mfg_{2n})$. As a consequence, not every irreducible finite-dimensional representation of $\mfg_{2n}$ can be extended to a representation of $X(\mfg_{2n})$. 
The fusion procedure allows us to extend any symmetric representation of $\mfso_{2n}$ and any skew-symmetric representation of $\mfsp_{2n}$ to a representation of $X(\mfg_{2n})$. 
In the remaining part of this section we briefly recall the main aspects of the fusion procedure starting with the vector representation of $\mfg_{2n}$. 

The vector representation of $\mfg_{2n}$ on $\C^{2n}$ is a highest weight representation of weight $\la=(1,0,\ldots,0)$ and the highest vector $e_1$ given by the assignment $F_{ij} \mapsto e_{ij} - \theta_{ij}\,e_{2n-j+1,2n-i+1}$. The assignment
\[
\varrho \;:\; t_{ij}(u) \mapsto \del_{ij} + \frac1u \, e_{ij} - \frac1{u+\ka}\, \theta_{ij}\, e_{2n-j+1,2n-i+1} 
\]
equips $\C^{2n}$ with a structure of a $X(\mfg_{2n})$-module. Since we are interested in the lowest weight $X(\mfg_{2n})$-modules, we need to compose the map $\varrho$ with the anti-automorphisms $sign$ and $tran$. We also include the shift automorphism $\tau_c$. Denoting the resulting map by $\bm\varrho_c := \varrho \circ sign \circ tran \circ \tau_c$ we have
\[
\bm\varrho_c \;:\; t_{ij}(u) \mapsto \del_{ij} - \frac1{u-c} \, e_{ji} + \frac1{u-c-\ka}\, \theta_{ij}\, e_{2n-i+1,2n-j+1} .
\]
It follows that 
\gan{
\bm\varrho_c(T(u))= R(u-c) , \qq \bm\varrho_c(T(u))\,\bm\varrho_{-c}(T(-u)) = \bm\varrho_c(T(u))\,\bm\varrho_{c}(T^t(u+\ka)) = 1-\frac1{(u-c)^2}.
}
This allows us to view the space $\C^{2n}$ as an irreducible lowest weight $X(\mfg_{2n})$-module with weight $\la(u)$ given~by
\equ{
\la_1(u) = 1-\frac{1}{u-c}, \qq
\la_2(u)=\ldots=\la_{2n-1}(u)=1, \qq
\la_{2n}(u) = 1 + \frac{1}{u-c-\ka}. \label{la-fund}
}
We denote this module by $L(\la)_c$. We will use this notation for all irreducible finite-dimensional representations of $\mfg_{2n}$ that can be equipped with a structure of a $X(\mfg_{2n})$-module.

Consider the tensor product space $(\C^{2n})^{\ot k}$ with $k\ge 2$. Each $\C^{2n}$ carries the vector representation of $\mfg_{2n}$ so that the vector space $(\C^{2n})^{\ot k}$ is a representation of $\mfg_{2n}$. The Brauer algebra $\mf{B}_k(\pm 2n)$ acts naturally on this tensor space and commutes with the action of $\mfg_{2n}$, see e.g.~Chapter 10 of \cite{GmWa}. The Brauer-Schur-Weyl duality allows us to obtain irreducible representations of $\mfg_{2n}$ by studying primitive idempotents in $\mf{B}_k(\pm 2n)$. Recall that irreducible representations of $\mf{B}_k(\pm 2n)$ are labelled by all partitions $\la=(\la_1,\la_2,\ldots)$ of the non-negative integers $k$, $k-2$, $k-4$, \dots . Denote by $\la'$ the partition conjugate to $\la$, e.g.~if $\la=(2,1,1)$, then $\la'=(3,1)$. Then the vector space $(\C^{2n})^{\ot k}$ decomposes as 
\[
(\C^{2n})^{\ot k} \cong \bigoplus_{f=0}^{\lfloor k/2 \rfloor} \bigoplus_{\substack{\la \vdash k-2f\\\la'_1+\la'_2\le 2n}} V_\la \ot L(\la)
\]
in the orthogonal case, and as
\[
(\C^{2n})^{\ot k} \cong \bigoplus_{f=0}^{\lfloor k/2 \rfloor} \bigoplus_{\substack{\la \vdash k-2f\\2\la'_1\le 2n}} V_{\la'} \ot L(\la)
\]
in the symplectic case; here $V_\la$ and $L(\la)$ are irreducible representations of $\mf{B}_k(\pm 2n)$ and $\mfg_{2n}$, respectively, labelled by the partition $\la$. 
We will focus on the symmetric representation labelled by the partition $(k)$ and the skew-symmetric representation labelled by the partition $(1,\dots,1)$ of $k$. Assume that $k\ge 1$ in the orthogonal case and $1\le k\le n$ in the symplectic case. By Theorem 2.2 of \cite{IMO} (see also Example~2.4 (iii) and Section~4 therein) the corresponding primitive idempotents act on the space $(\C^{2n})^{\ot k}$ via operators $\Pi^\pm_{k}$ defined by
\equ{
\Pi^\pm_{1} = 1 \qu\text{and}\qu \Pi^\pm_{k} = \frac{1}{k!}\,\prod_{i=2}^k \Big( R_{1i}(\mp (i-1) ) \cdots R_{i-1,i}(\mp1) \Big) \qu\text{if}\qu k\ge 2. \label{Proj}
}
The subspace $L^\pm_k = \Pi^\pm_{k} (\C^{2n})^{\ot k}$ is a $\mfg_{2n}$-submodule of $(\C^{2n})^{\ot k}$ isomorphic to the highest weight representation $L(\la)$ of weight $\la = (k,0,\ldots,0)$ in the orthogonal case and of weight $\la = (1,\ldots,1,0,\ldots,0)$, where the number of $1$'s is $k$, in the symplectic case. The highest vector in the orthogonal case is 
\[
\xi = e_1 \ot \cdots \ot e_1 .
\]
In the symplectic case it is
\[
\xi = \sum_{\si\in\mf{S}_k} \sign(\si)\; e_{\si(1)} \ot \cdots \ot e_{\si(k)} ,
\]
where $\mf{S}_k$ is the group of permutations on the set $\{1,2,\dots,k\}$.

By combining the comultiplication in \eqref{Hopf:Y} with the map $\bm\varrho_c$ and an appropriate choice of the shift automorphisms, we obtain a representation of $X(\mfg_{2n})$ on the vector space $(\C^{2n})^{\ot k}$ given by the assignment 
\equ{
T(u) \mapsto R_{01}(u-c)\, R_{02}(u-c\mp1) \cdots R_{0k}(u-c\mp k\pm 1) \in \End((\C^{2n})^{\ot (k+1)}) \label{T->RRR}
}
where the ``zero'' space denotes the matrix space of $T(u)$.

\begin{prop} \label{P:X-rep}
The subspace $L^\pm_k\subset (\C^{2n})^{\ot k}$ is stable under the action of $X(\mfg_{2n})$ defined by \eqref{T->RRR}. Moreover, the representation of $X(\mfg_{2n})$ on $L^\pm_k$ obtained by restriction is an irreducible lowest weight representation of weight $\la(u)$ given by, for $1\le i\le n$,
\equ{
\la_i(u) = 1-\frac{\la_i}{u-c} , \qq \la_{2n-i+1}(u) = 1 + \dfrac{\la_i}{u-c\mp k\pm1-\ka} , \label{l(u)-fused}
}
where $\la = (k,0,\ldots,0)$ in the orthogonal case and $\la = (1,\ldots,1,0,\ldots,0)$, with the number of $1$'s being $k$, in the symplectic case.
\end{prop}

\begin{proof}
Using the explicit form of the idempotent $\Pi^\pm_k$ and the Yang-Baxter equation multiple times we find
\aln{
& R_{01}(u-c)\, R_{02}(u-c\mp1) \cdots R_{0k}(u-c\mp k\pm 1) \, \Pi^\pm_k \\ & \qq = \Pi^\pm_k\, R_{0k}(u-c\mp k\pm 1) \cdots R_{02}(u-c\mp1)\,R_{01}(u-c) ,
}
which implies the first part of the proposition. Since $U(\mfg_{2n})\subset X(\mfg_{2n})$ we have  $X(\mfg_{2n})(e_1 \ot \cdots \ot e_1) = L^\pm_k$. 
By Lemma 5.17 in \cite{AMR} adapted to lowest weight representations, the tensor product of lowest vectors $e_1 \ot \cdots \ot e_1$ is again a lowest vector of weight given by the product of the individual weights with respect to the action \eqref{T->RRR}, namely $\prod_{j=0}^{k-1} \la_i(u\mp j)$, where $\la_i(u\mp j)$ are those given by \eqref{la-fund}. 
This implies the second part of the proposition for the orthogonal case. For the symplectic case we refer the reader to the proof of Theorem 5.16 in \cite{AMR}.
\end{proof}

These representations of $X(\mfg_{2n})$ will be denoted by $L(\la)_{c}$.
We define the Lax operator $\mc{L}(u)$ of $X(\mfg_{2n})$ by $T(u)\cdot L(\la)_{c} = \mc{L}(u-c)\, L(\la)_{c}$. 
It will be useful to know that
\equ{
\mc{L}(u)\,\mc{L}^t(u+\ka) = \mc{L}^t(u+\ka)\,\mc{L}(u) =  \prod_{i=0}^{k-1} \frac{(u\mp i)^2-1}{(u\mp i)^2} \cdot I  = \frac{u\pm1}{u}\cdot\frac{u\mp k}{u\mp k\pm1}\cdot I \label{Lax-cross}
}
which follows from the relations $R(u)\,R^t(\ka+u)=R^t(\ka+u)\,R(u)=(1-u^{-2})\,I$ and \eqref{T->RRR}.

\begin{rmk}
In the present work we do not need to know the explicit form of the Lax operators $\mc{L}(u)$. We nevertheless provide an example of $\mc{L}(u)$ in the case when $\mfg_{2n}=\mfsp_4$ and $k=2$. Then $\Pi^-_2 = \frac12R_{12}(1)$ projects the space $\C^4 \ot \C^4$ to the 5-dimensional subspace $L_2^-$, an irreducible highest-weight representation of $\mfsp_4$ of weight $\la=(1,1)$. Choose
\gan{
v_{-2} = \tfrac{1}{\sqrt{2}}\, e_{1} \wedge e_{2} ,\qu
v_{-1} = \tfrac{1}{\sqrt{2}}\,e_{1} \wedge e_{3} , \qu
v_{1} = \tfrac{1}{\sqrt{2}}\,e_{2} \wedge e_{4} , \qu
v_{2} = \tfrac{1}{\sqrt{2}}\,e_{3} \wedge e_{4} , \\
v_{0} = \tfrac{1}{2}\,(e_{2}\ot e_3 - e_{3}\ot e_{2} - e_{1} \ot e_{4} - e_{4} \ot e_{1}) , \qu
}
where $a \wedge b = a\ot b - b\ot a$, to be an orthonormal basis of $L_2^-$. Let $x_{ij} \in \End(L_2^-)$ denote the matrix units of $\End(L_2^-)$ with respect to the above basis, namely $x_{ij} v_k= \del_{jk} v_i$ for all $i,j,k$. Then the Lax operator can be written as $\mc{L}(u) = \sum_{i,j,k,l} \mathscr{l}_{ijkl}(u)\, e_{ij} \ot x_{kl}$ where $\mathscr{l}_{ijkl}(u) = (e^*_i \ot v^*_k)\,R_{01}(u)R_{02}(u+1)\,(e_j \ot v_{l})$. In particular,
\[
\mc{L}(u) = \frac{u-1}{u-2}\,\Big(I - \frac{2}{u} (P+\bar P)\Big) ,
\]
where 
\aln{
P &= \tfrac{1}{\sqrt{2}}\, \big( (e_{12} - e_{34}) \ot (x_{0,-1}-x_{10}) - (e_{13} + e_{24} ) \ot (x_{0,-2}+x_{20}) \big) \\
& \qu + e_{33}\ot (x_{-1,-1}+ x_{22}) + e_{44}\ot (x_{11} + x_{22}) - e_{14}\ot (x_{1,-2}+ x_{2,-1}) + e_{23}\ot (x_{-1,-2}+ x_{21}) 
}
and $\bar{P}$ is obtained from $P$ using the transposition rule $e_{ij} \ot x_{kl} \to e_{5-i,5-j} \ot x_{-k-l}$.
\end{rmk}

%

\subsection{The twisted Yangian $\TX$} \label{sec:TX}

We now focus on the Extended twisted Yangian $\TX$ and its representation theory adhering closely to \cite{GR,GRW1,GRW2}. As before, we drop the ``Extended'' part of the name to simplify the notation. We introduce an additional ``shift'' parameter $\rho\in\C$ in the definition of $\TX$ which will play a role in the algebraic Bethe anstaz discussed in Sections \ref{sec:NABA-RE}~and~\ref{sec:NABA-X}.  

Recall the definition of the matrix $G$ from Section \ref{sec:pairs}. Introduce a matrix-valued rational function
\equ{
G(u) = \frac{d I - u\,G}{d-u} \qq\text{where}\qq d=\frac14 \tr G, \label{G(u)}
}
so that $d=0$ for symmetric pairs CI and DIII, $d=n/2$ for CD0, and $d=(p-q)/4$ for CII and DI.

\begin{defn} 
The twisted Yangian $\TX$ is the subalgebra of $X(\mfg_{2n})$ generated by the coefficients of the entries of the matrix 
\eq{
\Si(u) = T(u-\tfrac\ka2)\,G(u+\tfrac\rho2)\,T^t(\tu-\tfrac\ka2) \qq\text{where}\qq \tu=\ka-u-\rho. \label{S=TT}
}
\end{defn}

The ``$\rho$-shifted'' twisted Yangian defined above is isomorphic to the one introduced by one of the authors in \cite{GR}. The isomorphism is provided by the map $\Si(u) \mapsto S(u+\tfrac\rho2)$. (Note that $\Si(u)$ is used to denote the special twisted Yangian in \cite{GR}.) The Lemma below is due to Lemmas~4.1 and 4.3 in \cite{GR}.

\begin{lemma} 
The matrix $\Si(u)$ defined in \eqref{S=TT} satisfies the reflection equation and the symmetry relation:
\gat{
R(u-v)\,\Si_1(u)\,R(u+v+\rho)\,\Si_2(v) = \Si_2(v)\,R(u+v+\rho)\,\Si_1(u)\,R(u-v) , \label{RE}
\\
\Si^t(u) = (\pm)\,\Si(\tu) \pm \frac{\Si(u)-\Si(\tu)}{u-\tu} + \frac{\tr(G(u+\tfrac\rho2))\,\Si(\tu) - \tr(\Si(u))\cdot I}{u-\tu-\ka} , \label{symm}
}
where the lower sign in $(\pm)$ distinguishes symmetric pairs CI and DIII from the remaining ones.
\end{lemma}

The relations \eqref{RE} and \eqref{symm} are in fact the defining relations of $\TX$. Their form in terms of matrix elements $\si_{ij}(u)$ of $\Si(u)$, for $\rho=0$, can be found in (4.4) and (4.5) of \cite{GR} (note that indices $i,j,k,l$ are indexed by $-n,-n+1,\ldots,n-1,n$ in \cite{GR}). In this work we will utilize a special ``block'' form of the defining relation; these are discussed in Section \ref{sec:block}.

We want to obtain a more compact form of the symmetry relation \eqref{symm}. Introduce a rational function
\equ{
g(u) = \begin{cases}
1 &\text{ for CI, DIII}, \\[.5em]
2u-\ka\pm1+\rho &\text{ for CII, DI when } p=q, \\[.25em]
\dfrac{u-\tu-\ka}{\tr(G(u+\tfrac\rho2))} &\text{ for CD0 and CII, DI when } p>q. \label{g(u)}
\end{cases}
} 
Note that in the last case we have
\[
\frac{u-\tu-\ka}{\tr(G(u+\tfrac\rho2))} = \frac{(u-\ka+\tfrac\rho2)(u-d+\tfrac\rho2)}{d\,(2u-n+\rho)} .
\]
Define the matrix
\equ{
S(u) = g(u)\,\Si(u) \in \TX ((u^{-1})). \label{SS}
}

\begin{lemma}
The matrix $S(u)$ satisfies the ``compact'' symmetry relation:
\gat{
S^t(u) = -\bigg(1\pm\frac{1}{u-\tu} \bigg)\,S(\tu) \pm \frac{S(u)}{u-\tu} - \frac{\tr(S(u))\cdot I}{u-\tu-\ka} . \label{bsymm}
}
\end{lemma}

\begin{proof} 
Substituting \eqref{SS} to \eqref{bsymm} gives
\eqa{
\Si^t(u) = -\frac{g(\tu)}{g(u)}\,\bigg(1\pm\frac{1}{u-\tu} \bigg)\,\Si(\ka- u-\rho) \pm \frac{\Si(u)}{u-\tu}  - \frac{\tr(\Si(u))\cdot I}{u-\tu-\ka} . \label{symm-1}
}
For symmetric pairs CI and DIII we have $g(u) = 1$ giving
\[
-\frac{g(\tu)}{g(u)}\,\bigg(1\pm\frac{1}{u-\tu} \bigg)=-1\mp\frac{1}{u-\tu} .
\]
For symmetric pairs CII and DI when $p=q$ we have instead $g(u)=2u-\ka+1+\rho$ and so
\[
-\frac{g(\tu)}{g(u)}\,\bigg(1\mp \frac{1}{u-\tu} \bigg) = 1\mp \frac{1}{u-\tu}.
\]
Thus for the above symmetric pairs \eqref{symm-1} becomes
\[
\Si^t(u) = \bigg((\pm)1\mp\frac{1}{u-\tu} \bigg)\,\Si(\ka- u-\rho) \pm \frac{\Si(u)}{u-\tu} - \frac{\tr(\Si(u))\cdot I}{u-\tu-\ka} ,
\]
which is equivalent to \eqref{symm}, since the above cases have $\tr (G(u)) = 0$. 

Let us now focus on all the remaining symmetric pairs. 
By Lemma 2.2 in \cite{GRW1} the matrix $G(u)$ itself satisfies the symmetry relation \eqref{symm}, namely
\[
G^t(u+\tfrac\rho2) = G(\ka-u-\tfrac\rho2) \pm \frac{G(u+\tfrac\rho2) - G(\ka-u-\tfrac\rho2)}{u-\tu} + \frac{\tr(G(u+\tfrac\rho2))\,G(\ka-u-\tfrac\rho2) - \tr(G(u+\tfrac\rho2))\cdot I}{u-\tu-\ka}.
\]
Recall \eqref{g(u)}. Taking the trace of both sides we find
\[
-\frac{u-\tu-\ka}{2u+\rho}\,\bigg( 1 \mp \frac{1}{u-\tu} + \frac{2\ka \pm 2}{u-\tu-\ka}\bigg)\,g(\tu) = \bigg( 1 \mp \frac{1}{u-\tu} + \frac{g^{-1}(u)}{u-\tu-\ka}\bigg) \,g(u)
\]
and rearrange to
\[
- \frac{g(\tu)}{g(u)}\,\Big(1\pm\frac{1}{u-\tu}\Big)  = \bigg( 1 \mp \frac{1}{u-\tu} + \frac{g^{-1}(u)}{u-\tu-\ka}\bigg) .
\]
This allows us to rewrite \eqref{symm-1} as
\aln{
\Si^t(u) &= \bigg( 1 \mp \frac{1}{u-\tu} + \frac{g^{-1}(u)}{u-\tu-\ka}\bigg)\,\Si(\ka- u-\rho) \pm \frac{\Si(u)}{u-\tu}  - \frac{\tr(\Si(u))\cdot I}{u-\tu-\ka}
\\ 
&= \Si(\tu) \pm \frac{\Si(u)-\Si(\tu)}{u-\tu} + \frac{\tr(G(u+\tfrac\rho2))\,\Si(\tu) - \tr(\Si(u))\cdot I}{u-\tu-\ka} ,
}
which coincides with the symmetry relation \eqref{symm}, as required.
\end{proof}

The ``compact'' symmetry relation \eqref{bsymm} is more convenient than \eqref{symm} in the context of the algebraic Bethe ansatz for the $\TX$-chain. This will become evident in Sections \ref{sec:AB-multi} and \ref{sec:nested}, where the so-called exchange relations are obtained.

Next, we focus on the lowest weight representations. We will rephrase some of the statements given in Section 4 of \cite{GRW1}, where the highest weight representation theory of $\TX$ was introduced.

\begin{defn}
A representation $V$ of $\TX$ is called a lowest weight representation if there exists a non-zero vector $\eta\in V$ such that $V=\TX\,\eta$ and
\equ{ 
\si_{ij}(u)\,\eta = 0 \qu\text{for}\qu 1 \le j < i \le {2n} \qu\text{and}\qu \si_{ii}(u)\,\eta = \mu_i(u)\,\eta \qu\text{for}\qu 1 \leq i \leq n,
}
where $\mu_i(u)$ are formal power series in $u^{-1}$ with the constant term equal $g_{ii}$. The vector $\eta$ is called the lowest weight vector of $V$, and the $n$-tuple $\mu(u)=(\mu_{1}(u),\ldots,\mu_{n}(u))$ is called the lowest weight of $V$. 
\end{defn}

Symmetry relation \eqref{symm} implies that $\eta$ is also an eigenvector for $\si_{ii}(u)$ with $n<i\le 2n$. 
Given an $n$-tuple $\mu(u)$, we will often make use of the corresponding $n$-tuple $\wt{\mu}(u)$ with components defined by (cf., eq.~(4.10) in \cite{GRW1})
\equ{
\wt{\mu}_i(u) := (2u+\rho-i+1)\,\mu_i(u) + \sum_{j=1}^{i-1} \mu_j(u) . \label{tm(u)}
}

Our focus will be on the lowest weight $\TX$-modules obtained by tensoring lowest weight $X(\mfg_{2n})$- and $\TX$-representations. With this goal in mind we need the following rephrase of Proposition 4.10 in \cite{GRW1}. 

\begin{prop} \label{P:L*V}
Let $\xi$ be the lowest vector of a lowest weight $X(\mfg_N)$-module $L(\la(u))$ and let $\eta$ be the lowest vector of a lowest weight $\TX$-module $V(\mu(u))$. Then $\TX(\xi\ot\eta)$ is a lowest weight $\TX$-module with the lowest vector $\xi\ot\eta$ and the lowest weight $\ga(u)$ with components determined by the relations
\equ{
\wt{\ga}_i(u) = \wt{\mu}_i(u)\,\la_i(u-\tfrac\ka2)\,\la_{2n-i+1}(\tu-\tfrac\ka2) \qu\text{for}\qu 1\le i \le n, \label{tga(u)}
}
with $\wt{\ga}_i(u)$ defined by \eqref{tm(u)}.
\end{prop}

\begin{proof}
The proof is very similar to that of Proposition 4.10 in \cite{GRW1} and is essentially the same as that of Proposition \ref{P:L*V-bnr} stated in Section \ref{sec:bnr} below; we refer the reader to the latter. 
\end{proof}

We will restrict to the cases when $V(\mu(u))$ is a one-dimensional representation of $\TX$. It will be interpreted as the boundary quantum space of the open spin chain. The Lemma below rephrases Lemma~2.3 in \cite{GRW1} and Lemma~5.4 in \cite{GRW2}.

\begin{lemma} \label{L:1-dim} 
Let $a,b\in\C$. Then the matrices 
\equ{
K(u) = G - \frac{a}{u+\frac\rho2}\, I \label{K-1}
}
when $n\ge1$ and $G$ is type CI, or $n\ge 2$ and G is of type DIII, and
\ali{
K(u) = -\left(1-\frac{b}{u+\tfrac\rho2}\right) \left( \left(1-\frac{a}{u+\tfrac\rho2}\right) e_{11} - \left(1+\frac{a}{u+\tfrac\rho2}\right) e_{22} \right)& \el  + \left(1+\frac{b}{u+\tfrac\rho2}\right) \left( \left(1-\frac{a}{u+\tfrac\rho2}\right) e_{33} - \left(1+\frac{a}{u+\tfrac\rho2}\right) e_{44} \right) &, \label{K-2}
}
when $n=2$, and
\equ{
K(u) = \frac{(u-a+\tfrac\rho2)(u+a-2d+\tfrac\rho2)}{(u-d+\tfrac\rho2)^2} \left(I - \frac{2u+\rho}{u-a+\frac\rho2} \, e_{11} - \frac{2u+\rho}{u+a-2d+\frac\rho2} \, e_{2n,2n} \right), \label{K-3}
}
 when $n>2$ and $d=\frac n2-1$, are one- or two-parameter solutions of \eqref{RE} satisfying the symmetry relation \eqref{symm} (with $\Si(u)$ replaced by $K(u)$).
\end{lemma}

The non-zero matrix elements of $K(u)$ in (\ref{K-1}-\ref{K-3}) are power series in $u^{-1}$ of the form $g_{ii} + u^{-1}\C[[u^{-1}]]$.  This implies the following statement.

\begin{prop} \label{P:1-dim}  
(i) The assignment $\Si(u)\mapsto K(u)$ yields a one-dimensional representation of $\TX$ of weight $\mu(u)$ given by, in the case-by-case way,
\begin{itemize}
\item for CI and DIII by \eqref{K-1}:
\equ{
\mu_1(u) = \ldots = \mu_n(u) = 1 - \frac{a}{u+\frac\rho2}, \label{m:CI-DIII}
}
\item for DI when $n=p=q=2$ by \eqref{K-2}:
\equ{
\mu_1(u) = \left(-1+\frac{a}{u+\tfrac\rho2}\right)\left(1-\frac{b}{u+\tfrac\rho2}\right), \quad \mu_2(u)=\left(1+\frac{a}{u+\tfrac\rho2}\right)\left(1-\frac{b}{u+\tfrac\rho2}\right), \label{m:DI-1}
}
\item for DI when $n>2$, $p=2n-2$, $q=2$ by \eqref{K-3}:
\equ{
\mu_1(u) = -\frac{(u+a+\tfrac\rho2)(u+a-2d+\tfrac\rho2)}{(u-d+\tfrac\rho2)^2}, \qu \mu_2(u) = \ldots = \mu_{n}(u) = \frac{(u-a+\tfrac\rho2)(u+a-2d+\tfrac\rho2)}{(u-d+\tfrac\rho2)^2} . \label{m:DI-2}
}
\end{itemize}
(ii) The assignment $\Si(u)\mapsto K(u)=G(u+\frac\rho2)$ with $G(u)$ defined by \eqref{G(u)} yields a one-dimensional representation of $\TX$ of weight $\mu(u)$ given, case-by-case, by
\begin{itemize}
\item for CII when $p\ge q$ and DI when $p\ge q\ge 4$:
\equ{
\mu_i(u) = \frac{d-(u+\tfrac\rho2)\,g_{ii}}{d-u-\tfrac\rho2} \qu \text{for} \qu  1 \leq i \leq n, \label{m:CII-DI}
}
\item for CD0:
\ali{
\mu_1(u) = \ldots = \mu_n(u) = 1. \label{m:CD0}
}
\end{itemize}
\end{prop}


\subsection{Block decomposition of $X(\mfg_{2n})$ and $X_\rho(\mfg_{2n},\mfg_{2n}^\rho)^{tw}$ } \label{sec:block}

In this section, inspired by the arguments presented in \cite{Rs1,DVK} (see also Section 2.3 in \cite{GMR}) we demonstrate a block decomposition of the Yangian $X(\mfg_{2n})$ and the twisted Yangian $\TX$.
We decompose the $2n\times 2n$ dimensional matrices $T(u)$ and $S(u)$ into $n\times n$ dimensional blocks as follows:
\equ{
T(u) = \left(\begin{array}{cc} \ol{A}(u) & \ol{B}(u) \\ \ol{C}(u) & \ol{D}(u) \end{array}\right) , \qq 
S(u) = \left(\begin{array}{cc} A(u) & B(u) \\ C(u) & D(u) \end{array}\right) . \label{block}
}
Our goal is to derive algebraic relations between these smaller matrix operators (blocks), which is the crucial first step of the algebraic Bethe ansatz for a $\TX$-chain. 
We will denote the matrix elements of $A(u)$ by $\sca_{ij}(u)$ with $1\le i,j\le n$, and similarly for matrices $B(u)$, $C(u)$ and $D(u)$, and their barred counterparts.

Recall that $\C^{2n} \cong \C^2 \ot \C^n$. Let ${\tt e}_{ij}$ with $1\le i,j\le 2n$ denote the standard matrix units of $\End(\C^{2n})$. Moreover, let $x_{ij}$ with $1\le i,j\le 2$ (resp.~$e_{ij}$ with $1\le i,j\le n$) denote the standard matrix units of $\End(\C^2)$ (resp.~$\End(\C^n)$).
Then, for any $1\le i,j \le n$, we may write
\equ{
{\tt e}_{ij} = x_{11} \ot e_{ij}, \qu {\tt e}_{n+i,j} = x_{21} \ot e_{ij}, \qu {\tt e}_{i,n+j} = x_{12} \ot e_{ij} ,\qu {\tt e}_{n+i,n+j} = x_{22} \ot e_{ij} .
\label{e=x*e}
}
Hence any matrix $M\in \End(\C^{2n})$ with entries $(M)_{ij}\in \C$ can be written as
\[
M = \sum_{a,b=1}^2 x_{ab} \ot [\![M]\!]_{ab} \in \End(\C^2) \ot \End(\C^n),
\]
where $[\![M]\!]_{ab} = \sum_{i,j=1}^{n} [M]_{i+n(a-1),j+n(b-1)}\, e_{ij}$ are blocks of $M$, viz.~\eqref{block}.
Now let $M \in \End(\C^{2n} \ot \C^{2n})$. Then we may write
\[
M = \sum_{a,b,c,d=1}^2 x_{ab} \ot x_{cd} \ot [\![M]\!]_{abcd} \in \End(\C^2\ot \C^2) \ot \End(\C^n \ot \C^n),
\]
where $[\![M]\!]_{abcd}$ are obtained as follows. Writing $M = \sum_{i,j,k,l=1}^{2n} [M]_{ijkl}\,{\tt e}_{ij} \ot {\tt e}_{kl}$ we have
\equ{
[\![M]\!]_{abcd} = \sum_{i,j,k,l=1}^{n} [M]_{i+n(a-1),j+n(b-1),k+n(c-1),l+n(d-1)}\, e_{ij} \ot e_{kl} . \label{M-block}
}

Denote the $R$-matrix \eqref{R(u)} acting on $\C^{2n}\ot \C^{2n}$ by $\wt R(u)$. Viewing $\wt R(u)$ as element in $\End(\C^2\ot \C^2) \ot \End(\C^n \ot \C^n)[[u^{-1}]]$ and using \eqref{M-block} we recover the familiar six-vertex block structure,
\eqa{
\wt R(u) = \left(\begin{array}{cccc}  \!R(u)\! & \\ & \!R^t(\ka-u)& U(u)\\ & U(u) & R^t(\ka-u) \\ &&&  \!R(u)\! \end{array}\right) . \label{R:new}
}
The operators inside the matrix above are each acting on $\C^n \ot \C^n$ and are given by
\equ{
 R(u) =  I - \frac{1}{u} \, P , \qq
 U(u) = -\,\frac{1}{u}\,  P \pm \frac{1}{u-\ka} \, Q , \label{R-U}
}
where both the transpose $t$ and the projector $Q=\sum_{i,j=1}^{N}e_{ij} \ot e_{\bar\jmath\,\bar\imath}$ are of the orthogonal type (recall the notation $\bar\imath=n-i+1$), and $I$ is the identity matrix. These operators satisfy the following unitarity relations
\equ{
 R(u)\,  R(-u) = (1-u^{-2})\,  I , \qq  R^t(u)\, R^t(n-u)=  I . \label{RK-unit}
} 
In a similar way, the matrices $T_1(u)$ and $T_2(u)$, as elements of $\End(\C^2\ot \C^2) \ot \End(\C^n \ot \C^n) \ot X(\mfg_{2n})[[u^{-1}]]$, take the form
\ali{
T_1(u)&= \left(\begin{array}{cccc} \ol{A}_1(u) & & \ol{B}_1(u) \\ & \ol{A}_1(u)&& \ol{B}_1(u)\\ \ol{C}_1(u) & & \ol{D}_1(u) \\ & \ol{C}_1(u)&& \ol{D}_1(u)\end{array}\right) , \qu
T_2(u) &= \left(\begin{array}{cccc} \ol{A}_2(u) & \ol{B}_2(u) \\ \ol{C}_2(u) & \ol{D}_2(u)\\ && \ol{A}_2(u) & \ol{B}_2(u) \\ && \ol{C}_2(u) & \ol{D}_2(u)\\ \end{array}\right) .  \label{T:new}
}
where $\ol{A}_1(u)$ means $\ol{A}(u) \ot I \in \End(\C^n \ot \C^n) \ot X(\mfg_{2n})[[u^{-1}]]$, and similarly for the other blocks.
Substituting \eqref{R:new} and \eqref{T:new} to \eqref{Y:RTT} allows us to rewrite the defining relations of $X(\mfg_{2n})$ in terms of the matrices $\ol{A}(u)$, $\ol{B}(u)$, $\ol{C}(u)$ and $\ol{D}(u)$.
The relations that we will need are:
\gat{
R(u-v)\, \ol{A}_1(u)\, \ol{A}_2(v) = \ol{A}_2(v)\, \ol{A}_1(u)\, R(u-v), \label{Y:AA}\\
R(u-v)\, \ol{D}_1(u)\, \ol{D}_2(v) = \ol{D}_2(v)\, \ol{D}_1(u)\, R(u-v), \label{Y:DD}\\
R^t(\ka-u+v)\,\ol{C}_1(u)\, \ol{A}_2(v) = \ol{A}_2(v)\, \ol{C}_1(u)\, R(u-v) + Q(u-v)\,\ol{A}_1(u)\,\ol{C}_2(v), \label{Y:CA}\\
\ol{C}_2(v)\, \ol{D}_1(u)\,R^t(\ka-u+v) = R(u-v)\, \ol{D}_1(u)\, \ol{C}_2(v) - \ol{D}_2(v)\,\ol{C}_1(u)\,K(u-v), \label{Y:CD}\\
\ol{A}_2(v)\,\ol{D}_1(u)\,R^t(\ka-u+v) - R^t(\ka-u+v)\,\ol{D}_1(u)\,\ol{A}_2(v) \hspace{3.1cm} \nn\\ \hspace{4.12cm} = Q(u-v)\,\ol{B}_1(u)\,\ol{C}_2(v) - \ol{B}_2(v)\,\ol{C}_1(u)\,Q(u-v) .\label{Y:AD}
}
In particular, the coefficients of the matrix entries of $\ol{A}(u)$ generate a $Y(\mfgl_n)$ subalgebra of $X(\mfg_{2n})$. The same is true for $\ol{D}(u)$. We will recall the necessary details of the Yangian $Y(\mfgl_n)$ in Section \ref{sec:Y} below.


We now repeat the same steps for the twisted Yangian $\TX$. We substitute \eqref{R:new} to \eqref{RE} and view the matrices $S_1(u)$ and $S_2(u)$ as elements of $\End(\C^2\ot \C^2) \ot \End(\C^n \ot \C^n) \ot \TX((u^{-1}))$, so that they take the same form as in \eqref{T:new}. 
This allows us to write the defining relations of $\TX$ in terms of the matrices $A(u)$, $B(u)$, $C(u)$ and $D(u)$. 
The relations that we will need are: 
\gat{
R(u-v)\,A_1(u)\,R(u+v+\rho)\,A_2(v) = A_2(v)\,R(u+v+\rho)\,A_1(u)\,R(u-v) \hspace{2.5cm} \el
\hspace{3.5cm} - R(u-v)\,B_1(u)\,Q(u+v+\rho)\,C_2(v) + B_2(v)\,Q(u+v+\rho)\,C_1(u)\,R(u-v) , \label{RE:AA}
\\
A_2(v)\, R(u+v+\rho)\, B_1(u)\, Q(u-v) = R(u-v)\, B_1(u)\, Q(u+v+\rho)\, A_2(v) \hspace{2.5cm} \el
\hspace{3.5cm} - B_2(v)\, Q(u+v+\rho)\, A_1(u)\, Q(u-v) - B_2(v)\, Q(u+v+\rho) \, D_1(u)\, Q(u-v) , \label{RE:AB}
\\
R^t(\tu-v)\,C_1(u)\,R(u+v+\rho)\,A_1(v) = A_1(v)\,R^t(\tu-v)\,C_1(u)\,R(u-v) \hspace{1cm} \el 
\hspace{1.75cm}- Q(u-v)\,A_1(u)\,R^t(\tu-v)\,C_2(v) - R^t(\tu-v) D_1(u)\,Q(u+v+\rho)\,C_1(v) , \label{RE:AC}
\\
R(u-v)\, B_1(u)\, R^t(\tu-v)\, B_2(v) = B_2(v)\, R^t(\tu-v)\, B_1(u)\, R(u-v) . \label{RE:BB}
}
It remains to cast the symmetry relation \eqref{bsymm} in the block form. Observe that 
\equ{
S^t(u) = \left(\begin{array}{cc} D^t(u) & \pm B^t(u) \\ \pm C^t(u) & A^t(u) \end{array}\right) ,
}
which allows us to immediately extract linear relations between the operators $A(u),B(u),C(u)$ and $D(u)$, of which we will need the following two only: 
\gat{ \label{symmAD}
D^t(u) = -\bigg(1\pm\frac1{u-\tu}\bigg)A(\tu) \pm \frac{A(u)}{u-\tu}
- \frac{\tr(S(u))\cdot I}{u-\tu-\ka} ,
\\
\label{symmBB} 
B^t(u) = \bigg(\mp1-\frac1{u-\tu}\bigg)B(\tu) + \frac{B(u)}{u-\tu}.
}

Let $V$ be a lowest weight finite-dimensional representation of $\TX$ and let $V^0\subset V$ be the subspace annihilated by the operator $C(u)$. Then operators $A(u)$ in the space $V^0$ satisfy the defining relations of the extended reflection algebra $\tBnr$, cf.\ \eqref{RE:AA}. We will recall the necessary details about the algebra $\tBnr$ in Section \ref{sec:bnr} further below.


\subsection{The Yangian $Y(\mfgl_{n})$} \label{sec:Y}

We now briefly recall necessary details of the Yangian $Y(\mfgl_{n})$ and its representation theory adhering closely to \cite{Mo3}. We will often use the superscript $^\circ$ to indicate operators associated with the algebra $Y(\mfgl_{n})$. This is to avoid the overlap of notation with operators associated to the algebra $X(\mfg_{2n})$ and having the same name.

We first recall the $R$-matrix $R(u)= I - u^{-1} P$ stated in \eqref{R-U} and called the {\it Yang's $R$-matrix}.
It is a unitarity solution of the quantum Yang-Baxter equation, i.e., \eqref{YBE}.
We then introduce elements $t_{ij}^{\circ(r)}$ with $1 \le i,j \le n$ and $r\ge 0$ such that $t^{\circ(0)}_{ij}= \del_{ij}$. Combining these into formal power series $t^\circ_{ij}(u) = \sum_{r\ge 0} t_{ij}^{\circ(r)} u^{-r}$, we can then form the generating matrix $T^\circ(u)= \sum_{i,j=1}^n e_{ij} \ot t^\circ_{ij}(u)$.

\begin{defn}
The Yangian $Y(\mfgl_{n})$ is the unital associative $\C$-algebra generated by elements $t_{ij}^{\circ(r)}$ with $1 \le i,j \le n$ and $r\ge 0$ satisfying \eqref{Y:RTT} with the $R$-matrix $R(u)=I-u^{-1} P$. The Hopf algebra structure of $Y(\mfgl_{n})$ is given by the same formulae as in \eqref{Hopf:Y}.
\end{defn}

Note that analogues of the (anti-)automorphisms \eqref{Aut1} and \eqref{Aut2} hold for the algebra $Y(\mfgl_n)$. The symmetry (cross-unitarity) relation is replaced with an identity for quantum determinant and quantum comatrix, which we will discuss a bit further below. We first recall the definition of the lowest weight representation of $Y(\mfgl_{n})$. 

\begin{defn}
A representation $V$ of $Y(\mfgl_{n})$ is called a lowest weight representation if there exists a non-zero vector $\eta\in V$ such that $V=Y(\mfgl_{n})\,\eta$ and
\[
t^\circ_{ij}(u)\,\xi = 0 \qu\text{for}\qu 1 \le j<i \le n 
\qu\text{and}\qu 
t^\circ_{ii}(u)\,\xi=\la^\circ_i(u)\,\xi \qu\text{for}\qu 1 \leq i \leq n,
\]
where $\la^\circ_i(u)$ is a formal power series in $u^{-1}$ with a constant term equal to $1$. The vector $\eta$ is called the lowest vector of $V$, and the $n$-tuple $\la^\circ(u)=(\la^\circ_1(u),\ldots,\la^\circ_{n}(u))$ is called the lowest weight of $V$. 
\end{defn}

Given a lowest weight representation $V$ of $Y(\mfgl_{n})$ and a lowest vector $\xi\in V$ the action of the inverse matrix $T^{\circ-1}(u)$ on $\xi$ is defined as follows. 
Introduce the \emph{quantum determinant} $\qdet T^\circ(u)$ of the matrix $T^\circ(u)$ by (see Definition 1.6.5 and Proposition 1.6.6 in \cite{Mo3})
\[
\qdet T^\circ(u) = \sum_{\si \in \mf{S}_n} \sgn(\si)\,t^\circ_{1\si(1)}(u-n+1)\cdots t^\circ_{n\si(n)}(u).
\]
In particular, $\qdet T^\circ(u)$ is a formal power series in $u^{-1}$ with coefficients central in $Y(\mfgl_n)$ and constant term~1. Define \emph{quantum comatrix} $\wh{T}^\circ(u)$ with matrix elements $\wh{t}^\circ_{ij}(u)$ by $\wh{T}^\circ(u)\,T^\circ(u-n+1) = \qdet T^\circ(u)$. Then the inverse matrix $T^{\circ-1}(u)$ with matrix elements $t^{\prime\circ}_{ij}(u)$ can be defined by
\[
t^{\prime\circ}_{ij}(u) = (\qdet T^\circ(u+n-1))^{-1}\cdot \wh{t}^\circ_{ij}(u+n-1) .
\]  
It follows from the definitions of $\qdet T^\circ(u)$, $\wh{T}^\circ(u)$ and $\xi$ that 
\[
t^{\prime\circ}_{ij}(u)\,\xi = 0\qu\text{for}\qu 1 \le j<i \le n \qu\text{and}\qu t^{\prime\circ}_{ii}(u)\,\xi=\la^{\prime\circ}_i(u)\,\xi \qu\text{for}\qu 1 \leq i \leq n
\]
with the ``inverse-weights'' $\la^{\prime\circ}_i(u)$ defined by
\equ{ 
\la^{\prime\circ}_i(u) = \frac{\la^\circ_1(u+1)\cdots \la^\circ_{i-1}(u+i-1)}{\la^\circ_1(u)\cdots \la^\circ_i(u+i-1)} . \label{la'(u)-gln} 
}
The latter expression follows from the fact that matrix elements $\wh{t}^\circ_{ij}(u)$ equal $(-1)^{i+j}$ times the quantum determinant of the submatrix of $T^\circ(u)$ obtained by removing the $i$th column and the $j$th row, and application of the quantum determinant to $\xi$, see Proposition 1.9.2 in \cite{Mo3}. (Also see the proof of Theorem 4.2 in \cite{MoRa}, where an analogue of formula \eqref{la'(u)-gln} for the highest weight was obtained.)

The Yangian $Y(\mfgl_{n})$ contains the universal enveloping algebra $U(\mfgl_{n})$ as a Hopf subalgebra. An embedding  $U(\mfgl_{n})\into Y(\mfgl_{n})$ is given by $E_{ij} \mapsto t_{ij}^{\circ(1)}$ for all $1\le i,j\le n$. We will identify $U(\mfgl_{n})$ with its image in $Y(\mfgl_{n})$ under this embedding. Conversely, the map $t_{ij}^{\circ(1)} \mapsto E_{ij}$ and $t_{ij}^{\circ(r)} \mapsto 0$ for all $r\ge2$ defines a surjective homomorphism $ev : Y(\mfgl_{n}) \to U(\mfgl_{n})$ called the {\it evaluation homomorphism}. We compose the map $ev$ with the anti-automorphisms $sign$ and $tran$ and the shift automorphism $\tau_c$. Denoting the resulting map by $\bm{ev}_c := ev \circ sign \circ tran \circ \tau_c$ we have
\equ{
\bm{ev}_c \;:\; t^\circ_{ij}(u) \mapsto \del_{ij} - E_{ji} (u-c)^{-1}  .  \label{ev-hom}
}

By the virtue of the map $\bm{ev}_c$, any $\mfgl_{n}$-representation can be regarded as $Y(\mfgl_{n})$-module. Moreover, any irreducible $\mfgl_{n}$-representation remains irreducible over $Y(\mfgl_{n})$, by surjectivity of $\bm{ev}_c$. We will denote by $L^\circ(\la^\circ)_c$ the $Y(\mfgl_{n})$-module obtained from the irreducible representation $L^\circ(\la^\circ)$ of $\mfgl_{n}$ via the map \eqref{ev-hom}. Clearly, $L^\circ(\la^\circ)_c$ is a lowest weight $Y(\mfgl_{n})$-module with the components of the lowest weight given by
\[
\la^\circ_i(u) = 1 - \frac{\la^\circ_i}{u-c} \qu\text{for}\qu 1\le i\le n .
\]

We will mostly be interested in the representations $L^\circ(\la^\circ)_c$ when $\la_i^\circ(u)$ for $1\le i \le n$ coincide with those in \eqref{l(u)-fused}. 
Formula \eqref{la'(u)-gln} implies that the ``inverse-weights'' of the lowest vector of $L^{\circ}(\la^\circ)_c$ are given by
\equ{
\la^{\prime\circ}_i(u) = \frac{u-c}{u-c\mp k}\cdot \frac{u-c\mp k\pm1}{u-c\pm1} \bigg( 1 + \frac{\la_i}{u-c\mp k\pm1}\bigg). \label{l'(u)-fused}
}


\subsection{Reflection algebra $\tBnr$} \label{sec:bnr}

We now focus on the extended $\rho$-shifted Molev-Ragoucy reflection algebra $\tBnr$ and its lowest weight representation theory, adhering closely to \cite{MoRa}. (We use notation $\mc{B}^{\,\rm ex}$ instead of $\wt{\mcB}$ used in {\it loc.\ cit.} to avoid overuse of the tilded notation.) We will need to prove some additional statements that are necessary for the algebraic Bethe ansatz along the way. We start with introducing the non-extended reflection algebra $\Bnr$.

\begin{defn} 
The reflection algebra $\Bnr$ is the subalgebra of $Y(\mfgl_{n})$ generated by the coefficients of the entries of the matrix 
\[
B^\circ(u) = T^\circ(u)\,G^\circ T^{\circ-1}(-u-\rho) \qq\text{where}\qq G^\circ = \sum_{i=r+1}^{n} e_{ii} - \sum_{i=1}^{r} e_{ii}. 
\]
\end{defn}

The reflection algebra defined above is isomorphic to the usual one studied in \cite{MoRa}. The isomorphism is provided by the map $\Bnr \to \mc{B}(n,n-r)$, $B^\circ(u) \mapsto -B(u-\tfrac\rho2)$.
The matrix $B^\circ(u)$ satisfies the reflection equation \eqref{RE} with the $R$-matrix $R(u)=I-u^{-1} P$ and the unitarity relation
\equ{
B^\circ(u)\,B^\circ(-u-\rho) = I . \label{unit}
}
The reflection equation and the unitarity relation are in fact the defining relations of $\Bnr$. Their form in terms of matrix elements $b^\circ_{ij}(u)$ of $B^\circ(u)$, for $\rho=0$, are given by formulas (2.7) and (2.8) in \cite{MoRa}. We will recall them in a suitable form in Section \ref{sec:Bnr-rels}.

We now turn to the extended reflection algebra $\tBnr$. We will use the same notation  $B^\circ(u)$ to denote the generating matrix of $\tBnr$.

\begin{defn} 
The extended reflection algebra $\tBnr$ is the unital associative algebra generated by the coefficients of the entries of the abstract generating matrix $B^\circ(u)$ satisfying the reflection equation and its constant part being equal to the matrix $G^\circ$, that is $b^\circ_{ij}(u) = g^\circ_{ij} + \sum_{r\ge1} b^{\circ(r)}_{ij} u^{-r}$.
\end{defn}

By the same arguments as in Proposition 2.1 in \cite{MoRa}, the product $B^\circ(u)\,B^\circ(-u-\rho)$ is a matrix
\equ{
B^\circ(u)\,B^\circ(-u-\rho) =  f^\circ (u) \, I, \label{unit-f}
}
where $f^\circ(u)$ is an even series in $u^{-1}$ with coefficients central in $\tBnr$. In fact, the algebra $\Bnr$ may be viewed as the quotient of $\tBnr$ by the two-sided ideal generated by the unitarity relation \eqref{unit}. It is important to note that the algebra $\tBnr$ has the same coalgebra structure as $\Bnr$,
\[
\Delta(b^\circ_{ij}(u)) = \sum_{a,b=1}^n t^\circ_{ia}(u)\,t^{\prime\circ}_{bj}(-u-\rho) \ot b^\circ_{ab}(u).
\]

\begin{defn}
A representation $V$ of $\tBnr$ is called a lowest weight representation if there exists a non-zero vector $\xi\in V$ such that $V=\tBnr\,\xi$ and
\[
b^\circ_{ji}(u)\,\xi = 0 \qu\text{for}\qu 1 \le i < j \le n \qu\text{and}\qu b^\circ_{ii}(u)\,\xi = \mu^\circ_i(u)\,\xi \qu\text{for}\qu 1 \leq i \leq n,
\]
where $\mu_i(u)$ are formal power series in $u^{-1}$ with constant terms equal to $1$ if $i\le n-r$ and $-1$ if $i>n-r$. The vector $\xi$ is called the lowest vector of $V$, and the $n$-tuple $\mu^\circ(u)=(\mu^\circ_{1}(u),\ldots,\mu^\circ_{n}(u))$ is called the lowest weight of $V$. 
\end{defn}

We note that any representation of $V$ of $\Bnr$ may be extended to a representation of $\tBnr$ by allowing the series $f^\circ(u)$ to act as the identity operator on $V$.


The Proposition below is an analogue of Proposition \ref{P:L*V} for the algebra $\tBnr$. 

\begin{prop} \label{P:L*V-bnr}
Let $\xi$ be the lowest vector of a lowest weight $Y(\mfgl_n)$-module $L(\la(u))$ and let $\eta$ be the lowest vector of a lowest weight $\tBnr$-module $V(\mu(u))$. Then $\tBnr(\xi\ot\eta)$ is a lowest weight $\tBnr$-module with the lowest vector $\xi\ot\eta$ and the lowest weight $\ga^\circ(u)$ with components determined by the relations
\equ{
\wt{\ga}^\circ_i(u) = \wt{\mu}^\circ_i(u)\,\la^\circ_i(u)\,\la^{\prime\circ}_{i}(-u-\rho) \qu\text{for}\qu 1\le i \le n \label{tga(u)-bnr}
}
with $\wt{\ga}^\circ_i(u)$ and $\wt{\mu}^\circ_i(u)$ defined by \eqref{tm(u)}.
\end{prop}

\begin{proof}
The proof is very similar to that of Proposition 4.10 in \cite{GRW1}. We will use the symbol ``$\equiv$'' to denote equality of operators on the spaces $\C(\xi\otimes \eta)$ or $\C\xi$.
We begin by showing that $b^\circ_{ij}(u)\cdot (\xi\otimes \eta)=0$ for all $i>j$. 
We have 
\[
\Delta(b^\circ_{ij}(u)) \equiv \sum_{1\le a\le b \le n} t^\circ_{ia}(u)\,t^{\prime\circ}_{bj}(-u-\rho)\otimes b^\circ_{ab}(u).
\]
Since $t^{\prime\circ}_{bj}(-u-\rho)\,\xi=0$ if $b>j$, we can assume $b\leq j$ implying $a\le b \le j<i$ and $t^\circ_{ia}(u)\,\xi=0$. 
The defining relations (4.33) in \cite{MoRa}
\equ{
[t^\circ_{ia}(u),t^{\prime\circ}_{bj}(v)] = \frac{1}{u-v}  \sum_{k=1}^n \big( \del_{ab}\, t^\circ_{ik}(u)\,t^{\prime\circ}_{kj}(v) - \del_{ij}\,t^{\prime\circ}_{bk}(v)\,t^\circ_{ka}(u)\big) \label{L*V-bnr-1}
}
further imply $t^\circ_{ia}(u)\,t^{\prime\circ}_{bj}(v)\equiv 0$ unless $a=b$. 
Hence it suffices to show that $t^\circ_{ia}(u)\,t^{\prime\circ}_{aj}(v)\equiv 0$ for $i>j$ and $a\leq j$. By \eqref{L*V-bnr-1} we have, for all $a\le j$,
\[
t^\circ_{ia}(u)\,t^{\prime\circ}_{aj}(v)\equiv \frac{1}{u-v} \sum_{b=1}^j t^\circ_{ib}(u)\,t^{\prime\circ}_{bj}(v)
\]
Summing both sides over $1\le a \le j$ we obtain
\[
\sum_{a=1}^j t^\circ_{ia}(u)\,t^{\prime\circ}_{aj}(v)\equiv \frac{j}{u-v}\sum_{a=1}^j t^\circ_{ia}(u)\,t^{\prime\circ}_{aj}(v), 
\]
implying $t^\circ_{ia}(u)\,t^{\prime\circ}_{aj}(v)\equiv 0$. This proves that $\Delta(b^\circ_{ij}(u))(\xi\otimes \eta)=0$ for all $i<j$. 


It remains to compute $\Delta(b^\circ_{ii}(u))(\xi\otimes \eta)$ for all $1\le i \le n$. By similar arguments as above we deduce that $t_{ia}(u)\,t'_{bi}(v)\,\xi = 0$ whenever $a<b$. Therefore, 
\equ{
\Delta(b^\circ_{ii}(u))(\xi\otimes \eta)= \sum_{a=1}^i t^\circ_{ia}(u)\,t^{\prime\circ}_{ai}(-u-\rho)\,\xi\ot b^\circ_{aa}(u)\,\eta = (\hat{b}^\circ_{ii}(u)\,\xi)\otimes \eta, \label{L*V-bnr-2}
}
where $\hat{b}^\circ_{ii}(u)$ is the operator defined by the formula 
\[
\hat{b}^\circ_{ii}(u)=\sum_{a=1}^i \mu^\circ_a(u)\,t^\circ_{ia}(u)\,t^{\prime\circ}_{ai}(-u-\rho).
\]
Define the operator $A_{i}(u)=\sum_{a=1}^i t^\circ_{ia}(u)\,t^{\prime\circ}_{ai}(-u-\rho)$.
We first show that $A_{i}(u)\,\xi=\mu_i^\bullet(u)\,\xi$ for some formal series $\mu_i^\bullet(u)\in\C[[u^{-1}]]$. From \eqref{L*V-bnr-1} we obtain that, for all $a<i$, 
\equ{
t^\circ_{ia}(u)\,t^{\prime\circ}_{ai}(-u-\rho) \equiv \frac{1}{2u+\rho}\left(\sum_{k=1}^i t^\circ_{ik}(u)\, t^{\prime\circ}_{ki}(-u-\rho) - \sum_{k=1}^a t^{\prime\circ}_{ak}(-u-\rho)\,t^\circ_{ka}(u) \right) . \label{L*V-bnr-4}
}
Summing both sides over $1\le a <i$ we obtain
\[
 A_i(u)\equiv t^\circ_{ii}(u)\,t^{\prime\circ}_{ii}(-u-\rho) + \frac{i-1}{2u+\rho}A_i(u)-\frac{1}{2u+\rho}\sum_{a=1}^{i-1} B_a(u),
\]
where $B_a(u)= \sum_{k=1}^a t^{\prime\circ}_{ak}(-u-\rho)\,t_{ka}(u)$. We have thus shown that
\equ{
 \frac{2u+\rho-i+1}{2u+\rho}\, A_i(u) \equiv \la^\circ_{i}(u)\,\la^{\prime\circ}_{i}(-u-\rho) - \frac{1}{2u+\rho}\sum_{a=1}^{i-1} B_a(u). \label{L*V-bnr-5}
}
In a similar way we find that, for $1\le i \le n$,
\[
 \frac{2u+\rho-i+1}{2u+\rho}B_i(u)\equiv \la^\circ_{i}(u)\,\la^{\prime\circ}_{i}(-u-\rho) - \frac{1}{2u+\rho} \sum_{a=1}^{i-1} A_{a}(u) . 
\]
A simple induction on $i$ shows that $B_i(u)\equiv A_i(u)$ for all $1\le i\le n$. This allows us to rewrite \eqref{L*V-bnr-5} as
\equ{
\frac{2u+\rho-i+1}{2u+\rho}A_i(u)\equiv \la^\circ_{i}(u)\,\la^{\prime\circ}_{i}(-u-\rho) - \frac{1}{2u+\rho} \sum_{a=1}^{i-1} A_{a}(u). \label{L*V-bnr-7}                      
}
Recall the notation \eqref{tm(u)}. Using induction on $i$ once again we deduce that $A_{i}(u)\,\xi=\mu_i^\bullet(u)\,\xi$ with the series $\mu^\bullet_i(u)$ determined by 
\equ{
\wt\mu_i^\bullet(u)=(2u+\rho)\,\la^\circ_i(u)\,\la^{\prime\circ}_{i}(-u-\rho) , \label{L*V-bnr-8}
}
Since $B_i(u)\equiv A_{i}(u) \equiv \mu^\bullet_i(u)$, we may rewrite \eqref{L*V-bnr-4} as $t^\circ_{ia}(u)\,t^{\prime\circ}_{ai}(-u-\rho)\equiv \frac{1}{2u+\rho}\left(\mu^\bullet_{i}(u)-\mu^\bullet_{a}(u) \right)$ yielding the identity
\[
\hat{b}^\circ_{ii} \equiv \mu^\circ_{i}(u)\,\la^\circ_i(u)\,\la^{\prime\circ}_{i}(-u-\rho) + \frac{1}{2u+\rho}\sum_{a=1}^{i-1} \mu^\circ_a(u)\left(\mu_i^\bullet(u)-\mu_a^\bullet(u) \right) .
\]
We have thus shown that $\Delta(b^\circ_{ii}(u))(\xi\otimes \eta) = \ga^\circ_i(u)\,(\xi\otimes \eta)$ with the series $\ga^\circ_i(u)$ given by the r.h.s.~above. 
Next, using $\mu_i^\bullet(u)=\mfrac{1}{2u+\rho-i+1}\Big(\wt \mu_i^\bullet(u)-\sum_{b=1}^{i-1}\mu_b^\bullet(u) \Big)$ and the above expression for $\wt\mu^\bullet_i(u)$ we rewrite the series $\ga^\circ_i(u)$ as
\ali{
 (2u+\rho-i+1)\,\ga^\circ_i(u) &= \frac{2u+\rho-i+1}{2u+\rho}\,\mu^\circ_i(u)\,\wt\mu_i^\bullet(u) + \frac{1}{2u+\rho} \sum_{a=1}^{i-1} \mu^\circ_a(u)\,\wt\mu_i^\bullet(u)\el
                    &-\frac{1}{2u+\rho} \sum_{a,b=1}^{i-1} \mu^\circ_a(u)\,\mu_b^\bullet(u)-\frac{2u+\rho-i+1}{2u+\rho} \sum_{a=1}^{i-1} \mu^\circ_a(u)\,\mu_a^\bullet(u) . \label{L*V-bnr-9}
}
Induction on $i$ then shows that
\[
\sum_{a=1}^{i-1} \ga^\circ_a(u)=\frac{1}{2u+\rho}\sum_{a,b=1}^{i-1} \mu^\circ_a(u)\,\mu_b^\bullet(u)+\frac{2u+\rho-i+1}{2u+\rho}\sum_{a=1}^{i-1}\mu^\circ_a(u)\,\mu_a^\bullet(u). 
\]
By combining this with \eqref{L*V-bnr-9} and \eqref{L*V-bnr-8} we obtain
\[
(2u+\rho-i+1)\,\ga^\circ_i(u) + \sum_{a=1}^{i-1} \ga^\circ_a(u) = \left( (2u+\rho-i+1)\,\mu^\circ_i(u) + \sum_{a=1}^{i-1} \mu^\circ_a(u) \right) \la^\circ_i(u)\,\la^{\prime\circ}_{i}(-u-\rho)
\]
which, by \eqref{tm(u)}, coincides with \eqref{tga(u)-bnr}.
\end{proof}

\begin{rmk}
The components $\ga^\circ_i(u)$ of the lowest weight $\ga^\circ(u)$ in the explicit form are given by the formulas
\gat{
\ga^\circ_i(u) = \mu^\sharp_i(u)\,\la^\circ_i(u)\,\la^{\prime\circ}_i(-u-\rho) - \sum_{j=1}^{i-1}  \frac{\mu^\sharp_j(u)\,\la^\circ_j(u)\,\la^{\prime\circ}_j(-u-\rho)}{2u+\rho-j} , \label{ga(u)-bnr} \\
\mu^\sharp_i(u) = \mu^\circ_i(u) + \sum_{j=1}^{i-1} \frac{\mu^\circ_j(u)}{2u+\rho-i+1} \label{sma(u)-bnr} .
}
For $\rho=0$ these agree with formulas (4.38) and (4.39) in \cite{BeRa1}.
\end{rmk}

\begin{prop} \label{P:rest-bnr}
Let $\mc{M}$ be a lowest weight $\tBnr$-module. For any $1\le k \le n-1$ define a subspace $\mc{M}^{(k)}\subseteq \mc{M}$ by
\[
\mc{M}^{(k)} := \{ v \in \mc{M} \,:\, b^\circ_{ij}(u)\, v = 0 \;\text{for}\;i>j\;\text{and}\;j<k \} .
\]
Then operators
\equ{
b^{(k)}_{ij}(u) := b^\circ_{ij}\big(u+\tfrac{k-1}{2}\big) + \del_{ij} \sum_{l=1}^{k-1} \frac{b^\circ_{ll}\big(u+\tfrac{k-1}{2}\big)}{2u+\rho} , \label{Bnr-b-rest}
}
where $k\le i,j\le n$ form a representation of the algebra $\wt{\mc{B}}(n-k+1,r-k+1)$ or $\wt{\mc{B}}(n-k+1,0)$ in $V^{(k)}$ for $r>k-1$ or $r\le k-1$, respectively.
\end{prop}

\begin{proof}
The $k=1$ case is trivial. The $k=2$ case follows by the same arguments presented in the proof of Theorem~4.6 in \cite{MoRa}, yielding $b^{(2)}_{ij}(u) = b^\circ_{ij}\big(u+\tfrac{1}{2}\big) + \del_{ij}\,\frac{1}{2u+\rho}\,b^\circ_{11}\big(u+\tfrac{1}{2}\big)$. The $k\ge 3$ case then follows by a simple induction. 
\end{proof}

The Proposition above in fact rephrases Theorem 3.1 in \cite{BeRa1} for the algebra $\Bnr$.

\begin{rmk}
We note the reader that an analogue of Proposition \ref{P:rest-bnr} for the ``non-extended'' reflection algebra $\Bnr$ would require operators $b^{(k)}_{ij}(u)$ in \eqref{Bnr-b-rest} to be multiplied by a suitable series in $u^{-1}$ with coefficients central in $\Bnr$ to ensure that the corresponding generating matrix $B^{\circ(k)}(u)$ satisfies the unitarity relation in the space $V^{(k)}$.
\end{rmk}

For any $a\in\C$ define a matrix-valued rational function
\equ{
K^{\circ}(u)=
G^\circ-\frac{a}{u+\frac\rho2} I. \label{K-1-dim-Bnr}
}
It is a one-parameter solution of the reflection equation \eqref{RE} with the $R$-matrix given by \eqref{R-U}.
We thus have the following. 

\begin{prop} \label{P:1-dim-Bnr}
(i) Let $r=0$. The assignment $B^\circ(u) \mapsto I$ yields a one-dimensional representation of $\mcB_\rho(n,0)$ of weight 
\equ{
\mu^\circ_{1}(u) = \ldots = \mu^\circ_n(u) = 1. \label{Bnr-M-0}
}
(ii) Let $1\le r\le n$. The assignment $B^\circ(u) \mapsto K^\circ(u)$ yields a one-dimensional representation of $\Bnr$ of weight $\mu^\circ(u)$ given by
\equ{
\mu^\circ_{1}(u) = \ldots = \mu^\circ_r(u) = - 1 - \frac{a}{u+\tfrac{\rho}{2}}, \qq
\mu^\circ_{r+1}(u) = \ldots = \mu^\circ_{n}(u) = 1 - \frac{a}{u+\tfrac{\rho}{2}} . \label{Bnr-M-k}
}
\end{prop}


\section{Algebraic Bethe ansatz for a $\tBnr$-chain} \label{sec:NABA-RE}

This section provides the necessary prerequisites to our main results presented in Section \ref{sec:NABA-X}. Here we study a spectral problem in the space: 
\[
M^\circ = L^\circ \ot V(\mu^\circ) = L^\circ(\la^{(1)})_{c_1} \ot \cdots \ot L^\circ(\la^{(\ell)})_{c_\ell} \ot V(\mu^\circ) , 
\]
where $L^\circ(\la^{(i)})_{c_i}$ is an arbitrary lowest weight evaluation $Y(\mfgl_n)$-module (irreducible and finite-dimen\-sional) and $V(\mu^\circ)$ is a one-dimensional $\tBnr$-module described by Proposition \ref{P:1-dim-Bnr}. In particular, the space $M^\circ$ is a lowest weight $\tBnr$-module of weight $\ga^\circ(u)$ with components $\ga^\circ_i(u)$ determined by (recall \eqref{tm(u)}, \eqref{la'(u)-gln} and Proposition \ref{P:L*V-bnr})
\equ{
\wt\ga^\circ_i(u) = \wt\mu^\circ_i(u) \prod_{j=1}^\ell \la_i^{(j)}(u)\,\la_i^{\prime(j)}(u) \qu\text{with}\qu \la_i^{(j)}(u) = \frac{u-c_j-\la_i^{(j)}}{u-c_j} \label{Bnr-gai}
}
and $\mu^\circ_i(u)$ given by \eqref{Bnr-M-0} and \eqref{Bnr-M-k}. We say that $M^\circ$ is a (full) quantum space of a $\tBnr$-chain, a $\mfgl_{n}$-symmetric open spin chain with (trivial left and non-trivial right) diagonal boundary conditions.
The spectral problem for such a chain was first addressed by Belliard and Ragoucy in \cite{BeRa1}, thus we will keep this section concise and provide the key steps in the proofs only. 

The main result of this section is Theorem \ref{T:Bnr-spec} stating eigenvectors, their eigenvalues and Bethe equations for a $\tBnr$-chain with the quantum space $M^\circ$. 
This provides a necessary step in solving the spectral problem for a $\TX$-chain in Section \ref{sec:NABA-SP} (in the symplectic case) and Section \ref{sec:NABA-SO} (in the orthogonal case). We note the reader that Theorem \ref{T:Bnr-spec} may be viewed as a special case of the results presented in Section 6 of \cite{BeRa1}. 

We also provide a trace formula for Bethe vectors. This formula may be viewed as a special case of the supertrace formula given by Theorem 7.1 in \cite{BeRa1}. This is the second main result of this section. We note the reader that only an outline of the proof of Theorem 7.1 in \cite{BeRa1} was given; here we provide a detailed proof of the trace formula under consideration.


\subsection{Exchange relations} \label{sec:Bnr-rels}

For any matrix $A= \sum_{i,j=1}^n a_{ij} e_{ij}$ with $e_{ij}\in \End(\C^n)$ and any $1\le k\le n$ define a $k$-reduced matrix $A^{(k)} = \sum_{i,j=k}^n a_{ij} e^{(k)}_{i-k+1,j-k+1}$ with $e^{(k)}_{ij}\in\End(\C^{n-k+1})$. 
We use this notation to define $k,l$-reduced $R$- and $\check{R}$-matrices acting on the spaces $V^{(k)}_a \cong \C^{n-k+1}$ and $V^{(l)}_b \cong \C^{n-l+1}$ by
\[
R^{(k,l)}_{ab}(u) := \frac{u}{u-1} \bigg( I_{ab}^{(k,l)} - \frac1u P_{ab}^{(k,l)} \bigg) , \qq \check{R}_{ab}^{(k,l)}(u) := P_{ab}^{(k,l)} R_{ab}^{(k,l)}(u).
\]
Note that $P^{(k,l)} e^{(k)}_i \ot e^{(l)}_j = 0$ if $k<l$ and $i+k-l\le 0$, and $R^{(n,n)}_{ab}(u)$ and $\check{R}^{(n,n)}_{ab}(u)$ are identity operators.
We denote the $k$-reduced generating matrix of $\tBnr$ in $\End( V^{(k)}_a)$ as $D^{(k)}_a(u)$ and decompose it as
\equ{ \def\arraystretch{1.25}
D_{a}^{(k)}(u) = \begin{pmatrix} a^{(k)}(u) & B_{a}^{(k)}(u) \\ C_{a}^{(k)}(u) & D_{a}^{(k+1)}(u) \end{pmatrix} . \label{Bnr-D}
}
We also set
\gat{
\hat{D}_a^{(k)}(u) := D_a^{(k)}\big(u+\tfrac{k-1}2\big) + \sum_{i=1}^{k-1} \frac{a^{(i)}\big(u+\tfrac{k-1}2\big)}{2u+\rho} \, I_a^{(k)} , \label{D-hat} \\
\hat a^{(k)}(u) := a^{(k)}\big(u+\tfrac{k}2\big) + \sum_{i=1}^{k-1} \frac{a^{(i)}\big(u+\tfrac{k}2\big)}{2u+1+\rho} , \qq \hat B_a^{(k)}(u) := B_a^{(k)}\big(u+\tfrac{k}2\big)  \label{a-hat-B-hat}
}
leading to the following recursive relations:
\gat{
\big[\hat{D}_a^{(k)}(u) \big]_{ij} = \big[ \hat D^{(k-1)}\big(u+\tfrac12\big) \big]_{1+i,1+j} + \frac{\del_{ij}}{2u+\rho} \big[ \hat D^{(k-1)}\big(u+\tfrac12\big) \big]_{11} \,, \label{D-hat-D-hat} \\
\hat a^{(k)}(u) = \big[\hat{D}_a^{(k)}\big(u+\tfrac12\big) \big]_{11} = \big[\hat{D}_a^{(k-1)}(u+1) \big]_{22} 
+ \frac{1}{2u+1+\rho}  \,\hat a^{(k-1)}\big(u+\tfrac12\big) , \label{a-hat-D-hat}
\\[.25em]
\big[\hat{D}_a^{(k)}(u) \big]_{1,1+l} = \big[ \hat B^{(k)}\big(u-\tfrac12\big) \big]_l  \label{D-hat-B-hat}
}
for $1\le i,j\le n-k+1$ and $1\le l \le n-k$. We note that operator $\hat{D}_a^{(k)}(u)$ is a generalisation of Sklyanin's $\wt D(u)$ operator (see Section 5 in \cite{Sk}) for arbitrary rank.

\begin{lemma} \label{L:Ex-Bnr}
Let $\mc{M}$ be a lowest weight $\tBnr$-module. For any $1\le k \le n-1$ define a subspace 
\equ{
\mc{M}^{(k)} := \{ \xi \in \mc{M} \,:\, b^\circ_{ij}(u)\, \xi = 0 \;\text{for}\;i>j\;\text{and}\;j< k \} . \label{Bnr-Mk}
}
Let $\equiv$ denote equality of operators in the space $V^{(k)}_a \ot V^{(k)}_b \ot \mc{M}^{(k)}$. Then
\ali{
\hat B^{(k)}_{a}(v)\,\hat B^{(k)}_{b}(u) &\equiv \hat B^{(k)}_{a}(u)\,\hat B^{(k)}_{b}(v)\, \check{R}^{(k+1,k+1)}_{ab}(v-u), \label{Bnr-BB} \\[.3em]
\hat a^{(k)}(v)\,\hat B^{(k)}_{b}(u) &\equiv \frac{(v-u+1)(v+u+1+\rho)}{(v-u)(v+u+\rho)}\,\hat B^{(k)}_{b}(u)\,\hat a^{(k)}(v) - \frac{2u+1+\rho}{(v-u)(2u+\rho)}\,\hat B^{(k)}_{b}(v)\,\hat a^{(k)}(u) \el & \qu + \frac{1}{v+u+\rho} \hat B^{(k)}_b(v)\, \hat D^{(k+1)}_b(u) , \label{Bnr-AB} \\[.3em]
\hat{D}^{(k+1)}_{a}(v)\,\hat{B}^{(k)}_{b}(u) &\equiv \frac{(v-u-1)(v+u-1+\rho)}{(v-u)(v+u+\rho)}\, \hat{B}^{(k)}_b(u)\,R_{ab}^{(k+1,k+1)}(v+u+\rho)\,\hat{D}^{(k+1)}_{a}(v)\,R^{(k+1,k+1)}_{ab}(v-u)\!\!\!\! \el &\qu - \frac{(2v-1+\rho)(2u+1+\rho)}{(2 v+\rho)(2 u+\rho)(v+u+\rho)}\,\hat{B}^{(k)}_b(v)\,R_{ab}^{(k+1,k+1)}(2v+\rho)\,P^{(k+1,k+1)}_{ab}\,\hat{\mr{a}}^{(k)}(u) \el
& \qu + \frac{2v-1+\rho}{(v-u)(2v+\rho)} \,\hat{B}^{(k)}_b(v)\,R_{ab}^{(k+1,k+1)}(2v+\rho)\,\hat{D}^{(k+1)}_{a}(u)\,P^{(k+1,k+1)}_{ab} , \label{Bnr-DB} \\[.3em]
\hat{a}^{(k)}(v)\,\hat{D}^{(k+1)}_{b}(u) &\equiv \hat{D}^{(k+1)}_{b}(u)\,\hat{a}^{(k)}(v) + \frac{1}{v-u} \tr_a P^{(k+1,k+1)}_{ab} \Big( \hat{B}^{(k)}_{b}(u)\,\hat{C}^{(k)}_{a}(v) - \hat{B}^{(k)}_{a}(u)\,\hat{C}^{(k)}_{b}(u) \Big) \el &\qu + \frac{1}{(v-u)(2 u-1+\rho)} \tr_b\Big(\hat{B}^{(k)}_{b}(v)\,\hat{C}^{(k)}_{b}(u)-\hat{B}^{(k)}_{b}(u)\,\hat{C}^{(k)}_{b}(v)\Big) \cdot I^{(k+1,k+1)}_b . \label{Bnr-AD} 
}
\end{lemma}

\begin{proof}
The $k=1$ case ($\mc{M}^{(1)} = \mc{M}$) is a restatement of the defining relations of $\tBnr$. When $k>1$ we additionally need to use Proposition \ref{P:rest-bnr}. 
\end{proof}


\subsection{Quantum spaces and monodromy matrices} \label{sec:Bnr-sp-mono}

Choose $m_1,\dots,m_{n-1}\in\Z_{\ge0}$, which we call excitation numbers. Let $k=1,\dots,n-1$. For each $m_k$, assign an $m_k$-tuple $\bm u^{(k)} = (u^{(k)}_1,\dots,u^{(k)}_{m_k})$ of complex parameters and a set of labels $\bm a^{k} = \{ a^k_1, \dots, a^k_{m_k} \}$. We will use notation from \cite{BeRa1} to denote multi-tuples: 
\equ{
\bm u^{(1\dots k)} := (\bm u^{(1)},\dots,\bm u^{(k)}), \qq 
\bm a^{1\dots k} := (\bm a^{1},\dots,\bm a^{k}) . \label{multi}
}
We will say that $M^\circ$ is a \emph{level-$1$ quantum space} and denote it by $M^{(1)}$. Then for each $2\le k \le n$ we define a \emph{level-$k$ quantum space} $M^{(k)}$ recursively by
\equ{
M^{(k)} := W^{(k)}_{\bm a^{k-1}} \ot (M^{(k-1)})^0 \label{Bnr-k-space}
}
where
\[
W^{(k)}_{\bm a^{k-1}} := V^{(k)}_{a^{k-1}_1} \ot \cdots \ot V^{(k)}_{a^{k-1}_{m_{k-1}}} 
\]
and $(M^{(k-1)})^0$ is level-$(k-1)$ \emph{vacuum sector} defined by
\equ{
(M^{(k-1)})^0 := \{ \xi \in M^{(k-1)} \,:\, b^\circ_{ij}(u)\, \xi = 0 \;\text{for}\;i>j\;\text{and}\;j<k-1 \} . \label{Bnr-vacuum}
}
Propositions \ref{P:L*V-bnr} and \ref{P:rest-bnr} imply that the space $M^{(k)}$ is a $\wt{\mcB}_\rho(n-k+1,r-k+1)$- or $\wt{\mcB}_\rho(n-k+1,0)$-module for $k<r+1$ or $k\ge r+1$, respectively. In particular,
\[
M^{(k)} = W^{(k)}_{\bm a^{k-1}} \ot (\C e^{(k-1)}_1 )^{\ot m_{k-2}} \ot \cdots \ot (\C e^{(2)}_1 )^{\ot m_{1}} \ot L^{(k)}(\la^{(1)})_{c_1} \ot \cdots \ot L^{(k)}(\la^{(\ell)})_{c_\ell} \ot V(\mu^\circ) 
\]
where
\[
L^{(k)}(\la^{(i)})_{c_i} := \{ \xi \in L^\circ(\la^{(i)})_{c_i} : t^\circ_{ij}(u)\,\xi = 0\;\text{for}\;i>j\;\text{and}\;j<k \}
\]
are evaluation $Y(\mfgl_{n-k+1})$-modules. (In the case when $L^\circ(\la^{(i)}) \cong \C^n$, i.e.\ the bulk quantum space is a tensor product of fundamental $\mfgl_n$-modules, $L^{(k)}(\la^{(i)})_{c_i} \cong \C$.)

\begin{defn} 
We will say that $\hat{D}^{(1)}_a(v):=D^{(1)}_a(v)$ is a level-1 monodromy matrix.
For each $2\le k \le n$ we recursively define a level-$k$ monodromy matrix, acting on the space $M^{(k)}$, via
\ali{ 
\hat{D}^{(k)}_{a {\bm a}^{1\dots k-1}}\big(v;\bm u^{(1\dots k-1)}\big) & := \Bigg(\prod_{i=1}^{m_{k-1}} R^{(k,k)}_{aa^{k-1}_i}\big(v+u^{(k-1)}_i+\rho\big) \Bigg)\el & \qq\qq \times \hat{D}^{(k)}_{a {\bm a}^{1 \dots k-2}}\big(v;\bm u^{(1 \dots k-2)}\big) \Bigg(\prod_{i=m_{k-1}}^{1} R^{(k,k)}_{aa^{k-1}_i}\big(v-u^{(k-1)}_i\big) \Bigg) , \label{Bnr-D1}
}
where $\hat{D}^{(k)}_{a {\bm a}^{1 \dots k-2}}(v;\bm u^{(1 \dots k-2)})$ is defined by \eqref{D-hat-D-hat},
\ali{
\big[\hat{D}^{(k)}_{a {\bm a}^{1 \dots k-2}}\big(v;\bm u^{(1 \dots k-2)}\big) \big]_{ij} & := \big[ \hat{D}^{(k-1)}_{a {\bm a}^{1 \dots k-2}} \big(v+\tfrac12;\bm u^{(1 \dots k-2)}\big) \big]_{1+i,1+j} \el & \qu\;\; + \frac{\del_{ij}}{2v+\rho}\big[ \hat{D}^{(k-1)}_{a {\bm a}^{1 \dots k-2}}\big(v+\tfrac12;\bm u^{(1 \dots k-2)}\big) \big]_{11} \label{Bnr-D2}
}
for $1\le i,j \le n-k+1$. 
\end{defn}

\begin{prop} \label{P:D-factor}
Let $\equiv$ denote equality of operators in the space $V_a^{(k)} \ot M^{(k)}$ for any $2\le k \le n$. Then
\ali{ 
\hat{D}^{(k)}_{a {\bm a}^{1\dots k-1}}\big(v;\bm u^{(1\dots k-1)}\big) & \equiv \left(\prod_{j=k-1}^1 \prod_{i=1}^{m_{j}} R^{(k,j+1)}_{aa^{j}_i}\big(v+u^{(j)}_i+\tfrac{k-1-j}{2}+\rho\big) \right)\el & \qq\qq \times \hat{D}^{(k)}_{a}(v) \left(\prod_{j=1}^{k-1} \prod_{i=m_{j}}^{1} R^{(k,j+1)}_{aa^{j}_i}\big(v-u^{(j)}_i+\tfrac{k-1-j}{2}\big) \right) . \label{Bnr-D3}
}
\end{prop}

Introduce a rational function
\equ{ \label{Bnr-La-PM}
\La^\pm(v;\bm u^{(k)}) := \prod_{i=1}^{m_k} \frac{(v+u^{(k)}_i\pm1+\rho)(v-u^{(k)}_i\pm1)}{(v+u^{(k)}_i+\rho)(v-u^{(k)}_i)} 
}
The technical lemma below will help us to prove Proposition \ref{P:D-factor}.

\begin{lemma} \label{L:M=RMR}
Let $A^{(k)}_a(v) \in \End(V_a^{(k)})[[v^{-1}]]$ be a matrix operator such that $[A_a^{(k)}(v)]_{1+i,1}=0$ for $i\ge1$. Set 
\aln{
A^{(k)}_{a \bm a^{k-1}}(v;\bm u^{(k-1)}) := \Bigg(\prod_{i=1}^{m_{k-1}} R^{(k,k)}_{aa^{k-1}_i}\big(v+u^{(k-1)}_i+\rho\big) \Bigg) A^{(k)}_{a}(v) \Bigg(\prod_{i=m_{k-1}}^{1} R^{(k,k)}_{aa^{k-1}_i}\big(v-u^{(k-1)}_i\big) \Bigg) .
}
Then, for $1\le i,j\le n-k$ and $\xi = (e^{(k)}_1)^{\ot m_{k-1}} \in W^{(k)}_{\bm a^{k-1}}$,
\ali{
\big[A^{(k)}_{a \bm a^{k-1}}(v;\bm u^{(k-1)})\big]_{11} \xi &= \big[A^{(k)}_a(v) \big]_{11} \xi , \qq
\big[A^{(k)}_{a \bm a^{k-1}}(v;\bm u^{(k-1)})\big]_{1+i,1} \xi = 0 , \label{Bnr-M-1} \\
\big[A^{(k)}_{a \bm a^{k-1}}(v;\bm u^{(k-1)})\big]_{1+i,1+j} \xi &= \frac{1}{\La^-(v;\bm u^{(k-1)})}\bigg( \!\big[A^{(k)}_{a}(v)\big]_{1+i,1+j} \el & \hspace{3.3cm} + \del_{ij}\frac{1-\La^-(v;\bm u^{(k-1)})}{2v-1+\rho} \big[A^{(k)}_{a}(v)\big]_{11} \!\bigg) \xi . \label{Bnr-M-2}
}
\end{lemma}

\begin{proof}
The first two identities follow from
\[
\big[ R^{(k,k)}_{aa^{k-1}_l}(v) \big]_{11} \xi = \xi , \qq  
\big[ R^{(k,k)}_{aa^{k-1}_l}(v) \big]_{1+i,1} \xi = 0 .
\]
To prove the third identity we need to use 
\[
\big[ R^{(k,k)}_{aa^{k-1}_l}(v) \big]_{1+i,1+j} \xi = \frac{v}{v-1}\,\del_{ij} \xi , \qq \big[ R^{(k,k)}_{aa^{k-1}_l}(v) \big]_{1+i,1} \big[ R^{(k,k)}_{aa^{k-1}_l}(u) \big]_{1,1+j} \xi = \frac{vu}{(v-1)(u-1)}\cdot \frac{1}{vu} \,\del_{ij} \xi 
\] 
giving
\[
\big[M^{(k)}_{a \bm a^{k-1}}(v;\bm u^{(k-1)})\big]_{1+i,1+j} \xi = \frac1{\La^-(v;\bm u^{(k-1)})} \Bigg( \big[M^{(k)}_{a}(v)\big]_{1+i,1+j} - \del_{ij} f(v;\bm u^{(k-1)})  \big[M^{(k)}_{a}(v)\big]_{11} \Bigg) \xi 
\]
where
\[
f(v;\bm u^{(k-1)}) = \sum_{i=1}^{m_{k-1}}\frac{1}{(v+u^{(k-1)}_i+\rho)(v-u^{(k-1)}_i)} \prod_{j=1}^{i-1} \frac{(v+u^{(k-1)}_i-1+\rho)(v-u^{(k-1)}_i-1)}{(v+u^{(k-1)}_i+\rho)(v-u^{(k-1)}_i)} .
\]
A simple induction on $m_{k-1}$ then yields
\[
f(v;\bm u^{(k-1)}) = \frac{1-\La^-(v;\bm u^{(k-1)})}{2v-1+\rho} ,
\]
implying the third identity.
\end{proof}

\begin{proof}[Proof of Proposition \ref{P:D-factor}]
It is sufficient to prove that (cf.\ \eqref{Bnr-D2})
\ali{
\hat{D}^{(k)}_{a {\bm a}^{1 \dots k-2}}(v;\bm u^{(1 \dots k-2)}) & \equiv \left(\prod_{j=k-2}^1 \prod_{i=1}^{m_{j}} R^{(k,j+1)}_{aa^{j}_i}\big(v+u^{(j)}_i+\tfrac{k-1-j}{2}+\rho\big) \right) \el & \qq\qq
 \times \hat{D}^{(k)}_{a}(v) \left(\prod_{j=1}^{k-2} \prod_{i=m_{j}}^{1} R^{(k,j+1)}_{aa^{j}_i}\big(v-u^{(j)}_i+\tfrac{k-1-j}{2}\big) \right). \label{D=RDR-proof}
}
We will use induction on $k$ to prove the claim. The $k=2$ case follows from the definition and provides a base for induction.
Now assume that the statement holds for $\hat{D}^{(k-1)}_{a {\bm a}^{1 \dots k-3}}(v;\bm u^{(1 \dots k-3)})$. Note that
\equ{
\big[R^{(k,l)}_{ab}(v)\big]_{ij} e^{(l)}_1 = \frac{v}{v-1}\,\del_{ij} e^{(l)}_1  \label{R1=1}
}
for $1\le i,j\le n-k+1$ and any $k>l$. Combining this with Lemma \ref{L:M=RMR} we obtain
\aln{
\big[ \hat{D}^{(k-1)}_{a {\bm a}^{1 \dots k-2}}(v+\tfrac12;\bm u^{(1 \dots k-2)}) \big]_{11} & \equiv \left( \prod_{l=k-3}^1 \frac1{\La^-\big(v+\tfrac{k-1-l}{2};\bm u^{(l)}\big)} \right) \big[\hat{D}^{(k-1)}_{a}(v+\tfrac12)\big]_{11} ,
\\
\big[ \hat{D}^{(k-1)}_{a {\bm a}^{1 \dots k-2}}(v+\tfrac12;\bm u^{(1 \dots k-2)}) \big]_{1+i,1+j} & \equiv \left( \prod_{l=k-2}^1 \frac1{\La^-\big(v+\tfrac{k-1-l}{2};\bm u^{(l)}\big)} \right) \Bigg( \big[\hat{D}^{(k-1)}_{a}(v+\tfrac12)\big]_{1+i,1+j} \\ & \hspace{3cm} + \del_{ij} \frac{1-\La^-\big(v+\tfrac12;\bm u^{(k-2)}\big)}{2v+\rho} \big[\hat{D}^{(k-1)}_{a}(v+\tfrac12)\big]_{11} \Bigg)
}
for $1\le i,j \le n-k+1$. The identities above together with \eqref{Bnr-D2} and \eqref{D-hat-D-hat} imply  
\aln{
\big[\hat{D}^{(k)}_{a {\bm a}^{1 \dots k-2}}\big(v;\bm u^{(1 \dots k-2)}\big) \big]_{ij} 
& \equiv \left( \prod_{l=k-2}^1 \frac1{\La^-\big(v+\tfrac{k-1-l}{2};\bm u^{(l)}\big)} \right) \big[\hat{D}^{(k)}_{a}(v)\big]_{ij} 
}
which is equivalent to \eqref{D=RDR-proof}, as required.
\end{proof}

The Corollary bellow follows from Propositions \ref{P:D-factor} and \ref{P:rest-bnr} and a virtue of the Yang-Baxter equation.

\begin{crl} \label{C:Dk-RE-Bnr}
Let $\equiv$ denote equality of operators in the space $V_a^{(k)} \ot V_b^{(k)} \ot M^{(k)}$ for any $2\le k \le n$. Then
\aln{
& R_{ab}^{(k,k)}(v-w) \, D^{(k)}_{a{\bm a}^{1\dots k-1}}\big(v;\bm u^{(1\dots k-1)}\big) \, R_{ab}^{(k,k)}(v+w) \, D_{b{\bm a}^{1\dots k-1}}^{(k)}\big(w;\bm u^{(1 \dots k-1)}\big) \\
& \qq \equiv D_{b{\bm a}^{1 \dots k-1}}^{(k)}\big(w;\bm u^{(1 \dots k-1)}\big) \, R_{ab}^{(k,k)}(v+w) \, D_{a{\bm a}^{ 1\dots k-1}}^{(k)}\big(v;\bm u^{(1 \dots k-1)}\big) \, R_{ab}^{(k,k)}(v-w) .
}
In other words, matrix entries of the level-$k$ monodromy matrix satisfy the defining relations of the algebra $\wt{\mcB}_\rho(n-k+1,r-k+1)$ or $\wt{\mcB}_\rho(n-k+1,0)$ in $M^{(k)}$ for $r>k-1$ or $r\le k-1$, respectively.
\end{crl}


\subsection{Transfer matrix, creation operators and Bethe vectors}

We are now ready to introduce transfer matrices and creation operators acting on the level-$k$ quantum space $M^{(k)}$.


\nc{\Ak}[2]{\hat{\mr{a}}^{(#1)}_{\bm a^{1 \dots #1-1}}\big(#2; \bm u^{(1 \dots #1-1)}\big)}

\begin{defn}
The level-$k$ $a$-operator is the first diagonal entry of the level-$k$ monodromy matrix, namely
\equ{
\Ak{k}{v} := \big[\hat{D}^{(k)}_{a {\bm a}^{1 \dots k-1}}\big(v+\tfrac12;\bm u^{(1 \dots k-1)}\big) \big]_{11} . \label{Bnr-a-op}
}
\end{defn}


\begin{defn}
The level-$k$ transfer matrix for a $\tBnr$-chain is obtained by taking trace of the level-$k$ monodromy matrix, namely
\ali{
\tau^{(k)}\big(v;\bm u^{(1\dots k-1)}\big) & := \tr_a \hat{D}^{(k)}_{a \bm a^{1\dots k-1}}\big(v-\tfrac{k-1}{2} ; \bm u^{(1 \dots k-1)} \big) \el
& \; =  \frac{2v-n+\rho}{2v-k+\rho}\cdot \Ak{k}{v-\tfrac{k}{2}} + \tr_a \hat{D}^{(k+1)}_{a\bm a^{1\dots k-1}}\big(v-\tfrac{k}{2} ; \bm u^{(1 \dots k-1)} \big). \label{Bnr-t(v)}
}
\end{defn}

Our goal is to find eigenvectors (Bethe vectors) of the level-$1$ transfer matrix $\tau^{(1)}(v)$ and the corresponding eigenvalues. With this goal in mind we introduce a lowest weight vector with respect to the action of the level-$n$ monodromy matrix,
\[
\eta^{(n)} := (e^{(n)}_1)^{\ot m_{n-1}} \ot \cdots \ot (e^{(2)}_1)^{\ot m_{1}} \ot \eta \in M^{(n)}.
\]
This vector will serve as a vacuum vector for constructing Bethe vectors.

\begin{lemma}
The level-$k$ $a$-operator acts on vector $\eta^{(n)}$ by
\equ{
\Ak{k}{v-\tfrac k2}\, \eta^{(n)} = \left(\prod_{i=1}^{k-2} \frac1{\La^-\big(v-\tfrac{i}{2};\bm u^{(i)}\big)}\right) \frac{\wt{\ga}^\circ_k(v)}{2v-k+1+\rho} \,\, \eta^{(n)}  . \label{Bnr-a-eig}
}
\end{lemma}

\begin{proof}
Recall that $\eta$ is a lowest vector of weight $\ga^\circ(v)$ with components $\ga^\circ_i(v)$ determined by \eqref{Bnr-gai}. It follows from \eqref{D-hat} and \eqref{tm(u)} that
\[
\big[\hat{D}^{(k)}_a\big(v-\tfrac{k-1}2\big) \big]_{11}\eta^{(n)} = \bigg( \ga^\circ_k(v) + \sum_{i=1}^{k-1} \frac{\ga^\circ_i(v)}{2v-k+1+\rho}\bigg) \eta^{(n)} = \frac{\wt{\ga}^\circ_k(v)}{2v-k+1+\rho} \,\, \eta^{(n)} .
\]
All that remains is to apply Proposition \ref{P:D-factor}, Lemma \ref{L:M=RMR} and identity \eqref{R1=1}.
\end{proof}

From now on we will view $B$-operators (cf.\ \eqref{Bnr-D}) of the monodromy matrix $\hat{D}^{(k)}_{a {\bm a}^{1 \dots k}} \big(v;\bm u^{(1 \dots k-1)}\big)$ as row-vectors, that is 
\[
\hat{B}^{(k)}_{a\bm a^{1\dots k-1}}\big(v;\bm u^{(1 \dots k-1)}\big) \in (V^{(k+1)}_{a})^* \ot \End (M^{(k)})[v^{-1}].
\]
These row-vectors will give rise to level-$k$ creation operators and level-$k$ Bethe vectors. Since $M^{(k)}$ is a finite-dimensional vector space, we can evaluate the formal parameter $v$ to any non-zero complex number.

\nc{\Bk}[1]{\mathscr{B}^{(#1)}_{\bm a^{1\dots #1}}\big(\bm u^{(1 \dots #1)}\big)}

\begin{defn}
The level-$k$ creation operator is defined by
\[
\Bk{k} := \prod_{i=1}^{m_k} \hat{B}^{(k)}_{a_i^k \bm a^{1 \dots k-1}}\big(u^{(k)}_i;\bm u^{(1 \dots k-1)}\big) .
\]
\end{defn}

Note that operator $\Bk{k}$ is a row-vector with respect to all tensorands in $W^{(k+1)}_{\bm a^{k}}$.


\begin{defn} \label{D:Bnr-BV-k}
The level-$k$ Bethe vector is defined by
\[
\Phi^{(k)}\big(\bm u^{(k\dots n-1)}; \bm u^{(1,\dots k-1)} \big) := \prod_{i=k}^{n-1} \Bk{i} \cdot \eta^{(n)} . 
\]
where $\bm u^{(1,\dots k-1)}$ are viewed as fixed parameters.
\end{defn}

The level-1 Bethe vector $\Phi^{(1)}\big(\bm u^{(1\dots n-1)} \big) \in M^{(1)}$ is a vector in the level-$1$ quantum space. For arbitrary $\bm u^{(1\dots n-1)}$ it is called an off-shell Bethe vector.


\begin{exam}
Recall \eqref{D-hat-B-hat}. Given $e^{(k+1)}_j\in V^{(k+1)}_{a^k_i}$ observe that
\aln{
\hat{B}^{(k)}_{a_i^k \bm a^{1 \dots k-1}}\big(u^{(k)}_i;\bm u^{(1 \dots k-1)}\big) e^{(k+1)}_j 
& = \big[\hat{D}^{(k)}_{a_i^k \bm a^{1 \dots k-1}}\big(u^{(k)}_i+\tfrac12;\bm u^{(1 \dots k-1)}\big)\big]_{1,1+j}
\\
& =
(e^{(k)}_1)^* \hat{D}^{(k)}_{a_i^k \bm a^{1 \dots k-1}}\big(u^{(k)}_i+\tfrac12;\bm u^{(1 \dots k-1)}\big) \,e^{(k)}_{j+1}.
}
%
%
Let $n\ge 2$, $m_1\ge 1$ and $m_2 = \dots = m_{n-1} = 0$. Then 
\aln{
\Phi^{(1)}\big(u^{(1)}_1, \dots , u^{(1)}_{m_1}\big) & = \hat{B}^{(1)}_{a_1^1}(u_1^{(1)}) \cdots \hat{B}^{(1)}_{a_{m_1}^1}(u_{m_1}^{(1)}) \cdot (e^{(2)}_1)^{\ot m_1} \ot \eta \\ 
&= \big[D^{(1)}_{a^1_1}\big(u^{(1)}_1+\tfrac12\big)\big]_{12} \cdots \big[D^{(1)}_{a^1_{m_1}}\big(u^{(1)}_{m_1}+\tfrac12\big)\big]_{12} \cdot \eta .
}
%
%
Let $n\ge 3$, $m_1 = m_2 = 1$ and $m_3 = \dots = m_{n-1} = 0$. Then 
\aln{
\Phi^{(1)}\big(u^{(1)}_1, u^{(2)}_1\big) &= \hat{B}^{(1)}_{a_1^1}(u_1^{(1)}) \, \hat{B}^{(2)}_{a_1^2}(u_1^{(2)})\cdot e^{(3)}_1 \ot e^{(2)}_1 \ot \eta \\
& = (e^{(2)}_1)^* \ot (e^{(1)}_1)^* \cdot \hat{D}^{(1)}_{a^1_1}\big(u^{(1)}_1\!+\tfrac12\big) \\ 
& \qu \times R^{(2,2)}_{a^2_1a_1^1} \big(u^{(2)}_1\! + u^{(1)}_1\! + \tfrac12+ \rho \big) \,\hat{D}^{(2)}_{a^2_1}\big(u^{(2)}_1\!+\tfrac12\big)\, R^{(2,2)}_{a^2_1a_1^1} \big(u^{(2)}_1\!-u^{(1)}_1\!+\tfrac12\big) \cdot e^{(2)}_2 \ot e^{(2)}_1 \ot \eta \\
& = \frac{u^{(2)}_1\! + u^{(1)}_1\! +\tfrac12 + \rho}{u^{(2)}_1\! + u^{(1)}_1\! - \tfrac12+ \rho}\cdot \frac{u^{(2)}_1\! - u^{(1)}_1\! +\tfrac12}{u^{(2)}_1\! - u^{(1)}_1\! - \tfrac12}  \cdot (e^{(1)}_1)^* \cdot \hat{D}^{(1)}_{a^1_1}\big(u^{(1)}_1\!+\tfrac12\big) \\ 
& \qu \times \Bigg( \!\big[\hat{D}^{(2)}_{a^2_1}\big(u^{(2)}_1\! + \tfrac12\big)\big]_{12} \, e^{(2)}_1  - \frac{1}{u^{(2)}_1\! - u^{(1)}_1\! + \tfrac12} \big[\hat{D}^{(2)}_{a^2_1}\big(u^{(2)}_1\! + \tfrac12\big)\big]_{11} \, e^{(2)}_2 \\ 
& \hspace{4.3cm} - \frac{1}{u^{(2)}_1\! + u^{(1)}_1\! + \tfrac12 + \rho} \sum_{i=1}^{2} \big[\hat{D}^{(2)}_{a^2_1}\big(u^{(2)}_1\! +\tfrac12\big)\big]_{i2} \, e^{(2)}_i \!\Bigg)\cdot\eta \\
& = \frac{1}{u^{(2)}_1\! - u^{(1)}_1\! -\tfrac12} \Bigg( \big(u^{(2)}_1\! - u^{(1)}_1\!+\tfrac12\big) \big[\hat{D}^{(1)}_{a^1_1}\big(u^{(1)}_1\!+\tfrac12\big)\big]_{12} \big[\hat{D}^{(2)}_{a^2_1}\big(u^{(2)}_1\!+\tfrac12\big)\big]_{12} \\ 
& \hspace{3.1cm} - \frac{u^{(2)}_1\! + u^{(1)}_1\!+\tfrac12+ \rho}{u^{(2)}_1\! + u^{(1)}_1\!-\tfrac12+\rho}\big[\hat{D}^{(1)}_{a^1_1}\big(u^{(1)}_1\!+\tfrac12\big)\big]_{13} \big[\hat{D}^{(2)}_{a^2_1}\big(u^{(2)}_1\!+\tfrac12\big)\big]_{11} \\ 
& \hspace{3.1cm} - \frac{u^{(2)}_1\! - u^{(1)}_1\! + \tfrac12}{u^{(2)}_1\! + u^{(1)}_1\! -\tfrac12+ \rho} \big[\hat{D}^{(1)}_{a^1_1}(u^{(1)}_1\!+\tfrac12)\big]_{13} \big[\hat{D}^{(2)}_{a^2_1}\big(u^{(2)}_1\!+\tfrac12\big)\big]_{22} \Bigg)\cdot \eta .
}
%
%
Let $n\ge 4$, $m_1 = m_2 = m_3 = 1$ and $m_4 = \dots = m_{n-1} = 0$. Then
\ali{
\Phi^{(1)}\big(u^{(1)}_1, u^{(2)}_1, u^{(3)}_1\big) &= \hat{B}^{(1)}_{a_1^1}(u_1^{(1)}) \, \hat{B}^{(2)}_{a_1^2}(u_1^{(2)})\, \hat{B}^{(3)}_{a_3^1}(u_1^{(3)}) \cdot e^{(4)}_1 \ot e^{(3)}_1 \ot e^{(2)}_1 \ot \eta \el
& = (e^{(3)}_1)^* \ot (e^{(2)}_1)^* \ot (e^{(1)}_1)^* \cdot \hat{D}^{(1)}_{a^1_1}\big(u^{(1)}_1\!+\tfrac12\big) \el 
& \qu \times R^{(2,2)}_{a^2_1a_1^1} \big(u^{(2)}_1\! + u^{(1)}_1\! + \tfrac12+ \rho \big) \,\hat{D}^{(2)}_{a^2_1}\big(u^{(2)}_1\!+\tfrac12\big)\, R^{(2,2)}_{a^2_1a_1^1} \big(u^{(2)}_1\!-u^{(1)}_1\!+\tfrac12\big) \label{Bnr-BV32}\\
& \qu \times R^{(3,3)}_{a^3_1a^2_1} \big(u^{(3)}_1\! + u^{(2)}_1\! + \tfrac12+ \rho \big)\,R^{(3,2)}_{a^3_1a^1_1} \big(u^{(3)}_1\! + u^{(1)}_1\! + 1 + \rho \big) \label{Bnr-BV33}\\ 
& \qu \times \hat{D}^{(3)}_{a^3_1}\big(u^{(3)}_1\!+\tfrac12\big)\,R^{(3,2)}_{a^3_1a^1_1} \big(u^{(3)}_1\!-u^{(1)}_1\!+1\big)\, R^{(3,3)}_{a^3_1a^2_1} \big(u^{(3)}_1\!-u^{(2)}_1\!+\tfrac12\big) \cdot e^{(3)}_2 \ot e^{(3)}_1 \ot e^{(2)}_1 \ot \eta . \label{Bnr-BV34}
}
Acting with $(e^{(3)}_1)^*$ on lines \eqref{Bnr-BV33} and \eqref{Bnr-BV34} gives, up to an overall scalar,
\aln{
 & \frac{u^{(3)}_1\!+u^{(2)}_1\!-\tfrac12+\rho}{u^{(3)}_1\!+u^{(2)}_1\!+\tfrac12+\rho} \big[ \hat{D}^{(3)}_{a^3_1}\big(u^{(3)}_1\!+\tfrac12\big) \big]_{12} \cdot e^{(3)}_1  \ot e^{(2)}_1 \ot \eta \\
 & - \Bigg( \frac{1}{u^{(3)}_1\!-u^{(2)}_1\!+\tfrac12} \big[ \hat{D}^{(3)}_{a^3_1}\big(u^{(3)}_1\!+\tfrac12\big) \big]_{11} + \frac{1}{u^{(3)}_1\!+u^{(2)}_1\!+\tfrac12+\rho} \big[ \hat{D}^{(3)}_{a^3_1}\big(u^{(3)}_1\!+\tfrac12\big) \big]_{22} \Bigg)\cdot  e^{(3)}_2  \ot e^{(2)}_1 \ot \eta .
}
Writing the above expression as $A_1\cdot e^{(2)}_2 \ot e^{(2)}_1 \ot \eta + A_2\cdot e^{(2)}_3 \ot e^{(2)}_1 \ot \eta$ and acting with $(e^{(2)}_1)^*$ and \eqref{Bnr-BV32} we obtain, up to an overall scalar,
\aln{
&  \frac{u^{(2)}_1\!+u^{( 1)}_1\!-\tfrac12+\rho}{u^{(2)}_1\!+u^{( 1)}_1\!+\tfrac12+\rho} \Big( \big[ \hat{D}^{(2)}_{a^2_1}\big(u^{(2)}_1\!+\tfrac12\big) \big]_{12} A_1 +  \big[ \hat{D}^{(2)}_{a^2_1}\big(u^{(2)}_1\!+\tfrac12\big) \big]_{13} A_2 \Big) \cdot e^{(2)}_1 \ot \eta \\
& - \Bigg( \frac{\big[ \hat{D}^{(2)}_{a^2_1}\big(u^{(2)}_1\!+\tfrac12\big) \big]_{11} A_1 }{u^{(2)}_1\!-u^{( 1)}_1\!+\tfrac12} + \frac{ \big[ \hat{D}^{(2)}_{a^2_1}\big(u^{(2)}_1\!+\tfrac12\big) \big]_{22} A_1 + \big[ \hat{D}^{(2)}_{a^2_1}\big(u^{(2)}_1\!+\tfrac12\big) \big]_{23} A_2  }{u^{(2)}_1\!+u^{( 1)}_1\!+\tfrac12+\rho} \Bigg) \cdot  e^{(2)}_2 \ot \eta \\
& - \Bigg( \frac{ \big[ \hat{D}^{(2)}_{a^2_1}\big(u^{(2)}_1\!+\tfrac12\big) \big]_{12} A_2 }{u^{(2)}_1\!+u^{(1)}_1\!+\tfrac12+\rho} + \frac{\big[ \hat{D}^{(2)}_{a^2_1}\big(u^{(2)}_1\!+\tfrac12\big) \big]_{13} A_1 }{u^{(2)}_1\!-u^{(1)}_1\!+\tfrac12}  \Bigg) \cdot e^{(2)}_3 \ot \eta .
}
Finally, acting with $(e^{(1)}_1)^* \cdot \hat{D}^{(1)}_{a^1_1}\big(u^{(1)}_1\!+\tfrac12\big)$ yields, up to an overall scalar,
\aln{
& \Bigg( \frac{u^{(2)}_1\!+u^{( 1)}_1\!-\tfrac12+\rho}{u^{(2)}_1\!+u^{( 1)}_1\!+\tfrac12+\rho} \big[ \hat{D}^{(1)}_{a^1_1}\big(u^{(1)}_1\!+\tfrac12\big) \big]_{11} \Big( \big[ \hat{D}^{(2)}_{a^2_1}\big(u^{(2)}_1\!+\tfrac12\big) \big]_{12} A_1 + \big[ \hat{D}^{(2)}_{a^2_1}\big(u^{(2)}_1\!+\tfrac12\big) \big]_{13} A_2 \Big) \\
& \qu - \big[ \hat{D}^{(1)}_{a^1_1}\big(u^{(1)}_1\!+\tfrac12\big) \big]_{12}  \Bigg( \frac{\big[ \hat{D}^{(2)}_{a^2_1}\big(u^{(2)}_1\!+\tfrac12\big) \big]_{11} A_1 }{u^{(2)}_1\!-u^{( 1)}_1\!+\tfrac12} + \frac{ \big[ \hat{D}^{(2)}_{a^2_1}\big(u^{(2)}_1\!+\tfrac12\big) \big]_{22} A_1 + \big[ \hat{D}^{(2)}_{a^2_1}\big(u^{(2)}_1\!+\tfrac12\big) \big]_{23} A_2  }{u^{(2)}_1\!+u^{( 1)}_1\!+\tfrac12+\rho} \Bigg) \\
& \qu - \big[ \hat{D}^{(1)}_{a^1_1}\big(u^{(1)}_1\!+\tfrac12\big) \big]_{13} \Bigg( \frac{ \big[ \hat{D}^{(2)}_{a^2_1}\big(u^{(2)}_1\!+\tfrac12\big) \big]_{12} A_2 }{u^{(2)}_1\!+u^{(1)}_1\!+\tfrac12+\rho} + \frac{\big[ \hat{D}^{(2)}_{a^2_1}\big(u^{(2)}_1\!+\tfrac12\big) \big]_{13} A_1 }{u^{(2)}_1\!-u^{(1)}_1\!+\tfrac12}  \Bigg) \Bigg) \cdot \eta .
}
The total overall scalar is $\La^-\big(u^{(2)}_1+\tfrac12;u^{(1)}_1\big) \,
\La^-\big(u^{(3)}_1+\tfrac12;u^{(1)}_1\big) \,
\La^-\big(u^{(3)}_1+\tfrac12;u^{(2)}_1\big) $. 
\end{exam}

Set $\mf{S}_{m_k,\dots,m_{n-1}} := \mf{S}_{m_k}\times \cdots \times \mf{S}_{m_{n-1}}$.
For any $\si^{(l)} \in \mf{S}_{m_l}$ with $k\le l \le n-1$ define an action of $\mf{S}_{m_k,\dots,m_{n-1}}$ on $\Phi^{(k)}(\bm u^{(k\dots n-1)})$ by 
\[
\si^{(l)}: \bm u^{(k\dots n-1)} \mapsto \bm u^{(k\dots n-1)}_{\si^{(l)}} := ( \bm u^{(k)} , \dots , \bm u^{(l)}_{\si^{(l)}}, \dots , \bm u^{(n-1)} ) \qu\text{with}\qu \bm u^{(l)}_{\si^{(l)}} := (u^{(l)}_{\si^{(l)}(1)}, \dots , u^{(l)}_{\si^{(l)}(m_l)}).
\]
The relation \eqref{Bnr-BB} together with the identity $\check{R}^{(l,l)}_{a^l_ia^l_j}(u) \,\eta^{(n)} = \eta^{(n)}$ implies the following Lemma.

\begin{lemma} \label{L:Bnr-bethe-symm}
The level-$k$ Bethe vector $\Phi^{(k)}\big(\bm u^{(k\dots n-1)}; \bm u^{(1\dots k-1)} \big) $ is invariant under the action of $\mf{S}_{m_k,\dots,m_{n-1}}$.
\end{lemma}


The Theorem below is the first main result of Section \ref{sec:NABA-RE}.

\begin{thrm} \label{T:Bnr-spec}
The level-1 Bethe vector $\Phi^{(1)}\big(\bm u^{(1\dots n-1)} \big)$ is an eigenvector of $\tau^{(1)}(v)$ with the eigenvalue 
\ali{
\La^{(1)}\big(v;\bm u^{(1\dots n-1)}\big) &:= 
\frac{2v-n+\rho}{2v-1+\rho}\,\,\La^+\big(v-\tfrac 12;\bm u^{(1)}\big)\,\frac{\wt{\ga}^\circ_{1}(v)}{2v+\rho} 
+ \La^-\big(v-\tfrac{n-1}2;\bm u^{(n-1)}\big)\, \frac{\wt{\ga}^\circ_n(v)}{2v-n+1+\rho} \el & \qu 
+ \sum_{i=2}^{n-1} \frac{2v-n+\rho}{2v-i+\rho}\,\,\La^-\big(v-\tfrac{i-1}2;\bm u^{(i-1)}\big)\,\La^+\big(v-\tfrac i2;\bm u^{(i)}\big)\,\frac{\wt{\ga}^\circ_{i}(v)}{2v-i+1+\rho} \label{La(v,u)}
}
provided 
\equ{
\Res{v\to u^{(j)}_i} \La^{(1)}\big(v+\tfrac j2;\bm u^{(1\dots n-1)}\big) = 0 \label{Bnr-BE}
}
for all $1\le i \le m_j$ and $1\le j \le n-1$.
\end{thrm}

\begin{rmk}
The equations \eqref{Bnr-BE} are Bethe equations for a $\tBnr$-chain. Their explicit form is
\ali{
& \frac{\wt{\ga}^\circ_{k}\big(u^{(k)}_j+\tfrac k2\big)}{\wt{\ga}^\circ_{k+1}\big(u^{(k)}_j+\tfrac{k}{2}\big)} \prod_{\substack{i=1\\i\ne j}}^{m_k} \frac{(u^{(k)}_j-u^{(k)}_i+1)(u^{(k)}_j+u^{(k)}_i+1+\rho)}{(u^{(k)}_j-u^{(k)}_i-1)(u^{(k)}_j+u^{(k)}_i-1+\rho)} \el
& \qu = \prod_{i=1}^{m_{k-1}} \frac{(u^{(k)}_j-u^{(k-1)}_i+\tfrac12)(u^{(k)}_j+u^{(k-1)}_i+\tfrac12+\rho)}{(u^{(k)}_j-u^{(k-1)}_i-\tfrac12)(u^{(k)}_j+u^{(k-1)}_i-\tfrac12+\rho)} \el & \qq \times  \prod_{i=1}^{m_{k+1}} \frac{(u^{(k)}_j-u^{(k+1)}_i+\tfrac12)(u^{(k)}_j+u^{(k+1)}_i+\tfrac12+\rho)}{(u^{(k)}_j-u^{(k+1)}_i-\tfrac12)(u^{(k)}_j+u^{(k+1)}_i-\tfrac12+\rho)} \label{Bnr-BE2} 
}
for $1\le j\le m_k$ and $1\le k \le n-1$ assuming $m_0 = m_n = 0$. For example, when $n=2$, we have $k=1$ and the r.h.s.\ of \eqref{Bnr-BE2} equals $1$. 
\end{rmk}

\begin{proof}[Proof of Theorem \ref{T:Bnr-spec}]
Using Lemma \ref{L:Ex-Bnr}, symmetry $\mf{S}_{m_k,\dots,m_{n-1}}$ of $\Phi^{(k)}\big(\bm u^{(k\dots n-1)};\bm u^{(1\dots k-1)}\big)$ and standard arguments, we obtain
\ali{
& \Ak{k}{v-\tfrac{k}2}\,\Phi^{(k)} \big(\bm u^{(k\dots n-1)};\bm u^{(1\dots k-1)} \big) \el
& \qu = \Bigg( \La^+\big(v-\tfrac{k}2;\bm u^{(k)}\big)\,\, \Bk{k} \,\, \Ak{k}{v-\tfrac{k}2}  \el 
& \qq\qu 
- \sum_{i=1}^{m_k} \frac{1}{v-\tfrac k2-u^{(k)}_i} \Res{w\to u^{(k)}_i}  \La^+\big(w;\bm u^{(k)}\big) \, \mr{B}^{(k)}_{\bm a^{1\dots k}}\big( \bm u^{(1\dots k)}_{\si^{(k)}_i,u^{(k)}_i \to v-\frac k2} \big) \,\, \Ak{k}{u^{(k)}_i}  \el 
&\qq\qu 
- \sum_{i=1}^{m_k} \frac{2u^{(k)}_i+\rho}{(v-\frac k2+u^{(k)}_i+\rho)(2u^{(k)}_i-n+k+\rho)} \Res{w\to u^{(k)}_i} \La^-\big(w;\bm u^{(k)}\big) \el & \qq\qq \times  \mr{B}^{(k)}_{\bm a^{1\dots k}}\big(\bm u^{(1\dots k)}_{\si^{(k)}_i,u^{(k)}_i\to v-\frac k2}\big) \tr_a \hat{D}^{(k+1)}_{a{\bm a}^{1\dots k}}\big(u^{(k)}_i;\bm u^{(1\dots k)}_{\si^{(k)}_i}\big) \Bigg) \Phi^{(k+1)}\big(\bm u^{(k+1\dots n-1)};\bm u^{(1\dots k)} \big)  \label{aF1}
\intertext{and}
& \tr_a\hat{D}^{(k+1)}_{a\bm a^{1\dots k-1}}\big(v-\tfrac k2;\bm u^{(1\dots k-1) }\big)\,\Phi^{(k)} \big(\bm u^{(k\dots n-1)};\bm u^{(1\dots k-1)} \big) 
\el
& \qu = \Bigg( \La^-\big(v-\tfrac k2;\bm u^{(k)}\big)\, \mr{B}^{(k)}_{\bm a^{1\dots k}}\big(\bm u^{(1\dots k)}\big) \tr_a \hat{D}^{(k+1)}_{a{\bm a}^{1\dots k}}\big(v-\tfrac k2;\bm u^{(1\dots k)} \big)  \el 
&\qq\qu 
- \sum_{i=1}^{m_1} \frac{2v-n+\rho}{(2v-k+\rho)(v-\frac k2+u^{(k)}_i+\rho)} \Res{w\to u^{(k)}_i} \La^+\big(w;\bm u^{(k)}\big) \el & \qq\qq \times \mr{B}^{(k)}_{\bm a^{1\dots k}}\big(\bm u^{(1\dots k)}_{\si^{(k)}_i,u^{(k)}_i\to v-\frac k2}\big) \, \Ak{k}{u^{(k)}_i} \el
&\qq\qu 
- \sum_{i=1}^{m_1} \frac{(2u^{(k)}_i+\rho)(2v-n+\rho)}{(2v-k+\rho)(v-\frac k2-u^{(k)}_i)(2u^{(k)}_i-n+k+\rho)} \Res{w\to u^{(k)}_i} \La^-\big(w;\bm u^{(k)}\big) \el & \qq\qq \times \mr{B}^{(k)}_{\bm a^{1\dots k}} \big(\bm u^{(1\dots k)}_{\si^{(k)}_i,u^{(k)}_i\to v-\frac k2}\big) \tr_a \hat{D}^{(k+1)}_{a{\bm a}^{1\dots k}}\big(u^{(k)}_i;\bm u^{(1\dots k)}_{\si^{(k)}_i}\big) \Bigg) \Phi^{(k+1)}\big(\bm u^{(k+1\dots n-1)};\bm u^{(1\dots k)} \big) . \label{tDF1}
}
Here $\si^{(k)}_i\in\mf{S}_{m_k}$ denotes a cyclic permutation such that $\bm u^{(k)}_{\si^{(k)}_i} = (u^{(k)}_{i}, u^{(k)}_{i+1}, \dots , u^{(k)}_{m_k}, u^{(k)}_{1}, u^{(k)}_{2}, \dots, u^{(k)}_{i-1})$.
Below we indicate key identities that were used in obtaining \eqref{aF1} and \eqref{tDF1}. For this we need to introduce additional notation. Set
$\bm a^k_{\not a^k_j} = (a^k_1,\dots,a^k_{j-1},a^k_{j+1},\dots, a^k_{m_k})$ and $\bm u^{(k)}_{\not{u^{(k)}_j}} = (u^{(k)}_1, \dots, u^{(k)}_{j-1} ,u^{(k)}_{j+1}, \dots,u^{(k)}_{m_k})$. 
Then let
\ali{ 
\hat{D}^{(k+1)}_{a\bm a^{1\dots k}_{\not a^k_1}}\big(v;\bm u^{(1\dots k)}_{\si^{(k)}_j,\not u^{(k)}_j}\big) &:= \left(\prod_{i=2}^{m_k} R^{(k+1,k+1)}_{aa^k_i}\big(v+u^{(k)}_{\si^{(k)}_j(i)}+\rho\big) \right) \el & \qq \times \hat{D}^{(k+1)}_{a \bm a^{1\dots k-1}}\big(v;\bm u^{(1\dots k-1)}\big) \left(\prod_{i=m_k}^{2} R^{(k+1,k+1)}_{aa^k_i}\big(v-u^{(k)}_{\si^{(k)}_j(i)}\big) \right) 
}
and
\[
\La^\pm \big(w;\bm u^{(k)}_{\not u^{(k)}_j} \big) := \prod_{\substack{i=1\\i\ne j}}^{m_k} \frac{(w-u^{(k)}_i\pm1)(w+u^{(k)}_i\pm1+\rho)}{(w-u^{(k)}_i)(w+u^{(k)}_i+\rho)} 
\]
so that
\[
\Res{w\to u^{(k)}_j} \La^\pm(w;\bm u^{(k)}) = \pm \frac{2u^{(k)}_j \pm 1 + \rho}{2u^{(k)}_j +\rho} \, \La^\pm \big(u^{(k)}_j;\bm u^{(k)}_{\not{u^{(k)}_j}} \big) .
\]
Also note that
\[
\tr_a R^{(k+1,k+1)}_{ab}(u)\, P^{(k+1,k+1)}_{ab} = \frac{u-n+k}{u-1}\cdot I^{(k+1,k+1)}_b \,.
\]
To obtain the second terms in the r.h.s.\ of \eqref{tDF1} we used  
\aln{
\Res{w\to u^{(k)}_i} \La^+\big(w;\bm u^{(k)}\big) &= \frac{2u^{(k)}_i + 1 + \rho}{2u^{(k)}_i +\rho} \cdot \frac{2v-k-1 +\rho}{2v-n+\rho}  \\ & \qu \times \La^+\big(u^{(k)}_i;\bm u^{(k)}_{\not u^{(k)}_i} \big) \tr_a R^{(k+1,k+1)}_{aa^k_1}(2v-k +\rho) P^{(k+1,k+1)}_{aa^k_1} .
}
To obtain the third term in the r.h.s.\ of \eqref{aF1} we used the second equality below, and to obtain the third term in the r.h.s.\ of in \eqref{tDF1} we used the third equality below:
\aln{
& \Res{w\to u^{(k)}_i} \La^-\big(w;\bm u^{(k)}\big)\, \tr_a \hat{D}^{(k+1)}_{a{\bm a}^{1\dots k}}\big(u^{(k)}_i;\bm u^{(1\dots k)}_{\si^{(k)}_i}\big) \el 
& \qu = - \frac{2u^{(k)}_i - 1 + \rho}{2u^{(k)}_i +\rho} \, \La^-\big(u^{(k)}_i;\bm u^{(k)}_{\not u^{(k)}_i} \big) \tr_a \Big( R^{(k+1,k+1)}_{aa^k_1}\big(2u^{(k)}_i+\rho\big) \hat{D}^{(k+1)}_{a\bm a^{1\dots k}_{\not a^k_1}}\big(u^{(k)}_i;\bm u^{(k)}_{\si^{(k)}_i,\not u^{(k)}_i}\big) P^{(k+1,k+1)}_{aa^k_1} \Big) \\
& \qu = - \frac{2u^{(k)}_i - n + k + \rho}{2u^{(k)}_i +\rho} \, \La^-\big(u^{(k)}_i;\bm u^{(k)}_{\not u^{(k)}_i} \big) \hat{D}^{(k+1)}_{a^k_1\bm a^{1\dots k}_{\not a^k_1}}\big(u^{(k)}_i;\bm u^{(k)}_{\si^{(k)}_i,\not u^{(k)}_i}\big) \\ 
& \qu = - \frac{2u^{(k)}_i - n + k + \rho}{2u^{(k)}_i +\rho} \cdot \frac{2v - k - 1 + \rho}{2v - n +\rho} \\ & \qq \times \La^-\big(u^{(k)}_i;\bm u^{(k)}_{\not u^{(k)}_i} \big) \tr_a \Big( R^{(k+1,k+1)}_{aa^k_1}(2v-k+\rho) \hat{D}^{(k+1)}_{a\bm a^{1\dots k}_{\not a^k_1}}\big(u^{(k)}_i;\bm u^{(k)}_{\si^{(k)}_i,\not u^{(k)}_i}\big) P^{(k+1,k+1)}_{aa^k_1} \Big) .
}

Combining \eqref{aF1} and \eqref{tDF1} gives
\ali{
& \tau^{(k)}\big(v; \bm u^{(1\dots k-1)}\big) \,\Phi^{(k)}\big(\bm u^{(k\dots n-1)};\bm u^{(1\dots k-1)}\big) \el
& \qu = \Bk{k}\, \bigg( \frac{2v-n+\rho}{2v-k+\rho}\,\,\La^+\big(v-\tfrac k2;\bm u^{(k)}\big)\,\Ak{k}{v-\tfrac k2} \el & \hspace{3.4cm} + \La^-\big(v-\tfrac k2;\bm u^{(k)}\big) \tr_a \hat{D}^{(k+1)}_{a{\bm a}^{1\dots k}}\big(v-\tfrac k2;\bm u^{(1\dots k)} \big) \bigg) \Phi^{(k+1)}\big(\bm u^{(k+1\dots n-1)};\bm u^{(1\dots k)}\big) \el 
& \qq
- \sum_{i=1}^{m_k} F_{n,k}\big(v,u^{(k)}_i\big)\,\, \mr{B}^{(k)}_{\bm a^{1\dots k}}\big(\bm u^{(1\dots k)}_{\si^{(k)}_i,u^{(k)}_i\to v-\frac k2}\big) \el & \qq\qu \times\Res{w\to u^{(k)}_i} \bigg(  \frac{2w-n+k+\rho}{2w+\rho} \,\, \La^+(w;\bm u^{(k)}) \, \Ak{k}{w} \el & \hspace{3.4cm} +  \La^-(w;\bm u^{(k)}) \,\tr_a \hat{D}^{(k+1)}_{a{\bm a}^{1\dots k}}\big(w;\bm u^{(1\dots k)}_{\si^{(k)}_i}\big) \bigg) \Phi^{(k+1)}\big(\bm u^{(k+1\dots n-1)};\bm u^{(1\dots k)}\big)  \label{tF}
}
where
\[
F_{n,k}(v,u) = \frac{(2v-n+\rho)(2u+\rho)}{(v-\tfrac k2-u)(v-\tfrac k2+u+\rho)(2u-n+k+\rho)} .
\]
When $k=n-1$ and $n=2$, using \eqref{Bnr-a-op} and \eqref{Bnr-a-eig} we have that $\Phi^{(k+1)}\big(\bm u^{(k+1\dots n-1)}\big) = \eta$ and
\[
\hat{\mr{a}}^{(1)}(w)\,\eta = \frac{\wt{\ga}^\circ_1(w+\tfrac12)}{2w+1+\rho}\, \eta , \qq
\tr_a \hat{D}^{(2)}_{a{\bm a}^{1}}\big(w;\bm u^{(1)}_{\si^{(1)}_i}\big)\,\eta = \frac{\wt{\ga}^\circ_2(w+\tfrac12)}{2w+\rho}\, \eta ,
\]
yielding
\aln{
& \tau^{(1)}(v) \,\Phi^{(1)}\big(\bm u^{(1)}\big) \el
& \qu = \bigg( \frac{2v-2+\rho}{2v-1+\rho}\,\, \La^+\big(v-\tfrac12;\bm u^{(1)}\big)\,\frac{\wt{\ga}^\circ_{1}(v)}{2v+\rho} + \La^-\big(v-\tfrac12;\bm u^{(1)}\big) \frac{\wt{\ga}^\circ_2(v)}{2v-1+\rho} \bigg) \Phi^{(1)}\big(\bm u^{(1)}\big) \\ 
& \qq
- \sum_{i=1}^{m_{1}} F_{2,1}\big(v,u^{(1)}_i\big) \Res{w\to u^{(1)}_i} \bigg( \frac{2w-1+\rho}{2w+\rho} \,\, \La^+(w;\bm u^{(1)}) \, \frac{\wt{\ga}^\circ_{1}(w+\tfrac12)}{2w+1+\rho} \\ & \hspace{6.23cm} + \La^-(w;\bm u^{(1)}) \,\frac{\wt{\ga}^\circ_{2}(w+\tfrac12)}{2w+\rho} \bigg) \Phi^{(1)} \big(\bm u^{(1)}_{\si^{(1)}_i,u^{(1)}_i\to v-\frac12}\big) \\
& \qu = \La^{(1)}\big(v-\tfrac12;\bm u^{(1)}\big)\, \Phi^{(1)}\big(\bm u^{(1)}\big) - \sum_{i=1}^{m_{1}} F_{2,1}\big(v,u^{(1)}_i\big) \Res{w\to u^{(1)}_i} \La^{(1)}\big(w;\bm u^{(1)}\big)\,  \Phi^{(1)} \big(\bm u^{(1)}_{\si^{(1)}_i,u^{(1)}_i\to v-\frac12}\big) .
}
This completes the proof when $n=2$. 
Assuming $1<k<n$ and $n>2$ introduce notation
\aln{
\La^{(k)}\big(v;\bm u^{(k-1\dots n-1)}\big) & := \sum_{l=k}^{n-1} \frac{2v-n+\rho}{2v-l+\rho}\,\,\La^-\big(v-\tfrac{l-1}2;\bm u^{(l-1)}\big)\,\La^+\big(v-\tfrac l2;\bm u^{(l)}\big)\,\frac{\wt{\ga}^\circ_{l}(v)}{2v-l+1+\rho} \\ & \hspace{5.5cm} + \La^-\big(v-\tfrac{n-1}2;\bm u^{(n-1)}\big)\, \frac{\wt{\ga}^\circ_n(v)}{2v-n+1+\rho} 
}
and notice that, for all $k\le l \le n-1$ and $1\le i \le m_l$,
\[
\Res{w\to u^{(l)}_i} \La^{(k)}\big(w+\tfrac{k-1}2;\bm u^{(k-1,k)}\big) = \Res{w\to u^{(l)}_i} \La^{(1)}\big(w;\bm u^{(1\dots n-1)}\big).
\]
Hence, when $k=n-1$ and $n>2$, using similar arguments as before and symmetry of the Bethe vector, we find that
\aln{
& \tau^{(n-1)}\big(v;\bm u^{(1\dots n-2)}\big) \,\Phi^{(n-1)}\big(\bm u^{(n-1)};\bm u^{(1\dots n-2)}\big) \el
& \qu = \Bigg( \prod_{j=1}^{n-2} \La^-\big(v-\tfrac{j}{2};\bm u^{(j)}\big) \Bigg)^{\!-1}\! \bigg( \frac{2v-n+\rho}{2v-n+1+\rho}\,\,\La^-\big(v-\tfrac{n-2}{2};\bm u^{(n-2)}\big)\, \La^+\big(v-\tfrac {n-1}2;\bm u^{(n-1)}\big)\,\frac{\wt{\ga}^\circ_{n-1}(v)}{2v-n+2+\rho} \\ & \hspace{6.2cm} + \La^-\big(v-\tfrac {n-1}2;\bm u^{(n-1)}\big)\, \frac{\wt{\ga}^\circ_n(v)}{2v-n+1+\rho} \bigg) \Phi^{(n-1)}\big(\bm u^{(n-1)};\bm u^{(1\dots n-2)}\big) \\ 
& \qq
- \sum_{i=1}^{m_{n-1}} F_{n,n-1}\big(v,u^{(n-1)}_i\big) \Bigg( \prod_{j=1}^{n-2} \La^-\big(u^{(n-1)}_i\!-\tfrac{j-n+1}2;\bm u^{(j)}\big) \Bigg)^{\!-1} \\ & \qq\qu \times \Res{w\to u^{(n-1)}_i} \bigg( \frac{2w-1+\rho}{2w+\rho} \,\, \La^-\big(w+\tfrac{1}{2};\bm u^{(n-2)}\big)\, \La^+\big(w;\bm u^{(n-1)}\big) \, \frac{\wt{\ga}^\circ_{n-1}(w+\tfrac{n-1}2)}{2w+1+\rho}\\ & \hspace{6.2cm} + \La^-\big(w;\bm u^{(n-1)}\big) \,\frac{\wt{\ga}^\circ_{n}(w+\tfrac{n-1}2)}{2w+\rho} \bigg) \Phi^{(n-1)} \big(\bm u^{(n-1)}_{u^{(n-1)}_i\to v-\frac {n-1}2};\bm u^{(1\dots n-2)}\big) \\
& \qu = \Bigg( \prod_{j=1}^{n-2} \La^-\big(v-\tfrac{j}{2};\bm u^{(j)}\big) \Bigg)^{\!-1} \La^{(n-1)}\big(v;\bm u^{(n-2,n-1)}\big)\,  \Phi^{(n-1)}\big(\bm u^{(n-1)};\bm u^{(1\dots n-2)}\big) 
}
provided $\Res{w\to u^{(n-1)}_i} \La^{(1)}\big(w;\bm u^{(1\dots n-1)}\big) = 0$ for all $1\le i \le m_{n-1}$.  
Next, when $1<k<n-1$ and $n>3$, using negative inductive arguments we obtain
\aln{
& \tau^{(k)}\big(v;\bm u^{(1\dots k-1)}\big) \,\Phi^{(k)}\big(\bm u^{(k\dots n-1)};\bm u^{(1\dots k-1)}\big) \\ & \qq\qq = \Bigg( \prod_{j=1}^{k-1} \La^-\big(v-\tfrac{j}{2};\bm u^{(j)}\big) \Bigg)^{\!-1} \La^{(k)}\big(v;\bm u^{(k-1\dots n-1)}\big) \, \Phi^{(k)}\big(\bm u^{(k\dots n-1)};\bm u^{(1\dots k-1)}\big) 
}
provided $\Res{w\to u^{(l)}_i} \La^{(1)}\big(w+\tfrac{l}2;\bm u^{(1\dots n-1)}\big) = 0$ for all $1\le i \le m_{l}$ and $k\le l \le n-1$.
Finally, when $k=1$ and $n>2$, we obtain
\aln{
& \tau^{(1)}(v) \,\Phi^{(1)}\big(\bm u^{(1\dots n-1)}\big) = \La^{(1)}\big(v;\bm u^{1\dots n-1}\big) \,\Phi^{(1)}\big(\bm u^{(1\dots n-1)}\big) 
}
provided  $\Res{w\to u^{(l)}_i} \La^{(1)}\big(w+\tfrac{l}2;\bm u^{(1\dots n-1)}\big) = 0$ for all $1\le i \le m_{l}$ and $1\le l \le n-1$, which completes the proof.
\end{proof}


\subsection{A trace formula for Bethe vectors}

A trace formula for Bethe vectors for a $\mfgl_{n|m}$-symmetric open spin chain was given in Theorem 7.1 of \cite{BeRa1}. Below we state a specialization of that theorem for a $\mfgl_{n}$-symmetric open spin chain, namely a $\tBnr$-chain.

\begin{thrm} \label{T:tf-bnr}
A Bethe vector for a $\tBnr$-chain can be written as
\ali{ \label{nested_tf}
\Phi^{(1)}\big(\bm u^{(1\dots n-1)}\big) & = \tr_{\ol{V}} \Bigg[ \Bigg( \prod_{i=1}^{n-1} \prod_{j=1}^{m^{(i)} \!\!} 
\Bigg( \Bigg(\prod_{k=1}^{j-1} R^{(1)}_{a^i_j a^i_k}(u^{(i)}_j\!+u^{(i)}_k\!+\rho+1) \Bigg)
\Bigg(\prod_{b=i-1}^{1} \prod_{c=1}^{m^{(b)}\!\!} R^{(i,b+1)}_{a^i_j a^b_c}\big(u^{(i)}_j\!+u^{(b)}_c\!+\rho+\tfrac{i-b}2\big) \!\Bigg)
 \el 
& \qq\qq \times \wt{D}^{(i)}_{a^i_j}\big(u^{(i)}_j\!+\tfrac12 \big) \Bigg(\prod_{b=1}^{i-1} \prod_{c=m^{(b)}\!\!\!}^1 R^{(i,b+1)}_{a^i_j a^b_c}\big(u^{(i)}_j\!-u^{(b)}_c\!+\tfrac{i-b}2\big) \!\Bigg) \Bigg) \Bigg) \el 
& \qq\qq \times (e_{n,n-1})^{\ot m^{(n-1)}} \ot \cdots \ot (e_{21})^{\ot m^{(1)}} \Bigg]
 \cdot  \eta ,
}
where the trace is taken over the space $\ol{V} := V_{a^{n-1}_{m^{(n-1)}}} \ot \cdots \ot  V_{a_1^1} \cong (\C^n)^{\ot \ol{m}}$ with $\ol{m} = \sum_{i=1}^{n-1} m^{(i)}$, vector $\eta$ is a lowest vector in $M^{(1)}$, and the matrix operator $\wt{D}^{(i)}_{a^i_j}(u^{(i)}_j)$ is defined by
\[
\wt{D}^{(i)}_{a^i_j}(u^{(i)}_j) = \sum_{c,d=i}^n e_{cd} \ot \big[\hat{D}^{(i)}_{a^i_j}(u^{(i)}_j)\big]_{cd} \in \End(V_{a_j^i}) \ot \tBnr[[(u^{(i)}_j)^{-1}]].
\]
\end{thrm}

\begin{proof}
We will make use of both the full expression of the Bethe vector,
\[
\Phi^{(1)}\big(\bm u^{(1\dots n-1)}\big) 
= \prod_{i=1}^{n-1}\prod_{j=1}^{m_i} \hat{B}^{(i)}_{a^{i}_j \bm a^{1 \dots i-1}}\big(u^{(i)}_j;\bm u^{(1 \dots i-1)}\big) \cdot \eta^{(n)},
\]
as well as a recursive definition,
\[
\Phi^{(i)}\big(\bm u^{(i\dots n-1)};\bm u^{(1\dots i-1)}\big)
= \prod_{j=1}^{m_i}\hat{B}^{(i)}_{a^{i}_j \bm a^{1 \dots i-1}}\big(u^{(i)}_j;\bm u^{(1 \dots i-1)}\big) \cdot \Phi^{(i+1)}\big(\bm u^{(i+1\dots n-1)};\bm u^{(1\dots i)}\big).
\]
Recall that the auxiliary space labelled by $a^i_j$ is a copy of $\C^{n-i}$. 
Our first step, however, is to reconsider each such space as embedded within a copy of $\C^n$ . 
As a slight abuse of notation, we use the same labelling for these spaces.
With this, we are able to write the following,
\aln{
\hat{B}^{(i)}_{a^{i}_j \bm a^{1 \dots i-1}}(u^{(i)}_j;\bm u^{(1 \dots i-1)})
&=
\sum_{k=i+1}^n e^*_{k} \ot [\hat{D}^{(i)}_{b \bm a^{1 \dots i-1}}(u^{(i)}_j+\tfrac12;\bm u^{(1 \dots i-1)})]_{ik}
\\
&=
e_i^*\, \hat{D}^{(i)}_{a^i_j \bm a^{1 \dots i-1}}(u^{(i)}_j+\tfrac12;\bm u^{(1 \dots i-1)}) - e_i^* \ot \hat{a}^{(i)}_{\bm a^{1 \dots i-1}}(u^{(i)}_j;\bm u^{(1 \dots i-1)}).
}
Acting with $e_i^* \in V_{a^i_j}^*$ on the level-$(i+1)$ Bethe vector, we have
\equ{ \label{tf_BV_vanish}
(e_i^*)_{a^i_j} \, \, \Phi^{(i+1)}\big(\bm u^{(i+1 \dots n-1)};\bm u^{(1\dots i)}\big)=0,
}
as $\Phi^{(i+1)}(\bm u^{(i+1 \dots n-1)};\bm u^{(1 \dots i)})$ belongs to the vector subspace in which each $V_{a^i_j}$ is treated as a copy of $\C^{n-i}$.
As such, we may write the following expression for the level-$1$ Bethe vector
\[
\Phi^{(1)}(\bm u^{(1\dots n-1)}) 
= (e_{n-1}^*)^{\ot m_{n-1}} \ot \cdots \ot (e_1^*)^{\ot m_1} \,
\prod_{i=1}^{n-1}\prod_{j=1}^{m_i} \hat{D}^{(i)}_{a^i_j \bm a^{1 \dots i-1}}(u^{(i)}_j+\tfrac12;\bm u^{(1 \dots i-1)})
\cdot \eta^{(n)}.
\]
From this expression, it is relatively simple to recover the form of the trace formula given in \cite{BeRa1}.
Consider now the expression, which we denote by $\Psi^{(i)}\big(\bm u^{(i\dots n-1)};\bm u^{(1 \dots i-1)}\big)$, obtained from vector $\Phi^{(i)}\big(\bm u^{(i\dots n-1)};\bm u^{(1 \dots i-1)}\big)$ by inserting a permutation operator between the first two excitations as follows,
\aln{
\Psi^{(i)}\big(\bm u^{(i\dots n-1)};\bm u^{(1 \dots i-1)}\big) &:= 
(e^*_i)^{\ot m_i} \, \hat{D}^{(i)}_{a^i_1}(u^{(i)}_1+\tfrac12;\bm u^{(1 \dots i-1)})
P_{a^i_1 a^i_2}
\hat{D}^{(i)}_{a^i_2}(u^{(i)}_2+\tfrac12;\bm u^{(1 \dots i-1)})
\\
&\qq\qq\qu \times \prod_{j=3}^{m_i} \hat{D}^{(i)}_{a^i_j}(u^{(i)}_j+\tfrac12;\bm u^{(1 \dots i-1)})
\cdot \Phi^{(i+1)}\big(\bm u^{(i+1 \dots n-1)};\bm u^{(1 \dots i)}\big).
}
This leads to the following sequence of equalities,
\aln{
\Psi^{(i)}\big(\bm u^{(i\dots n-1)};\bm u^{(1 \dots i-1)}\big) &= 
(e^*_i)^{\ot m_i} \, P_{a^i_1 a^i_2} 
\hat{D}^{(i)}_{a^i_2}(u^{(i)}_1+\tfrac12;\bm u^{(1 \dots i-1)})
\hat{D}^{(i)}_{a^i_2}(u^{(i)}_2+\tfrac12;\bm u^{(1 \dots i-1)})
\\
&\qq\qq\qu \times \prod_{j=3}^{m_i} \hat{D}^{(i)}_{a^i_j}(u^{(i)}_j+\tfrac12;\bm u^{(1 \dots i-1)})
\cdot \Phi^{(i+1)}\big(\bm u^{(i+1 \dots n-1)};\bm u^{(1 \dots i)}\big)
\\
&= (e^*_i)^{\ot m_i} \,
\hat{D}^{(i)}_{a^i_2}(u^{(i)}_1+\tfrac12;\bm u^{(1 \dots i-1)})
\hat{D}^{(i)}_{a^i_2}(u^{(i)}_2+\tfrac12;\bm u^{(1 \dots i-1)}) 
\\
&\qq\qq\qu \times \prod_{j=3}^{m_i} \hat{D}^{(i)}_{a^i_j}(u^{(i)}_j+\tfrac12;\bm u^{(1 \dots i-1)})
\cdot \Phi^{(i+1)}\big(\bm u^{(i+1 \dots n-1)};\bm u^{(1 \dots i)}\big).
}
The expression subsequently vanishes, due to \eqref{tf_BV_vanish}.
Such an argument applies to any permutation of the spaces $V_{a^i_j}$, provided each permutation operator is inserted between the $D$ operators associated with the auxiliary spaces on which it acts. 
Therefore, we may insert $R$ matrices at these positions in the formula for the Bethe vector,
\aln{
\Phi^{(1)}(\bm u^{(1\dots n-1)}) 
&= f(\bm u^{(1\dots n-1)})(e_1^*)^{\ot m_1} \ot \cdots \ot (e_{n-1}^*)^{\ot m_{n-1}} \,
\\
&\qu\times
\prod_{i=1}^{n-1}\prod_{j=1}^{m_i} \Bigg( \prod_{k=1}^{j-1} R^{(i)}_{a^i_j a^i_k}(u^{(i)}_j+u^{(i)}_k+1+\rho) \Bigg) \hat{D}^{(i)}_{a^i_j \bm a^{1 \dots i-1}}(u^{(i)}_j+\tfrac12;\bm u^{(1 \dots i-1)})
\cdot \eta^{(n)},
}
where $f(\bm u^{(1 \dots n-1)})=\prod_{i=1}^{n-1}\prod_{j=1}^{m_i} \prod_{k=1}^{j-1} \frac{u^{(i)}_j+u^{(i)}_k+1+\rho}{u^{(i)}_j+u^{(i)}_k+\rho}$. 

From here we insert the expression for $\hat{D}^{(i)}_{a^i_j \bm a^{1 \dots i-1}}(u^{(i)}_j+\tfrac12;\bm u^{(1 \dots i-1)})$ from \eqref{Bnr-D3},
\aln{ 
\hat{D}^{(i)}_{a^i_j {\bm a}^{1\dots i-1}}\big(u^{(i)}_j;\bm u^{(1\dots i-1)}\big) & \equiv \left(\prod_{b=i-1}^1 \prod_{c=1}^{m_{b}} R^{(i,b+1)}_{a^i_j a^{b}_c}\big(u^{(i)}_j+u^{(b)}_c+\tfrac{i-1-b}{2}+\rho\big) \right)\el & \qu \times \hat{D}^{(i)}_{a^i_j}(u^{(i)}_j) \left(\prod_{b=1}^{i-1} \prod_{c=m_{b}}^{1} R^{(i,b+1)}_{a^i_j a^{b}_c}\big(u^{(i)}_j-u^{(b)}_c+\tfrac{i-1-b}{2}\big) \right) ,
}
yielding
\aln{
\Phi^{(1)}(\bm u^{(1\dots n-1)}) 
&= f(\bm u^{(1\dots n-1)})(e_1^*)^{\ot m_1} \ot \cdots \ot (e_{n-1}^*)^{\ot m_{n-1}} \,
\\
&\qu \times
\Bigg(\prod_{i=1}^{n-1}\prod_{j=1}^{m_i} \Bigg( \Bigg( \prod_{k=1}^{j-1} R^{(i)}_{a^i_j a^i_k}(u^{(i)}_j\!+u^{(i)}_k\!+1+\rho) \Bigg)  \left(\prod_{b=i-1}^1 \prod_{c=1}^{m_{b}} R^{(i,b+1)}_{a^i_j a^{b}_c}\big(u^{(i)}_j\!+u^{(b)}_c\!+\tfrac{i-b}{2}+\rho\big) \right) \\
& \qu \times \hat{D}^{(i)}_{a^i_j}(u^{(i)}_j+\tfrac12) \left(\prod_{b=1}^{i-1} \prod_{c=m_{b}}^{1} R^{(i,b+1)}_{a^i_j a^{b}_c}\big(u^{(i)}_j-u^{(b)}_c+\tfrac{i-b}{2}\big) \right)\Bigg) \Bigg)
\cdot \eta^{(n)}.
}
From here the result is obtained simply by noting, for a matrix $M$, $e_i^* \, M \, e_j = \tr( e_{ji} \, M)$.
\end{proof}


\section{Algebraic Bethe ansatz for a $\TX$-chain} \label{sec:NABA-X}

This section contains our main results. We study a spectral problem in the space
\equ{
M := L(\la^{(1)})_{c_1} \ot L(\la^{(2)})_{c_2} \ot \ldots \ot L(\la^{(\ell)})_{c_\ell} \ot V(\mu) , \label{M}
}
where each $L(\la^{(i)})_{c_i}$ is viewed as a lowest weight $X(\mfg_{2n})$-module obtained by the restriction described in Proposition \ref{P:X-rep} with $k$ set to $k_i$, $c$ to $c_i$ and $\la$ to $\la^{(i)}$, and $V(\mu)$ is the one-dimensional $\TX$-module described in Proposition \ref{P:1-dim}.
The generating matrix $S(u)$ (cf.~\eqref{SS}) of $\TX$ acts on this space as
\equ{
S(u)\cdot M = g(u)\Bigg( \prod_{i=1}^{\ell} \mc{L}_i(u-c_i-\tfrac\ka2) \Bigg) K(u)\, \Bigg( \prod_{i=\ell}^{1} \mc{L}^t_i(\tu-c_i-\tfrac\ka2) \Bigg) M , \label{S(u).M}
}
where $\mc{L}_i(u)$ are the fused Lax operators of $X(\mfg_{2n})$, $K(u)$ are given by Lemma \ref{L:1-dim}, and $\tu = \ka-u-\rho$.
By Proposition \ref{P:L*V}, the space $M$ is a lowest weight $\TX$-module of weight $\ga(u)$ with components defined by \eqref{tga(u)} with $\mu_i(u)$ as in Proposition \ref{P:1-dim} and $\la_i(u)$ given by
\equ{
\la_i(u) = \prod_{j=1}^\ell \la^{(j)}_i(u)   \label{L-weights}
}
with weights $\la^{(j)}_i(u)$ as in \eqref{l(u)-fused}.
We say that $M$ is the (full) quantum space of a $\TX$-chain, a~$\mfg_{2n}$-symmetric open spin chain with (trivial left and non-trivial right) diagonal boundary conditions.
The image of $S(u)$ on $M$ given by \eqref{S(u).M} is the \emph{monodromy matrix} of the open spin chain.

Our approach to the spectral problem of the $\TX$-chain is as follows. We start by defining \emph{creation operators} that create \emph{top-level excitations} in the quantum space. 
They correspond to the root vectors associated to the roots in $\Phi^+\backslash \Phi^{\prime+}$, where $\Phi^+$ is the set of positive roots of $\mfg_{2n}$ and $\Phi^{\prime+}$ is the set of positive roots of the canonical $\mfgl_n\subset \mfg_{2n}$. We then determine the exchange relations for creation operators. The exchange relations allow us to identify the nested monodromy matrix which turns out to be a monodromy matrix for a $\Bnr$-chain.
Our main results are Theorems \ref{T:spn-spec} and \ref{T:son-spec}, which state eigenvectors, their eigenvalues and Bethe equations for a $\TX$-chain with the quantum space $M$ is symplectic and orthogonal cases, respectively.

\nc{\Qt}{Q_{\at_1 a}}
\nc{\Q}{Q_{a_1 a}}


\subsection{Creation operator for a single excitation}

We begin by reinterpreting the $B$ operator of the generating matrix $S(u)$, viz.~\eqref{block}, as a row vector in two auxiliary spaces, $V^*_{\at_1} \ot V^*_{a_1} \cong (\C^n)^* \ot (\C^n)^*$, with components given by the matrix elements $\scb_{ij}(u)$.

\begin{defn} \label{D:beta} 
The creation operator for a single (top-level) excitation is given by
\equ{
\beta_{\at_1 a_1}(u) := \sum_{i,j=1}^{N} e^*_i \ot e^*_j \ot \scb_{\bar \imath j}(u) \in V^*_{\at_1} \ot V^*_{a_1} \ot \TX((u^{-1})). \label{beta}
}
\end{defn}

The exchange and symmetry relations involving the $B$ operator may now be rewritten using the above notation.
In general, we may switch between the two notations using the following relation, in matrix elements,
\eqa{ \label{XBY} 
(X_a \,B_a(u)\, Y_a)_{ij} = \sum_{1\le k,l \le n} x_{ik}\, \scb_{kl}(u)\, y_{lj} = (\beta_{\at a}(u)\, X^t_{\at}\, Y_{a})_{\bar \imath j} ,
}
where $X,Y$ are matrix operators with entries in $\C((u^{-1}))$ and may act nontrivially on the additional auxiliary spaces. 
The Lemma below states some useful properties of the creation operator.

\begin{lemma} \label{L:betaex}
The (top-level) creation operator satisfies the following two identities:
\ali{
&\beta_{\at_1 a_1}(u_1) \, \beta_{\at_2 a_2}(u_2) \, R_{a_1 \at_2}(-u_1-u_2-\rho) \, \check{R}_{\at_1 \at_2}(u_1-u_2) \el
& \hspace{2.77cm} = \beta_{\at_1 a_1}(u_2) \, \beta_{\at_2 a_2}(u_1) \, R_{a_1 \at_2}(-u_1-u_2-\rho) \, \check{R}_{a_1 a_2}(u_1-u_2) , \label{betabeta} 
\\[.5em]
&\beta_{\at_1 a_1}(v) \, \Qt \, \Q = \bigg(\mp1-\frac1{v-\tv}\bigg)\beta_{\at_1 a_1}(\tv) \, \Q + \frac{\beta_{\at_1 a_1}(v) \, \Q}{v-\tv}. \label{symm_beta}
}

\end{lemma}

\begin{proof}
The operator $B(u)$ satisfies the same exchange relation as the equivalent operator in \cite{GMR}, with an additional shift of $\ka$. 
Following Lemma~3.2 in \cite{GMR}, with $\rho$ replaced by $\rho-\ka$, we arrive at \eqref{betabeta}.
To prove \eqref{symm_beta} we work from \eqref{symmBB}. Acting from the right by $Q_{a_1a}$, and using the equalities $X^t_{a_1}Q_{a_1 a} = X_{a}Q_{a_1 a} = P_{a_1 a} X_{a_1} Q_{a_1 a}$, we obtain
\[
P_{a_1 a} \, B_{a_1}(v) \, \Q = \bigg(\mp1-\frac1{v-\tv}\bigg)B_{a_1}(\tv) \, \Q + \frac{B_{a_1}(v) \, \Q}{v-\tv}.
\]
Implementing \eqref{XBY} then yields the desired result.
\end{proof}


\subsection{The AB exchange relation for a single excitation}

Our next step is to rewrite the AB exchange relation \eqref{RE:AB} in terms of the creation operator \eqref{beta}. 
Typically, the Bethe ansatz method would also require us to consider the DB exchange relation. 
However, the operators $A$ and $D$ will enter the exchange relations under a trace. 
Because of this the symmetry relation \eqref{symmAD} will allow us to rewrite all the relations involving $D$ operators in terms of the $A$ operators only. 
Indeed, the Lemma below allows us to rewrite the trace of the monodromy matrix $\tr S(u) = \tr A(u) + \tr D(u)$ in terms of the $A$ operator only.
We will often make use of the following notation.
For a function $f$ we define a symmetrization operation by
\equ{
\label{u-symm} \{f(u)\}^u := f(u) + f(\tu). 
}
We also introduce the rational function
\equ{
p(u)=\frac1{u-\tu} . \label{p(u)}
}

\begin{lemma} \label{L:tra_symm} 
We have
\[
-\frac{\tr S(\tu)}{u-\tu+\ka} = \frac{\tr S(u)}{u-\tu-\ka} = \left\{ p(u) \tr A(u) \right\}^u . 
\]
\end{lemma}

\begin{proof}
Adding $A(u)$ to both sides of \eqref{symmAD} and taking the trace we obtain
\[
\tr S(u) = \bigg(1\pm \frac1{u-\tu} \bigg)\tr (A(u)-A(\tu)) - \frac{(\ka\pm1)\tr S(u)}{u-\tu-\ka}.
\]
Rearranging this, and dividing by $(u-\tu\pm1)$, we find
\[
\frac{\tr S(u)}{u-\tu-\ka} = \frac{\tr (A(u)-A(\tu))}{u-\tu},
\]
which, by \eqref{u-symm}, proves the second equality. The first equality is obtained by sending $u \mapsto \tu$, and noting that the r.h.s.~remains unchanged. 
\end{proof}

Applying Lemma \ref{L:tra_symm} to \eqref{symmAD} we obtain a new symmetry relation for the $A$ and $D$ operators:
\eqa{ \label{symmAA}
D^t(u) = -\bigg(1\pm\frac1{u-\tu}\bigg)A(\tu) \pm \frac{A(u)}{u-\tu}
- \bigg\{ \frac{\tr(A(u))\cdot I}{u-\tu} \bigg\}^u .
}
This symmetry relation allows us to obtain a $D$-independent form of the AB exchange relation.

\begin{lemma} \label{L:AB_half}
The AB exchange relation \eqref{RE:AB} may be equivalently written as
\equ{ \label{AB_half}
A_a(v) \,\beta_{\at_1 a_1}(u) = 
\beta_{\at_1 a_1}(u) \, S^{(1)}_{a\,\at_1 a_1}(v;u) + U^+ + U^-,
}
where
\gat{
S^{(1)}_{a\,\at_1 a_1}(v;u) := R^t_{\at_1 a}(u-v) \, R^t_{a_1 a}(\tu-v)\, A_a(v) \, R^t_{a_1 a}(u-v\pm1) \, R^t_{\at_1 a}(\tu-v\pm1), \label{nmono1}
\\
U^+ := \frac{\beta_{\at_1 a_1}(v)}{u-v} \, \Qt \, R^t_{a_1 a}(\tu-u) \, A_a(u) \, R^t_{a_1 a}(\pm1) \, R^t_{\at_1 a}(\tu-u\pm1), \label{U_plus}
\\
U^- := \pm \frac{\beta_{\at_1 a_1}(v)}{u-\tv} \, \Qt \Q  \bigg(1\pm\frac{1}{u-\tu} \bigg) A_a(\tu)\, R^t_{a_1 a}(u-\tu\pm1) \, R^t_{\at_1 a}(\pm1). \label{U_minus}
}
\end{lemma}

The matrix $S^{(1)}_{a\,\at_1 a_1}(v;u)$ is the nested monodromy matrix for a single excitation. The matrices $U^\pm$ are the ``unwanted terms'' (written in our equations as ``$UWT$'').

\begin{proof}[Proof of Lemma \ref{L:AB_half}]
The first step is to rewrite \eqref{RE:AB} in terms of $\beta_{\at_1 a_1}(u)$. 
We obtain, using \eqref{XBY},
\aln{
& A_a(v)\, \beta_{\at_1 a_1}(u)\, R^t_{\at_1 a}(u+v+\rho) \, R_{a_1 a}^t(\ka-u+v) \\
&\qu = \beta_{\at_1 a_1}(u)\, R^t_{\at_1 a}(u-v)\, R^t_{a_1 a}(\tu-v)\, A_a(v) \\
&\qq  - \beta_{\at_1 a_1}(v) \Qt \, P_{a_1 a} \, R^t_{a_1 a}(\tu-v)\, A_{a_1}(u)\, U_{a_1 a}(u-v) \\
&\qq  - \beta_{\at_1 a_1}(v) \Qt \, P_{a_1 a} U_{a_1a}(u+v+\rho) \, D_{a_1}(u)\, R^t_{a_1 a}(\ka-u+v).
}
Since $Q$ is a rank $n$ projector, $R^t(u)$ is invertible for $u \neq n$, with inverse $R^t(n-u) = R^t(\ka-u\pm1)$. Multiplying the expression above by the appropriate inverses, we have
\ali{
&A_a(v)\, \beta_{\at_1 a_1}(u) \el 
& \qu =\beta_{\at_1 a_1}(u)\, R^t_{\at_1 a}(u-v)\, R^t_{a_1 a}(\tu-v)\, A_a(v) \, R^t_{a_1 a}(u-v \pm 1) \, R^t_{\at_1 a}(\tu-v \pm1)
\el
&\qq - \beta_{\at_1 a_1}(v) \Qt \, P_{a_1 a} \, R^t_{a_1 a}(\tu-v)\, A_{a_1}(u)\, U_{a_1 a}(u-v) R^t_{a_1 a}(u-v \pm 1) \, R^t_{\at_1 a}(\tu-v \pm1)
\el
&\qq - \beta_{\at_1 a_1}(v) \Qt \, P_{a_1 a} U_{a_1a}(u+v+\rho) \, D_{a_1}(u) \, R^t_{\at_1 a}(\tu-v \pm1)
\el
&\qu = \beta_{\at_1 a_1}(u)\, S^{(1)}_{a\,\at_1 a_1}(v;u) + U^A + U^D, \label{AB_initial}
}
where 
\gan{
U^A := -\beta_{\at_1 a_1}(v) \,\Qt \, P_{a_1 a} \, R^t_{a_1 a}(\tu-v)\, A_{a_1}(u) U_{a_1 a}(u-v) R^t_{a_1 a}(u-v \pm 1) \, R^t_{\at_1 a}(\tu-v \pm1) ,
\\
U^D := -\beta_{\at_1 a_1}(v)\, \Qt \, P_{a_1 a} U_{a_1a}(u+v+\rho) \, D_{a_1}(u) \, R^t_{\at_1 a}(\tu-v \pm1) .
}
So far the first term, the ``wanted term'', matches the desired expression \eqref{AB_half}. We must now manipulate $U^A+U^D$ to match the remaining terms. 
First, note that
\aln{ U(u-v) \, R^t(u-v\pm 1) &= \bigg(-\frac{P}{u-v} \pm\frac{Q}{u-v-\ka} \bigg) \bigg(1-\frac{Q}{u-v\pm1} \bigg) 
\\
&= -\frac{P}{u-v} + \bigg( \frac1{(u-v)(u-v\pm1)}  \pm \frac1{u-v-\ka} \bigg(1-\frac{\ka\pm1}{u-v\pm1} \bigg) \bigg) Q
\\
&= -\frac{P}{u-v} + \bigg( \frac1{(u-v)(u-v\pm1)}  \pm \frac1{u-v \pm1}  \bigg) Q
\\
&= -\frac{P\,R^t(\pm1)}{u-v}.
}
Thus,
\ali{ 
U^A &= \frac{\beta_{\at_1 a_1}(v)}{u-v}\, \Qt \, P_{a_1 a} \, R^t_{a_1 a}(\tu-v)\, A_{a_1}(u)\, P_{a_1 a}\,R^t_{a_1 a}(\pm1) \, R^t_{\at_1 a}(\tu-v \pm1) \nn
\\
&=\frac{\beta_{\at_1 a_1}(v)}{u-v}\, \Qt \, R^t_{a_1 a}(\tu-v)\, A_{a}(u)\, R^t_{a_1 a}(\pm1) \, R^t_{\at_1 a}(\tu-v \pm1) . \label{UA_start}
}
With $U^D$, our strategy will be to use the symmetry relation \eqref{symmAA}, allowing us to combine the term with $U^A$. 
We first make the following manipulations in preparation:
\aln{
U^D &= -\beta_{\at_1 a_1}(v) \Qt \, P_{a_1 a} U_{a_1a}(u+v+\rho) \, R^t_{\at_1 a}(\tu-v \pm1) D_{a_1}(u)
\\
&= -\beta_{\at_1 a_1}(v) \Qt \, \bigg(-\frac1{u+v+\rho} \pm \frac{\Q}{u-\tv}\bigg) \, \bigg(1+\frac{\Qt}{u-\tv \mp1} \bigg)  D_{a_1}(u)
\\
&= - \beta_{\at_1 a_1}(v)  \, \bigg(-\frac{\Qt}{u+v+\rho} \pm \frac{\Qt\,\Q}{u-\tv} + \frac{\Qt}{u-\tv \mp1} \, \Big(-\frac{\ka\pm1}{u+v+\rho} \pm \frac{1}{u-\tv}\Big)  \bigg) D_{a_1}(u).
}
The final term may be factorised as follows:
\[
-\frac{\ka\pm1}{u+v+\rho} \pm \frac{1}{u-\tv}=-\frac{\ka(u-\tv\mp1)}{(u+v+\rho)(u-\tv)}.
\]
Thus,
\aln{
U^D &= - \beta_{\at_1 a_1}(v)  \bigg(-\frac{\Qt}{u+v+\rho} \pm \frac{\Qt\,\Q}{u-\tv} - \frac{\ka \, \Qt }{(u+v+\rho)(u-\tv)} \bigg)  D_{a_1}(u)
\\
&= - \beta_{\at_1 a_1}(v)  \bigg(\pm \frac{\Qt\,\Q}{u-\tv} - \frac{\Qt }{u-\tv} \bigg) D_{a_1}(u)
\\
&= - \frac{\beta_{\at_1 a_1}(v)}{u-\tv}  \big(\pm \Qt\,\Q - \Qt  \big) D_{a_1}(u).
}
Although this expression is now a lot simpler, the $D_{a_1}(u)$ operator is acting on the auxiliary space $V_{a_1}$, rather than $V_a$ as desired. 
To remedy this, we use the following identity:
\aln{
\Qt D_{a_1}(u) = \Qt \tr_a(\Q D_{a_1}(u)) = \Qt \tr_a(\Q D^t_{a}(u)) = \Qt \Q D^t_{a}(u)\, \Qt ,
}
where we have used that $\tr_a(\Q) = I_{a_1}$. Therefore, using also $Q_{a_1 a} D_{a_1}(u) = Q_{a_1 a}D_{a}^t(u)$,
\aln{
U^D = \mp \frac{\beta_{\at_1 a_1}(v)}{u-\tv}  \Qt \Q D^t_a(u)\,(1\mp \Qt) .
}
Applying the symmetry relation \eqref{symmAA} we obtain
\aln{
U^D = \mp \frac{\beta_{\at_1 a_1}(v)}{u-\tv}  \Qt \Q \bigg( -\bigg(1 \pm \frac1{u-\tu} \bigg)A_a(\tu)  \pm \frac{A_a(u)}{u-\tu} - \bigg\{ \frac{\tr(A(u))\cdot I_a}{u-\tu}\bigg\}^u\bigg)(1\mp \Qt) .
}
Since all terms in $U^A+U^D$ contain $A_a(u)$ or $A_a(\tu)$, we reorganise the sum $U^A+U^D$ accordingly.
Define 
\aln{
U^+ &:= \frac{\beta_{\at_1 a_1}(v)}{u-v} \,\Qt \, R^t_{a_1 a}(\tu-v)\, A_{a}(u)\, R^t_{a_1 a}(\pm1) \, R^t_{\at_1 a}(\tu-v \pm1)  
\\
&\hspace{3.5cm}- \frac{\beta_{\at_1 a_1}(v)}{u-\tv}   \Qt \Q \bigg(  \frac{A_a(u)\mp\tr(A(u))}{u-\tu} \bigg)(1\mp \Qt),
\\
U^- &:= \pm  \frac{\beta_{\at_1 a_1}(v)}{u-\tv}   \Qt \Q \bigg( \bigg(1\pm\frac{1}{u-\tu} \bigg)A_a(\tu)  - \frac{\tr(A(\tu))}{u-\tu}\bigg)(1\mp \Qt),
}
so that $U^A + U^D = U^+ + U^-$. 
It remains to match the expressions for $U^+$ and $U^-$ with those in the desired expressions \eqref{U_plus}, \eqref{U_minus}.
With $U^-$, we simply use $\Q \, \tr(A(\tu)) = \Q \, A_a(\tu) \, \Q$ to obtain the required form,
\[
U^-=\pm \frac{\beta_{\at_1 a_1}(v)}{u-\tv} \,\Qt \Q  \bigg(1\pm\frac{1}{u-\tu} \bigg) A_a(\tu)\, R^t_{a_1 a}(u-\tu\pm1) \, R^t_{\at_1 a}(\pm1).
\]
We now turn our attention to $U^+$. Using again the trace property of $\Q$,
\aln{
U^+ & = \frac{\beta_{\at_1 a_1}(v)}{u-v} \,\Qt \, R^t_{a_1 a}(\tu-v)\, A_{a}(u)\, R^t_{a_1 a}(\pm1) \, R^t_{\at_1 a}(\tu-v \pm1) \\
&\hspace{3cm}- \frac{\beta_{\at_1 a_1}(v)}{u-\tv} \Qt \Q \frac{A_a(u)}{u-\tu} R^t_{a_1 a}(\pm1)(1\mp \Qt).
}
The $\Qt$ and $R_{a_1 a}^t(\pm1)$ matrices are present in both terms as desired. The simplest way forward is to expand the remaining matrices in terms of projectors, then match term by term. Indeed,
\aln{
U^+ &= \beta_{\at_1 a_1}(v)\,\Qt \bigg( \frac1{u-v}A_a(u) R_{a_1 a}^t(\pm1)
+ \frac1{u-\tv}\bigg(\frac1{u-v} - \frac1{u-\tu}\bigg) \Q A_a(u) R_{a_1 a}^t(\pm1)
\\
&\hspace{1cm}  + \frac{A_a(u)\, R^t_{a_1 a}(\pm1) \, \Qt}{(u-v)(u-\tv\mp1)}  \pm \frac1{u-\tv}\bigg( \frac1{u-\tu} \pm \frac1{(u-v)(u-\tv\mp1)}\bigg) \Q A_a(u) R^t_{a_1 a}(\pm1) \Qt  \bigg) 
\\
&= \beta_{\at_1 a_1}(v)\, \Qt \bigg( \frac1{u-v}A_a(u) R_{a_1 a}^t(\pm1)
+ \Q A_a(u) R_{a_1 a}^t(\pm1)
\\
&\hspace{2.7cm} + \frac{A_a(u)\, R^t_{a_1 a}(\pm1) \, \Qt}{(u-v)(u-\tv\mp1)}  \pm \frac{u-v\pm1}{(u-\tu)(u-v)(u-\tv\mp1)} \Q A_a(u) R^t_{a_1 a}(\pm1) \Qt  \bigg).
}
Although all the terms have been fully written out, it is still not clear that this is equal to the desired expression. The discrepancy arises due to the terms on the last line. These terms contain two $\Qt$ operators, and so elements sandwiched between these operators appear as a trace. This leads to the following identities
\eqa{ \label{trace_ids}
\Qt \Q A_a(u) \Q \Qt = \Qt A_a(u) \Qt ,\\
\Qt \Q A_a(u) \Qt = \Qt A_a(u) \Q \Qt .
}
Then, expanding the $R^t_{a_1 a}(\pm1)$ matrices,
\aln{
U^+ &= \frac{\beta_{\at_1 a_1}(v)}{u-v} \Qt R^t_{a_1 a}(\tu-u) A_a(u) R_{a_1 a}^t(\pm1)
\\
&\qu+ \beta_{\at_1 a_1}(v) \big( \al_1 \Qt \Q A_a(u) \Qt + \al_2 \Qt \Q A_a(u) \Q \Qt \big),
}
where we find
\[
\al_1 = \mp\al_2 =  \mp \frac1{(u-v)(u-\tv\mp1)}
\pm \frac{u-v\pm1}{(u-\tu)(u-v)(u-\tv\mp1)} = \mp \frac{1}{(u-v)(u-\tu)}.
\]
So
\aln{
U^+ &= \frac{\beta_{\at_1 a_1}(v)}{u-v} \bigg( \Qt R^t_{a_1 a}(\tu-u) A_a(u) R_{a_1 a}^t(\pm1) 
\\
&\hspace{2cm}\mp\frac{\Qt \Q A_a(u) \Qt}{u-\tu} + \frac{\Qt \Q A_a(u) \Q \Qt}{u-\tu}\bigg) .
}
Writing 
\[
\frac1{u-\tu}=\frac1{u-\tu}\cdot \frac{u-\tu\mp1}{u-\tu\mp1}
= \frac{1}{u-\tu\mp1} \bigg( 1\mp \frac1{u-\tu}\bigg),
\]
and using \eqref{trace_ids} we obtain
\aln{
U^+ &= \frac{\beta_{\at_1 a_1}(v)}{u-v} \bigg( \Qt R^t_{a_1 a}(\tu-u) A_a(u) R_{a_1 a}^t(\pm1) \mp \frac{\Qt A_a(u) \Q \Qt}{u-\tu\mp1}
\\
&\hspace{2cm}  + \frac{\Qt \Q A_a(u) \Qt}{(u-\tu)(u-\tu\mp1)} 
+ \frac{\Qt A_a(u) \Qt}{u-\tu\mp1} \mp \frac{\Qt \Q A_a(u) \Q \Qt}{(u-\tu)(u-\tu\mp1)}\bigg)
\\
&= \frac{\beta_{\at_1 a_1}(v)}{u-v} \Qt \, R^t_{a_1 a}(\tu-u) \, A_a(u) \, R^t_{a_1 a}(\pm1) \, R^t_{\at_1 a}(\tu-u\pm1),
}
which matches \eqref{U_plus}, as required.
\end{proof}

From Lemma~\ref{L:AB_half}, to obtain most elegant form of the unwanted terms, written as a residue of the ``wanted term'', we must symmetrise over $v \rightarrow \tv$. We will employ the notation \eqref{u-symm}.

\begin{lemma} \label{L:AB_full}
The AB exchange relation for a single excitation is
\aln{
\big\{p(v)\, A_a(v) \big\}^v \beta_{\at_1 a_1}(u) &= 
\beta_{\at_1 a_1}(u) \big\{p(v)\,S^{(1)}_{a\,\at_1 a_1}(v;u)\big\}^v \\
& \qu +\frac1{p(u)}\bigg\{p(v)\,\frac{\beta_{\at_1 a_1}(v)}{u-v} \bigg\}^v \Res{w \rightarrow u} \big\{ p(w) \, S^{(1)}_{a\,\at_1 a_1}(w;u) \big\}^w .
}
\end{lemma}

\begin{proof}
Evaluating the residue, the desired expression is
\eqa{ \label{AB_to_show}
&\big\{p(v)A_a(v)\big\}^v \beta_{\at_1 a_1}(u)  
\\
&\qu =\beta_{\at_1 a_1}(u) \big\{p(v)\,R^t_{\at_1 a}(u-v) \, R^t_{a_1 a}(\tu-v) \, A_a(v) \, R^t_{a_1 a}(u-v\pm1) \, R^t_{\at_1 a}(\tu-v\pm1)\big\}^v
\\
&\qq  +\bigg\{p(v)\frac{\beta_{\at_1 a_1}(v)}{u-v} \bigg\}^v \, 
\Qt \, R^t_{a_1 a}(\tu-u) \, A_a(u) \, R^t_{a_1 a}(\pm1) \, R^t_{\at_1 a}(\tu-u\pm1)
\\
&\qq  +\bigg\{p(v)\frac{\beta_{\at_1 a_1}(v)}{u-v} \bigg\}^v 
R^t_{\at_1 a}(u-\tu) \, \Q \, A_a(\tu) \, R^t_{a_1 a}(u-\tu\pm1) \, R^t_{\at_1 a}(\pm1).
}
We will work from \eqref{AB_half}, and obtain this expression. 
Multiplying \eqref{AB_half} by $p(v)$ and symmetrising over $v\to \tv$ reveals that the ``wanted term'' and $U^+$ term are already of the correct form, while $U^-$ is of the form
\[
\big\{ p(v) \, U^- \big\}^v=\pm \bigg\{ p(v) \frac{\beta_{\at_1 a_1}(v)}{u-\tv} \bigg\}^v \Qt \Q  \bigg(1\pm\frac{1}{u-\tu} \bigg) A_a(\tu)\, R^t_{a_1 a}(u-\tu\pm1) \, R^t_{\at_1 a}(\pm1).
\]
From here, we will use the identity \eqref{symm_beta} to construct $R^t(u-\tu)$, and arrive at the desired expression. 
We must split the r.h.s.~into two portions, on one of which we will use the to construct the ``identity'' part of the $R^t$-matrix.
Combining this with the other portion will result in the desired $R^t$-matrix.
It turns out the correct proportions to take are given as follows:
\aln{
\frac1{u-\tv}\bigg(1\pm\frac1{u-\tu}\bigg) &= \frac1{u-\tv}\bigg(1\pm \frac1{u-v}\mp \frac1{u-v} \pm\frac1{u-\tu}\bigg)
\\
&=\frac1{u-\tv}\bigg(1\pm \frac1{u-v} \mp\frac{u-\tv}{(u-v)(u-\tu)}\bigg)
\\
&=\frac1{u-\tv}\bigg(1\pm \frac1{u-v}\bigg) \mp\frac{1}{(u-v)(u-\tu)}.
}
Then, 
\aln{
\big\{p(v) \, U^-\big\}^v &=\pm \bigg\{p(v) \, \frac{\beta_{\at_1 a_1}(v)}{u-v} \bigg( \frac{u-v\pm1}{u-\tv} \mp\frac{1}{u-\tu} \bigg)  \bigg\}^v \Qt \Q  A_a(\tu)\, R^t_{a_1 a}(u-\tu\pm1) \, R^t_{\at_1 a}(\pm1).
}
Applying \eqref{symm_beta} to the first of these terms, we have
\aln{
&\pm \bigg\{p(v) \, \frac{\beta_{\at_1 a_1}(v)}{u-v} \bigg( \frac{u-v\pm1}{u-\tv} \bigg) \bigg\}^v \Qt \Q
\\
&\hspace{2cm}= \pm \bigg\{ \frac{p(v)}{u-v} \bigg(\frac{u-v\pm 1}{u-\tv} \bigg) \bigg( \bigg(\mp1-\frac1{v-\tv}\bigg)\beta_{\at_1 a_1}(\tv) \, \Q + \frac{\beta_{\at_1 a_1}(v) \, \Q}{v-\tv} \bigg) \bigg\}^v
\\
&\hspace{2cm}= \pm \bigg\{ p(v)  \frac{\beta_{\at_1 a_1}(v)}{(u-v)(u-\tv)}   \bigg( (u-\tv\pm1)\bigg(\pm1-\frac1{v-\tv}\bigg)+ \frac{u-v\pm1}{u-\tv} \bigg) \bigg\}^v \Q
\\
&\hspace{2cm}=\bigg\{ p(v)  \frac{\beta_{\at_1 a_1}(v)}{(u-v)(u-\tv)}   \bigg( \frac{(u-\tv\pm1)(v-\tv\mp1)\pm (u-v\pm1)}{v-\tv} \bigg) \bigg\}^v \Q
\\
&\hspace{2cm}=\bigg\{ p(v)  \frac{\beta_{\at_1 a_1}(v)}{(u-v)(u-\tv)}   \bigg( \frac{(v-\tv)(u-\tv)}{v-\tv} \bigg) \bigg\}^v \Q
\\
&\hspace{2cm}=\bigg\{ p(v)  \frac{\beta_{\at_1 a_1}(v)}{u-v}    \bigg\}^v \Q.
}
Therefore
\eqa{ \label{U_minus_final}
\big\{p(v) \, U^-\big\}^v= \bigg\{p(v) \, \frac{\beta_{\at_1 a_1}(v)}{u-v}\bigg\}^v R^t_{\at_1 a}(u-\tu)   \Q \, A_a(\tu)\, R^t_{a_1 a}(u-\tu\pm1) \, R^t_{\at_1 a}(\pm1),
}
which agrees with the last term in \eqref{AB_to_show}.
\end{proof}

Lemmas \ref{L:AB_half} and \ref{L:AB_full} provide us with an insight into the expression for the nested monodromy matrix of the spin chain. 
The next step is to generalize the result of Lemma \ref{L:AB_full} for an arbitrary number of excitations.


\subsection{Creation operator for multiple excitations}

Choose $m \in \N$, the number of (top-level) excitations, and introduce $m$-tuple $\bm u = (u_1, u_2, \dots, u_m)$ of formal parameters and $m$-tuples $\tl{\bm a} = (\tl a_1,\dots,\tl a_m)$ and $\bm a = (a_1,\dots,a_m)$ of labels. For each label we associate an auxiliary vector space, $V_{\tl a_1}, V_{a_1}, \ldots, V_{\tl a_m}, V_{a_m}$, each isomorphic to $\C^n$. Then we define a tensor space $W_{\tl{\bm a}\bm a}$ and its dual $W_{\tl{\bm a}\bm a}^*$ by
\equ{
W_{\tl{\bm a}\bm a} := V_{\at_1} \ot V_{a_1} \ot \cdots \ot V_{\at_m} \ot V_{a_m}, \qq
W^*_{\tl{\bm a}\bm a} := V^*_{\at_1} \ot V^*_{a_1} \ot \cdots V^*_{\at_m} \ot V^*_{a_m}. \label{W-space}
}

\begin{defn} \label{D:creation-op-m}
The creation operator for $m$ (top-level) excitations is given by
\aln{
\beta_{\tl{\bm a}\bm a}(\bm u) & := \prod_{i=1}^m \bigg( \beta_{\at_i a_i}(u_i) \prod_{j=i-1}^1 R_{a_j \at_i}(\tu_i-u_j) \bigg) \in W^*_{\tl{\bm a}\bm a} \ot X(\mfg_{2n},\mfg^\rho_{2n})^{tw} ((u^{-1}_1,\dots,u^{-1}_m)).
}
\end{defn}
Note that $\beta_{\tl{\bm a}\bm a}(\bm u)$ satisfies the following recursion relation:
\aln{
\beta_{\at_1 a_1 \dots \at_m a_m}(u_1, \dots, u_m) 
&= \beta_{\at_1 a_1 \dots \at_{m-1} a_{m-1}}(u_1, \dots, u_{m-1}) \, \beta_{\at_m a_m}(u_m) 
\\
&\hspace{3cm}\times R_{a_{m-1} \at_m}(\tu_m-u_{m-1}) \cdots R_{a_1 \at_m}(\tu_m-u_1).
}
Our next step is to obtain an identity relating $\beta_{\tl{\bm a}\bm a}(\bm u)$ with $\beta_{\tl{\bm a}\bm a}(\bm u_{i\leftrightarrow i+1})$, where $\bm u_{i\leftrightarrow i+1}$ denotes the $m$-tuple obtained from $\bm u$ by interchanging $u_i$ with $u_{i+1}$ for any $1\le i \le m-1$. 
For this purpose, we define
\[
\check{R}(u) := \frac{u}{u-1}\,PR(u).
\]
The normalisation here is chosen such that $\check{R}(u)\,\check{R}(-u)=I$.

\begin{lemma} \label{L:B_symm}
The creation operator for $m$ (top-level) excitations obeys the following $\check{R}$-symmetry
\[
\beta_{\tl{\bm a}\bm a}(\bm u) = \beta_{\tl{\bm a}\bm a}(\bm u_{i \leftrightarrow i+1}) \, \check{R}_{a_i a_{i+1}}(u_i-u_{i+1}) \, \check{R}_{\at_i \at_{i+1}}(u_{i+1}-u_{i})
\]
for $1 \leq i \leq m-1$.

\end{lemma}

\begin{proof}
The operator $B(u)$ satisfies the same defining relations as those in \cite{GMR}, with an additional shift of $\ka$. Following the same argument as in Lemma~3.2 of \cite{GMR}, with $\rho-\ka$ instead of $\rho$, we arrive at the same conclusion. Finally, the normalised $\check{R}$ allows us to write $\check{R}^{-1}(u)=\check{R}(-u)$.
\end{proof}


\subsection{The AB exchange relation for multiple excitations} \label{sec:AB-multi}

We now generalise the single excitation nested monodromy matrix $S^{(1)}_{a\,\at_1 a_1}(v;u_1)$ introduced in Lemma~\ref{L:AB_full} to multiple excitations. 

\begin{defn} \label{D:nmono}
The nested monodromy matrix for $k$ (top-level) excitations is given by 
\ali{ 
S^{(1)}_{a\,\at_1 a_1 \dots \at_k a_k}(v; u_1, \dots, u_k) &:= \Bigg(\prod_{i=1}^k R^t_{\at_i a}(u_i-v) \Bigg) \Bigg( \prod_{i=1}^k R^t_{a_i a}(\tu_i-v) \Bigg) \el 
& \qq \times A_a(v) 
\Bigg( \prod_{i=k}^1 R^t_{a_i a}(u_i-v\pm1) \Bigg) 
\Bigg( \prod_{i=k}^1 R^t_{\at_i a}(\tu_i-v\pm1) \Bigg). \label{nmono}
}
\end{defn}

We will often omit the $\at_1 a_1 \dots \at_k a_k$ from the subscript, writing simply $S^{(1)}_{a}(v;  u_1, \dots, u_k)$.

\begin{lemma} \label{L:movement}
The following identity holds
\aln{ 
&S^{(1)}_a(v;u_1, \dots, u_{k-1}) \bigg( \beta_{\at_k a_k}(u_k) \prod_{j=k-1}^1 R_{a_j \at_k}(\tu_k-u_j) \bigg)
\\
&\qq = \bigg( \beta_{\at_k a_k}(u_k) \prod_{j=k-1}^1 R_{a_j \at_k}(\tu_k-u_j) \bigg) S^{(1)}_a(v;u_1, \dots, u_k) + UWT,
}
where $UWT$ denotes the ``unwanted terms'' that do not contain $A_a(v)$.
\end{lemma}

\begin{proof}
Working from the definition of $S^{(1)}_a(v;u_1, \dots, u_{k-1})$, and commuting matrices which act on different spaces, we use Lemma~\ref{L:AB_half} to obtain
\aln{ 
&S^{(1)}_a(v;u_1, \dots, u_{k-1}) \bigg( \beta_{\at_k a_k}(u_k) \prod_{j=k-1}^1 R_{a_j \at_k}(\tu_k-u_j) \bigg)
\\
&\qq =
\beta_{\at_k a_k}(u_k)
\Bigg(\prod_{i=1}^{k-1} R^t_{\at_i a}(u_i-v) \Bigg) 
\Bigg( \prod_{i=1}^{k-1} R^t_{a_i a}(\tu_i-v) \Bigg) 
\\
&\qq\qu \times R^t_{\at_k a}(u_k-v) \, R^t_{a_k a}(\tu_k-v) \, A_a(v) \, R^t_{a_k a}(u_k-v\pm1) \, R^t_{\at_k a}(\tu_k-v\pm1)
\\
&\qq\qu \times \Bigg( \prod_{i=k-1}^{1} R^t_{a_i a}(u_i-v\pm1) \Bigg) 
\Bigg( \prod_{i=k-1}^1 R^t_{\at_i a}(\tu_i-v\pm1) \Bigg)
\Bigg( \prod_{j=k-1}^1 R_{a_j \at_k}(\tu_k-u_j) \Bigg)\\ &\qq + UWT.
}
To obtain the result, we must move the rightmost product of $R$-matrices to the left, using the Yang-Baxter equation. The first move is simply commuting the rightmost product of $R$-matrices to the left, through the product of $R^t$-matrices, as there is no intersection of spaces on which these products act non-trivially. 

Next, we write
\aln{
&R^t_{\at_k a}(\tu_k-v\pm1) \Bigg( \prod_{i=k-1}^{1} R^t_{a_i a}(u_i-v\pm1) \Bigg) \Bigg( \prod_{j=k-1}^1 R_{a_j \at_k}(\tu_k-u_j) \Bigg)
\\
&\qq = \Bigg[ \Bigg( \prod_{i=1}^{k-1} R_{a_i a}(u_i-v\pm1) \Bigg) R_{\at_k a}(\tu_k-v\pm1) \Bigg( \prod_{j=k-1}^1 R_{a_j \at_k}(\tu_k-u_j) \Bigg) \Bigg]^{t_a}.
}
From here, repeated use of the Yang-Baxter equation allows us to swap the matrices on the left with those on the right. Indeed, the Yang-Baxter equation is
\[
R_{a_i a}(u_i-v\pm1) \, R_{\at_k a}(\tu_k-v\pm1) \, R_{a_i \at_k}(\tu_k-u_i) = R_{a_i \at_k}(\tu_k-u_i) \, R_{\at_k a}(\tu_k-v\pm1) \, R_{a_i a}(u_i-v\pm1).
\]
Note that after performing each swap, the $R$-matrix swapped to the left commutes with the remaining product of $R$-matrices on the left, and similarly for the $R$-matrix swapped to the right. Thus
\aln{
&R^t_{\at_k a}(\tu_k-v\pm1) \Bigg( \prod_{i=k-1}^{1} R^t_{a_i a}(u_i-v\pm1) \Bigg) \Bigg( \prod_{j=k-1}^1 R_{a_j \at_k}(\tu_k-u_j) \Bigg)
\\
&\qq = \Bigg[ \Bigg( \prod_{j=k-1}^1 R_{a_j \at_k}(\tu_k-u_j) \Bigg)  R_{\at_k a}(\tu_k-v\pm1) \Bigg( \prod_{i=1}^{k-1} R_{a_i a}(u_i-v\pm1) \Bigg) \Bigg]^{t_a}
\\
&\qq = \Bigg( \prod_{j=k-1}^1 R_{a_j \at_k}(\tu_k-u_j) \Bigg) \Bigg( \prod_{i=k-1}^{1} R^t_{a_i a}(u_i-v\pm1) \Bigg) R^t_{\at_k a}(\tu_k-v\pm1).
}
So far we have
\aln{
{\rm l.h.s.} &=
\beta_{\at_k a_k}(u_k)
\Bigg(\prod_{i=1}^{k-1} R^t_{\at_i a}(u_i-v) \Bigg) 
\Bigg( \prod_{i=1}^{k-1} R^t_{a_i a}(\tu_i-v) \Bigg) 
\\
&\qu \times R^t_{\at_k a}(u_k-v) \, R^t_{a_k a}(\tu_k-v) \, A_a(v) \, R^t_{a_k a}(u_k-v\pm1) 
\\
&\qu \times \Bigg( \prod_{j=k-1}^1 R_{a_j \at_k}(\tu_k-u_j) \Bigg)
\Bigg( \prod_{i=k-1}^{1} R^t_{a_i a}(u_i-v\pm1) \Bigg) 
\Bigg( \prod_{i=k}^1 R^t_{\at_i a}(\tu_i-v\pm1) \Bigg) \\ & +UWT.
}
Note that the product of $R$-matrices that we were moving commutes with $R^t_{a_k a}(\tu_k-v) \, A_a(v) \, R^t_{a_k a}(u_k-v\pm1)$. Then, moving further leftwards, we must use the Yang-Baxter relation again. Specifically, we use
\[
R^t_{a_i a}(\tu_i-v) \, R^t_{\at_k a}(u_k-v) \, R_{a_i \at_k}(\tu_k-u_i) = R_{a_i \at_k}(\tu_k-u_i)  \, R^t_{\at_k a}(u_k-v) \, R^t_{a_i a}(\tu_i-v),
\]
giving
\aln{
&\Bigg( \prod_{i=1}^{k-1} R^t_{a_i a}(\tu_i-v) \Bigg)  R^t_{\at_k a}(u_k-v) \Bigg( \prod_{j=k-1}^1 R_{a_j \at_k}(\tu_k-u_j) \Bigg)
\\
&\qq =  \Bigg( \prod_{j=k-1}^1 R_{a_j \at_k}(\tu_k-u_j) \Bigg) R^t_{\at_k a}(u_k-v) \Bigg( \prod_{i=1}^{k-1} R^t_{a_i a}(\tu_i-v) \Bigg).
}
Therefore,
\aln{
{\rm l.h.s.} &=
\beta_{\at_k a_k}(u_k)
\Bigg(\prod_{i=1}^{k-1} R^t_{\at_i a}(u_i-v) \Bigg) 
\Bigg( \prod_{j=k-1}^1 R_{a_j \at_k}(\tu_k-u_j) \Bigg)
 R^t_{\at_k a}(u_k-v)
\\
&\qq \times \Bigg( \prod_{i=1}^{k} R^t_{a_i a}(\tu_i-v) \Bigg) \, A_a(v) \Bigg( \prod_{i=k}^{1} R^t_{a_i a}(u_i-v\pm1) \Bigg) 
\Bigg( \prod_{i=k}^1 R^t_{\at_i a}(\tu_i-v\pm1) \Bigg) \\[.25em] & \qu +UWT
\\
&= 
\beta_{\at_k a_k}(u_k)
\Bigg( \prod_{j=k-1}^1 R_{a_j \at_k}(\tu_k-u_j) \Bigg)
\Bigg(\prod_{i=1}^{k} R^t_{\at_i a}(u_i-v) \Bigg) 
\\
&\qq \times \Bigg( \prod_{i=1}^{k} R^t_{a_i a}(\tu_i-v) \Bigg) \, A_a(v) \Bigg( \prod_{i=k}^{1} R^t_{a_i a}(u_i-v\pm1) \Bigg) 
\Bigg( \prod_{i=k}^1 R^t_{\at_i a}(\tu_i-v\pm1) \Bigg) \\[.25em] & \qu +UWT
\\
&= \beta_{\at_k a_k}(u_k)
\Bigg( \prod_{j=k-1}^1 R_{a_j \at_k}(\tu_k-u_j) \Bigg)
S^{(1)}_{a}(v;u_1, \dots, u_k) +UWT
}
as required.
\end{proof}

We may apply this result inductively to the creation operator for $m$ excitations $\beta_{\tl{\bm a}\bm a}(\bm u)$. 

\begin{crl} \label{C:wantedterm}
The AB exchange relation for multiple excitations has the form
\equ{ \label{wantedterm}
 \even{v}{p(v)\,A_a(v)} \beta_{\tl{\bm a}\bm a}(\bm u) 
= \beta_{\tl{\bm a}\bm a}(\bm u) \even{v}{p(v)\,S^{(1)}_{a}(v;\bm u)} + UWT
}
where $S^{(1)}_{a}(v;\bm u)$ is the nested monodromy matrix for $m$ excitations defined by \eqref{nmono} and $UWT$ denotes the terms that do not contain $A_a(v)$. \qed
\end{crl}


\subsection{Exchange relations for the nested monodromy matrix} \label{sec:nested}

We introduce a vector space $M^{(1)}$, called the {\it nested vacuum sector}, and a matrix $S^{(1)}_{a}(v;\bm w,\bm u)$, called the \emph{generalised nested monodromy matrix}, acting on this space, with $\bm w = (w_1, w_2, \dots, w_m)$ and $\bm u = (u_1, u_2, \dots, u_m)$ being $m$-tuples of non-zero complex parameters.
We show that $S^{(1)}_{a}(v;\bm w,\bm u)$ satisfies the defining relations of the algebra $\tBnr$ in the space $M^{(1)}$. 
This allows us to identify $S^{(1)}_{a}(v;\bm w,\bm u)$ as the monodromy matrix for the residual $\tBnr$-chain, in a suitable sense.
The space $M^{(1)}$ is then reinterpreted as the (full) quantum space of this residual chain, which we have studied in Section~\ref{sec:NABA-RE}.

For each bulk vector space $L(\lambda^{(i)})_{c_{i}}$ in \eqref{M} denote by $L^0(\lambda^{(i)})_{c_{i}}$ the subspace consisting of vectors annihilated by the operator $\ol{C}(u)$ of the generating matrix $T(u)$ of $X(\mfg_{2n})$, namely
\equ{
L^0(\lambda^{(i)})_{c_{i}} := \{\zeta \in L(\lambda^{(i)})_{c_{i}} \,:\, t_{n+k,l}(u) \cdot \zeta = 0 \; \text{ for }\; 1 \leq k,l \leq n \}. \label{L0}
}

\begin{lemma} \label{L:restrict}  
The space $L^0(\lambda^{(i)})_{c_{i}}$ is an irreducible lowest weight $Y(\mfgl_n)$-module.
\end{lemma}

\begin{proof}
Relation \eqref{Y:CA} implies that $L^0(\lambda^{(i)})_{c_{i}}$ is stable under the action of $\ol{A}(u)$. 
Then \eqref{Y:AA} allows us to view $L^0(\lambda^{(i)})_{c_{i}}$ as a $Y(\mfgl_n)$-module. 
Thus we only need to show that $L^0(\lambda^{(i)})_{c_{i}}$ is an irreducible $Y(\mfgl_n)$-module. 
Let $\xi\in L(\lambda^{(i)})_{c_{i}}$ be a lowest vector and note that $\xi\in L^0(\lambda^{(i)})_{c_{i}}$.
Set $L:=Y(\mfgl_n)\,\xi$ and note that $L\subseteq L^0(\lambda^{(i)})_{c_{i}}$. Since there are no more lowest vectors in $L^0(\lambda^{(i)})_{c_{i}}$, it follows that $L= L^0(\lambda^{(i)})_{c_{i}}$.
\end{proof}

Introduce a \emph{vacuum sector} $M^0$ of the full quantum space $M$ by 
\[
M^0 := L^0(\lambda^{(1)})_{c_{1}}  \ot \cdots \ot L^0(\lambda^{(\ell)})_{c_{\ell}}  \ot V(\mu) \subset M.
\]
The Lemma below is an analogue of Lemma 3.8 in \cite{GMR}.

\begin{lemma} \label{L:M0-stable}
The operator $C(u)$ of the matrix $S(u)$ acts by zero on the space $M^{0}$. Consequently, $M^{0}$ is stable under the action of the operator $A(u)$ of the matrix $S(u)$. \qed
\end{lemma}

Recall \eqref{W-space}. We define the level-$1$ {\it nested vacuum sector} by
\equ{ \label{nest_vac}
M^{(1)} := W_{{\tl{\bm a}\bm a}} \ot M^{0}.
}
Here an overlap of notation with $M^{(1)}$ defined in Section \ref{sec:Bnr-sp-mono} is intentional. It will be shown below that $M^{(1)}$ can be viewed as the (full) quantum space for a residual $\tBnr$-chain. 

Next, we define a generalised nested monodromy matrix which differs from the one in Definition \ref{D:nmono} by an addition $m$-tuple of complex parameters, $\bm w$. These parameters will play a prominent role in Section~\ref{sec:NABA-SO}.

\begin{defn} \label{D:gnmono}
The generalised nested monodromy matrix is defined by 
\ali{ 
S^{(1)}_{a}(v; \bm w, \bm u) &:= \Bigg(\prod_{i=1}^m R^t_{\at_i a}(u_i-v) \Bigg) \Bigg( \prod_{i=1}^m R^t_{a_i a}(w_i-v) \Bigg) \el 
& \qq \times A_a(v) 
\Bigg( \prod_{i=m}^1 R^t_{a_i a}(\tw_i-v\pm1) \Bigg) 
\Bigg( \prod_{i=m}^1 R^t_{\at_i a}(\tu_i-v\pm1) \Bigg). \label{gnmono}
}
\end{defn}

Matrix $S^{(1)}_{a}(v; \bm u)$ defined by \eqref{nmono} is recovered by setting $w_i=\tu_i$. It will be useful to know that \eqref{RK-unit} allows us to rewrite \eqref{gnmono} as
\ali{ 
S^{(1)}_{a}(v; \bm w, \bm u) &= \Bigg(\prod_{i=1}^m R^t_{\at_i a}(u_i-v) \Bigg) \Bigg( \prod_{i=1}^m R^t_{a_i a}(w_i-v) \Bigg) \el 
& \qq \times A_a(v) 
\Bigg( \prod_{i=1}^m R^t_{a_i a}(w_i+v+\rho) \Bigg)^{-1} 
\Bigg( \prod_{i=1}^m R^t_{\at_i a}(u_i+v+\rho) \Bigg)^{-1}. \label{gnmono2}
}
Set $r=0$ for types CI, DII and CD0, and $r=n-\tfrac p2$ for types DI and CII.

\begin{prop} \label{P:B(n,r)-rep} 
The mapping 
\ali{
\tBnr \to \End(M^{(1)}) \ot \TX, \qq B^\circ_a(v) \mapsto  S^{(1)}_a(v;\bm w,\bm u)
 \label{Bnr->TX} 
}
equips the space $M^{(1)}$ with the structure of a lowest weight $\tBnr$-module with lowest weight given by 
\ali{ 
\wt\ga_j(v;\bm w,\bm u) &= g(v)\,\wt\mu^\circ_j(v)\,\Bigg(\prod_{i=1}^\ell \la^{(i)}_j(v-\tfrac\ka2)\,\bar\la^{(i)}_{j}(\tv-\tfrac\ka2)\Bigg) , \label{la_i(v;w;u)} \\
\wt\ga_n(v;\bm w,{\bm u}) &= g(v)\,\wt\mu^\circ_n(v)\,\Bigg(\prod_{i=1}^m \frac{v-u_i+1}{v-u_i}\cdot \frac{v-w_i+1}{v-w_i}\cdot \frac{v-\tw_i\mp1+1}{v-\tw_i\mp1}\cdot \frac{v-\tu_i\mp1+1}{v-\tu_i\mp1} \Bigg) \el 
&\hspace{1.6cm}\times \Bigg(\prod_{i=1}^\ell \la^{(i)}_n(v-\tfrac\ka2)\, \bar\la^{(i)}_{n}(\tv-\tfrac\ka2)\Bigg) \label{la_n(v;u)}
}
for $1\le j \le n-1$ with $\mu_j(v)$, $\mu_n(v)$ defined in Proposition \ref{P:1-dim} and $\la^{(i)}_j(v-\tfrac\ka2)$, $\bar\la^{(i)}_j(\tv-\tfrac\ka2)$ are given by
\equ{
\la^{(i)}_j(v) = 1-\frac{\la^{(i)}_j}{v-c_i},  \qq \bar\la^{(i)}_j(v) = 1 + \frac{\la^{(i)}_j}{v-c_i\mp k_i\pm1-\ka} \label{la_i(v)}
}
for $1\le j\le n$, where $\la^{(i)} = (k_i,0,\ldots,0)$ in the orthogonal case and $\la^{(i)} = (1,\ldots,1,0,\ldots,0)$, with the number of $1$'s being $k_i$, in the symplectic case.
\end{prop}

\begin{proof} 
We start by proving the Proposition in the case $m=0$.
Relation \eqref{RE:AA} with Lemma~\ref{L:M0-stable} imply that $A(v)$ satisfies the reflection equation on $M^0$. That is, for any $\zeta \in M^0$,
\[
R_{ab}(v-x) A_a(v)R_{ab}(v+x+\rho)A_b(x)\cdot \zeta =  A_b(x)R_{ab}(v+x+\rho)A_a(v)R_{ab}(v-x) \cdot \zeta .
\]
The remaining terms, which contain $C(u)$ as the rightmost operator, vanish due to Lemma~\ref{L:M0-stable}. 
It follows that $M^0$ is a lowest weight $\tBnr$-module, with weights obtained from \eqref{l(u)-fused} and Proposition~\ref{P:L*V}. 
The $m>0$ case is then immediate from Proposition~\ref{P:L*V} and \eqref{gnmono2}, as the auxiliary spaces are regarded as dual vector evaluation representations of $Y(\mfgl_n)$ with shifts of $u_i$ or $w_i$ for $1 \leq i \leq m$, and lowest weight vector~$e_1$.
\end{proof}

Proposition \ref{P:B(n,r)-rep} implies that $M^{(1)}$ can be viewed as the (full) quantum space for a residual $\tBnr$-chain (since $S^{(1)}_a(v;\bm w,\bm u)$ satisfies \eqref{unit-f} but not \eqref{unit}).
We end this section with a lemma which will assist us in finding the explicit expressions of the unwanted terms. 
Recall that $\check{R}(u) := \frac{u}{u-1}\,PR(u)$.

\begin{lemma} \label{L:RRs=sRR}
The following identities hold:
\aln{
\check{R}(u) \, e_1 \ot e_1 &= e_1 \ot e_1 , \\
\check{R}_{\at_i \at_{i+1}}(u_{i+1}-u_{i}) \, s_{kl}(v;\bm w,\bm u) 
&= s_{kl}(v;\bm w,\bm u_{i \leftrightarrow i+1}) \, \check{R}_{\at_i \at_{i+1}}(u_{i+1}-u_{i}) ,
\\
\check{R}_{a_i a_{i+1}}(w_{i+1}-w_{i}) \, s_{kl}(v;\bm w,\bm u) 
&= s_{kl}(v;\bm w_{i \leftrightarrow i+1};\bm u) \, \check{R}_{a_i a_{i+1}}(w_{i+1}-w_{i}) .
}
\end{lemma}

\begin{proof}
The first identity follows from the definition of $\check{R}(u)$. 
To obtain the second identity we need to move $\check{R}_{\at_i \at_{i+1}}(u_{i+1}-u_{i})$ rightward through the products of $R$-matrices in the definition of $S^{(1)}_a(v;\bm w,\bm u)$ in \eqref{gnmono}. In each product we must use the (braided) Yang-Baxter equation once. 
For $\check{R}_{\at_i \at_{i+1}}(u_{i+1}-u_{i})$ in the leftmost product,
\[
\check{R}_{\at_i \at_{i+1}}(u_{i+1}-u_{i}) R^t_{\at_i a}(u_i-v) R^t_{\at_{i+1} a}(u_{i+1}-v) 
= R^t_{\at_i a}(u_{i+1}-v) R^t_{\at_{i+1} a}(u_{i}-v) \check{R}_{\at_i \at_{i+1}}(u_{i+1}-u_{i}),
\]
and in the rightmost product,
\aln{
&\check{R}_{\at_i \at_{i+1}}(u_{i+1}-u_{i}) R^t_{\at_{i+1} a}(\tu_{i+1}-v\pm1) R^t_{\at_i a}(\tu_i-v\pm1) 
\\
&\hspace{3cm} = R^t_{\at_{i+1} a}(\tu_i-v\pm1) R^t_{\at_i a}(\tu_{i+1}-v\pm1) \check{R}_{\at_i \at_{i+1}}(u_{i+1}-u_{i}).
}
Applying these identities, we obtain the second identity. The third identity is obtained similarly.
\end{proof}


\subsection{Transfer matrix and Bethe vectors for a $X_\rho(\mfsp_{2n},\mfsp_{2n}^\theta)^{tw}$-chain} \label{sec:NABA-SP}

We are now ready to introduce the transfer matrix acting on the quantum space $M$ defined in \eqref{M} and find its eigenvectors, the Bethe vectors.

\begin{defn} \label{D:tm}
The transfer matrix $\tau(v) \in \End(M)[v,v^{-1}]$ is the representative of $\dfrac{\tr S(v)}{2v-2\ka-\rho}$ on $M$.
\end{defn}

From arguments given in \cite{Sk} (see also Section 2.2 in \cite{Vl}) the reflection equation \eqref{RE} implies that transfer matrices commute,
\[
[\tau(u),\tau(v)] = 0.
\]
Lemma~\ref{L:tra_symm} allows us to deduce the following symmetry properties of the transfer matrix.

\begin{crl} 
The transfer matrix satisfies the following equivalence relation:
\[
\tau(\tv) = \tau(v) = \{p(v) \tr A(v) \}^v.
\]
\end{crl}

Recall the generalised nested monodromy matrix $S^{(1)}_{a}(v;\bm w,\bm u)$ defined in Definition~\ref{D:gnmono}, and the nested vacuum sector $M^{(1)}$ from \eqref{nest_vac}.
By Proposition~\ref{P:B(n,r)-rep} we regard $M^{(1)}$ as the (full) quantum space of a residual $\tBnr$-chain.
Let $\Phi^{(1)}(\bm u^{(1...n-1)};\bm w, \bm u)$ denote the level-$1$ Bethe vector constructed from $S^{(1)}_{a}(v;\bm w,\bm u)$ according to Definition~\ref{D:Bnr-BV-k}.

\begin{lemma} \label{L:spn-nBethe_symm}
The level-1 Bethe vector satisfies
\aln{
\check{R}_{\at_i \at_{i+1}}(w_i-w_{i+1}) \, \Phi^{(1)}(\bm u^{(1...n-1)};\bm w, \bm u) & = \Phi^{(1)}(\bm u^{(1...n-1)};\bm w_{i \leftrightarrow i+1}; \bm u),
\\
 \check{R}_{a_i a_{i+1}}(u_{i}-u_{i+1})  \, \Phi^{(1)}(\bm u^{(1...n-1)};\bm w, \bm u) &= \Phi^{(1)}(\bm u^{(1...n-1)};\bm w, \bm u_{i \leftrightarrow i+1}).
}
\end{lemma}

\begin{proof}
The level-1 Bethe vector is constructed from a linear combination of products of matrix elements $s_{kl}(v;\bm w, \bm u)$ of the generalised nested monodromy matrix acting on the highest weight vector. 
The result is therefore immediate from Lemma~\ref{L:RRs=sRR}.
\end{proof}

Recall the creation operator $\beta_{\bm \at \bm a}(\bm u)$ for $m$ excitations from Definition~\ref{D:creation-op-m}.
In what follows we will use the notation $u_i^{(n)}:=u_i-\frac{\ka}{2} = u_i-\frac{n+1}{2}$ and $m_n := m$.
Additionally, we will use $\bm u +a$ to mean $(u_1 +a, \dots, u_m+a)$.

\begin{defn} \label{D:spn-Bethe}
The (top-level) symplectic Bethe vector is defined by
\aln{
\Psi(\bm u^{(1...n)}) &:= \beta_{\bm \at \bm a}(\bm u^{(n)}+\tfrac{\ka}{2}) \cdot \Phi^{(1)}(\bm u^{(1...n-1)}; \tfrac{\ka}{2}- \bm u^{(n)}-\rho, \bm u^{(n)}+\tfrac{\ka}{2}) 
\\
&\;= \beta_{\bm \at \bm a}(\bm u) \cdot \Phi^{(1)}(\bm u^{(1...n-1)};\bm \tu, \bm u) ,
}
where $\bm \tu = (\tu_1,\dots, \tu_{m})$ with $\tu_i = \ka - u_i -\rho$. 
\end{defn}

As with the $\tBnr$ case, $\mf{S}_{\bm m} := \mf{S}_{m_1} \times \dots \times \mf{S}_{m_{n-1}} \times \mf{S}_{m_n}$ acts on the symplectic Bethe vector by reordering parameters. 
The invariance of the Bethe vector under this action can then be shown by combining Lemma~\ref{L:B_symm} and Lemma~\ref{L:spn-nBethe_symm}.

\begin{crl} \label{C:spn-Psi_inv}
The symplectic Bethe vector is invariant under the action of $\mf{S}_{\bm m}$. \qed
\end{crl}

Recall the notation $\La^\pm(v;\bm u^{(k)})$ in \eqref{Bnr-La-PM} and in addition define
\[
\Lambda^{+2}(v, \bm u^{(n)}) := \prod_{i=1}^{m_n} \frac{(v+u^{(n)}_i+2+\rho)(v-u^{(n)}_i+2)}{(v+u^{(n)}_i+\rho)(v-u^{(n)}_i)} .
\]
The Theorem below is our first main result.

\begin{thrm} \label{T:spn-spec}
The symplectic Bethe vector $\Psi(\bm u^{(1...n)})$ is an eigenvector of the transfer matrix $\tau(v)$ with eigenvalue 
\equ{
\Lambda(v;\bm u^{(1...n)}) := \big\{ p(v)\,\La^{(1)}\big(v;\bm u^{(1...n)}\big) \big\}^v  
\label{spn-La(v,u)}
}
where 
\aln{
\La^{(1)}\big(v;\bm u^{(1\dots n)}\big) &:= 
\frac{2v-n+\rho}{2v-1+\rho}\,\,\La^+\big(v-\tfrac 12,\bm u^{(1)}\big)\,\frac{\wt{\ga}_{1}(v)}{2v+\rho} 
\\ 
& \qu + \sum_{i=2}^{n-1} \frac{2v-n+\rho}{2v-i+\rho}\,\,\La^-\big(v-\tfrac{i-1}2,\bm u^{(i-1)}\big)\,\La^+\big(v-\tfrac i2,\bm u^{(i)}\big)\,\frac{\wt{\ga}_{i}(v)}{2v-i+1+\rho}\\
& \qu + \La^-\big(v-\tfrac{n-1}2,\bm u^{(n-1)}\big)\,\La^{+2}\big(v-\tfrac{\ka}{2},\bm u^{(n)}\big)\, \frac{\wt{\ga}_n(v)}{2v-n+1+\rho} 
\intertext{and} 
\wt\ga_j(v) & = g(v)\,\,\wt\mu^\circ_j(v)\,\prod_{i=1}^\ell \la^{(i)}_j(v-\tfrac\ka2)\,\prod_{i=1}^\ell \la^{\prime(i)}_{j}(\tv-\tfrac\ka2)
}
for $1 \leq j \leq n$, provided
\equ{ \label{spn-BE}
\Res{v \rightarrow u^{(i)}_j} \Lambda(v+\tfrac{i}2;\bm u^{(1...n)}) = 0 
\qq\text{and}\qq
\Res{v \rightarrow u^{(n)}_k} \Lambda(v+\tfrac{\ka}{2};\bm u^{(1...n)}) = 0 
}
for $1 \leq j \leq m_i$, $1 \leq i \leq n-1$ and $1 \leq k \leq m_n$.
\end{thrm}

\begin{rmk}
The equations \eqref{spn-BE} are Bethe equations for a $X_\rho(\mfsp_{2n},\mfsp_{2n}^\theta)^{tw}$-chain. Their explicit form for $u^{(i)}_j$ with $1 \leq i \leq n-2$ is the same as in \eqref{Bnr-BE2}. Those for $u^{(n-1)}_j$ receive an additional factor due to the top-level excitations,
\ali{\label{spn-BE2}
& \frac{\wt{\ga}_{n-1}\big(u^{(n-1)}_j+\tfrac {n-1}2\big)}{\wt{\ga}_{n}\big(u^{(n-1)}_j+\tfrac{n-1}{2}\big)} \prod_{\substack{i=1\\i\ne j}}^{m_{n-1}} \frac{(u^{(n-1)}_j-u^{(n-1)}_i+1)(u^{(n-1)}_j+u^{(n-1)}_i+1+\rho)}{(u^{(n-1)}_j-u^{(n-1)}_i-1)(u^{(n-1)}_j+u^{(n-1)}_i-1+\rho)} \el
& \qu = \prod_{i=1}^{m_{n-2}} \frac{(u^{(n-1)}_j-u^{(n-2)}_i+\tfrac12)(u^{(n-1)}_j+u^{(n-2)}_i+\tfrac12+\rho)}{(u^{(n-1)}_j-u^{(n-2)}_i-\tfrac12)(u^{(n-1)}_j+u^{(n-2)}_i-\tfrac12+\rho)} \el & \qq \times  \prod_{i=1}^{m_n} \frac{(u^{(n-1)}_j-u^{(n)}_i+1)(u^{(n-1)}_j+u^{(n)}_i+1+\rho)}{(u^{(n-1)}_j-u^{(n)}_i-1)(u^{(n-1)}_j+u^{(n)}_i-1+\rho)}. %
}
The top-level Bethe equations, for $u^{(n)}_j$, are
\ali{ \label{spn-BE3}
&\frac{\wt{\ga}_n(u_j^{(n)}+\tfrac{\ka}{2})}{\wt{\ga}_n(\tfrac{\ka}{2}-u_j^{(n)}-\rho)}
\prod_{\substack{i=1\\i\neq j}}^{m_n} \frac{(u_j^{(n)}-u_i^{(n)}+2)(u_j^{(n)}+u_i^{(n)}+2+\rho)}{(u_j^{(n)}-u_i^{(n)}-2)(u_j^{(n)}+u_i^{(n)}-2+\rho)}
\el 
& \qu =\prod_{i=1}^{m_{n-1}} \frac{(u_j^{(n)}-u_i^{(n-1)}+1)(u_j^{(n)}+u_i^{(n-1)}+1 +\rho)}{(u_j^{(n)}-u_i^{(n-1)}-1)(u_j^{(n)}+u_i^{(n-1)}-1 +\rho)}
.
}
\end{rmk}

\begin{proof}[Proof of Theorem~\ref{T:spn-spec}]
In order to prove the theorem, it will be necessary to calculate an expression for the unwanted terms.
As such, we will first expand on the exchange relations of the twisted Yangian, studied in Section~\ref{sec:AB-multi}.
Recall Corollary~\ref{C:wantedterm},
\aln{
\{p(v)\,A_a(v)\}^v \beta_{\bm \at \bm a}(\bm u) = \beta_{\bm \at \bm a}(\bm u)\,\{p(v)\, S^{(1)}_{a}(v;\bm u)\}^v +UWT.
}
Let $X_B$ denote the subalgebra of $\TXS$ generated by elements of the $B$ block matrix, i.e.\ $s_{i,n+j}^{(k)}$ with $1 \leq i,j \leq n$, $k \geq 1$. 
The closure of $X_B$ is guaranteed by \eqref{RE:BB}.
Then, considering repeated applications of Lemma~\ref{L:AB_half}, it is possible to write $UWT$ above such that, in each term, elements of the $X_B$ subalgebra appear to the left of the expression. 
That is, there exist $B^{\pm,k}_{ij} \in W^*_{\bm \at \bm a} \ot X_B((v^{-1}))$ such that
\aln{
\Tr_a \{ p(v)\,A_a(v)\}^v \beta_{\bm \at \bm a}(\bm u) = \beta_{\bm \at \bm a}(\bm u)\,\Tr_a \{p(v)\,S^{(1)}_{a}(v;\bm u)\}^v +\sum_{k=1}^m \sum_{i,j=1}^n \big( B_{ij}^{+,k} \sca_{ij}(u_k) + B_{ij}^{-,k} \sca_{ij}(\tu_k)\big).
}
Since we will not need the exact form the $B^{\pm,k}_{ij}$, we define the combination
\[
U^{k}(v;\bm u) := \sum_{i,j=1}^n  \big( B_{ij}^{+,k} \sca_{ij}(u_k) + B_{ij}^{-,k} \sca_{ij}(\tu_k)\big),  
\]
where we have made explicit the dependence on $v$ and $\bm u$.
From Lemma~\ref{L:AB_full} we obtain an exact expression for the unwanted terms for a single excitation. 
Applying this to the leftmost creation operator $\beta_{\at_1 a_1}(u_1)$, followed by Lemma~\ref{L:movement}, we are able to extract an expression for $U^1(v;\bm u)$:
\ali{ \label{spn-uwt_1}
U^1(v;\bm u) &= \frac1{p(u_1)}\bigg\{p(v)\,\frac{\beta_{\at_1 a_1}(v)}{u_1-v} \bigg\}^v
\el & \qu \times \prod_{i=2}^m \Big( \beta_{\at_i a_i}(u_i) \prod_{j=i-1}^1 R_{a_j \at_i}(-u_j-u_i-\rho) \Big) \Res{w \rightarrow u_1}  \Tr_a \big\{ p(w) \,  \,S^{(1)}_{a}(w;\bm u) \big\}^w .
}
From here, to find $U^{k}(v;\bm u)$ for $2 \leq k \leq m$ we make use of Lemma~\ref{L:B_symm}.
Specifically, by repeatedly applying transpositions, we may apply any permutation $\sigma \in \mf{S}_m$ to the parameters $\bm u$.
Let $\bm u_{\sigma}$ denote $(u_{\sigma(1)}, \dots, u_{\sigma(m)})$, and let $\sigma_k$ denote the cyclic permutation $(k,k+1, \dots, 1, m, \dots, k-1)$. 
Then
\[
\beta_{\bm \at \bm a}(\bm u) = 
\beta_{\bm \at \bm a}(\bm u_{\sigma_k}) \, \check{R}_{\bm a}[\sigma_k](\bm u) \, \check{R}_{\bm \at}[\sigma_k](\bm \tu)
\]
where $\check{R}_{\bm a}[\sigma_k](\bm u)$ is the product of $\check{R}$ matrices necessary to implement this cyclic permutation,
\aln{
\check{R}_{\bm a}[\sigma_k](\bm u) 
=\prod_{j=k-1}^1 \Bigg( \prod_{i=m-1}^1 \check{R}_{a_i a_{i+1}}(u_j - u_{j+1}) \Bigg).
}
With this permuted creation operator, repeating the arguments used to find \eqref{spn-uwt_1} yields
\ali{ \label{spn-uwt_k}
U^k(v;\bm u) &=  \frac1{p(u_k)} \bigg\{ \frac{p(v)}{u_k-v} \beta_{\at_1 a_1}(v) \bigg\}^v \prod_{i=2}^m \Big( \beta_{\at_i a_i}(u_{\sigma_k(i)}) \prod_{j=i-1}^1 R_{a_j \at_i}(-u_{\sigma_k(j)}-u_{\sigma_k(i)}-\rho) \Big) 
\el
&\qu \times \Res{w \rightarrow u_k} \big\{p(w)\tr_a S^{(1)}_{a}(w; \bm u_{\sigma_k})\big\}^w \check{R}_{\bm a}[\sigma_k](\bm u) \, \check{R}_{\bm \at}[\sigma_k](\bm \tu),
}
and therefore a full expression for the unwanted terms,
\aln{
\Tr_a \{ p(v) A_a(v)\}^v \beta_{\bm \at \bm a}(\bm u) &= \beta_{\bm \at \bm a}(\bm u) \,\Tr_a \{p(v)\,S^{(1)}_{a}(v;\bm u)\}^v 
\\
&\qu +\sum_{k=1}^m \frac1{p(u_k)} \bigg\{ \frac{p(v)}{u_k-v} \beta_{\at_1 a_1}(v) \bigg\}^v \prod_{i=2}^m \Big( \beta_{\at_i a_i}(u_{\sigma_k(i)}) \! \prod_{j=i-1}^1 \! R_{a_j \at_i}(-u_{\sigma_k(j)}{-}u_{\sigma_k(i)}{-}\rho) \Big) 
\el
& \qq \times \Res{w \rightarrow u_k} \tr_a \big\{p(w) S^{(1)}_{a}(w; \bm u_{\sigma_k})\big\}^w \check{R}_{\bm a}[\sigma_k](\bm u) \, \check{R}_{\bm \at}[\sigma_k](\bm \tu).
}
Acting now with this expression on the level-1 Bethe vector gives the full action for the transfer matrix on the top level Bethe vector with $u^{(n)}_i = u_i-\tfrac{\ka}{2}$,
\aln{
\tau(v) \cdot \Psi(\bm u^{(1...n)}) &= \beta_{\bm \at \bm a}(\bm u)\, \Tr_a \{p(v) \, S^{(1)}_{a}(v;\bm u)\}^v \cdot \Phi^{(1)}(\bm u^{(1...n-1)};\bm \tu, \bm u)
\\
&\qu +\sum_{k=1}^m \frac1{p(u_k)} \bigg\{ \frac{p(v)}{u_k-v} \beta_{\at_1 a_1}(v) \bigg\}^v \prod_{i=2}^m \Big( \beta_{\at_i a_i}(u_{\sigma_k(i)}) \prod_{j=i-1}^1 \! R_{a_j \at_i}(-u_{\sigma_k(j)}{-}u_{\sigma_k(i)}{-}\rho) \Big) 
\el
& \qq  \times \Res{w \rightarrow u_k} \tr_a \big\{p(w) \, S^{(1)}_{a}(w; \bm u_{\sigma_k})\big\}^w \check{R}_{\bm a}[\sigma_k](\bm u) \, \check{R}_{\bm \at}[\sigma_k](\bm \tu) \cdot \Phi^{(1)}(\bm u^{(1...n-1)};\bm \tu, \bm u)
\\
& = \beta_{\bm \at \bm a}(\bm u)\, \Tr_a \{p(v) S^{(1)}_{a}(v;\bm u)\}^v \cdot \Phi^{(1)}(\bm u^{(1...n-1)};\bm \tu, \bm u)
\\
&\qu +\sum_{k=1}^m \frac1{p(u_k)} \bigg\{ \frac{p(v)}{u_k-v} \beta_{\at_1 a_1}(v) \bigg\}^v \prod_{i=2}^m \Big( \beta_{\at_i a_i}(u_{\sigma_k(i)}) \prod_{j=i-1}^1 \! R_{a_j \at_i}(-u_{\sigma_k(j)}{-}u_{\sigma_k(i)}{-}\rho) \Big) 
\el
& \qq  \times \Res{w \rightarrow u_k} \tr_a \big\{p(w) \, S^{(1)}_{a}(w; \bm u_{\sigma_k})\big\}^w \cdot  \Phi^{(1)}(\bm u^{(1...n-1)};\bm \tu_{\sigma_k}, \bm u_{\sigma_k}).
}
The last equality follows from Lemma~\ref{L:spn-nBethe_symm}.
From the full expression \eqref{spn-La(v,u)}, the condition \eqref{spn-BE} for the parameters $u_i^{(j)}$ is equivalent to $\mathrm{Res}_{v \rightarrow u^{(j)}_i} \Lambda^{(1)}(v+\tfrac{j}2;\bm u^{(1...n)}) = 0$ with $ u^{(n)}_i = u_i-\tfrac{\ka}{2}$, as these poles are not present in $ \Lambda^{(1)}(\tv-\tfrac{j}2;\bm u^{(1...n)})$.
Therefore, from Theorem~\ref{T:Bnr-spec}, using weights from Proposition~\ref{P:B(n,r)-rep},
\aln{
\tau(v) \cdot \Psi\big(\bm u^{(1...n)}\big) &= \La\big(v;\bm u^{(1...n)}\big) \, \Psi\big(\bm u^{(1...n)}\big) \\ & \qu + \sum_{k=1}^m  \Res{w \rightarrow u_k} \Lambda(w;\bm u^{(1...n)}_{\sigma^{(n)}_k}) \frac1{p(u_k)} \bigg\{ \frac{p(v)}{u_k-v} \beta_{\at_1 a_1}(v) \bigg\}^v
\\
&\qq  \times \prod_{i=2}^m \Big( \beta_{\at_i a_i}(u_{\si_k(i)}) \prod_{j=i-1}^1 R_{a_j \at_i}(-u_{\sigma_k(j)}{-}u_{\sigma_k(i)}{-}\rho) \Big) \cdot  \Phi^{(1)}\big(\bm u^{(1...n-1)};\bm \tu_{\sigma_k}, \bm u_{\sigma_k}\big),
}
where $\Lambda(v;\bm u^{(1...n)}) = \big\{p(v) \Lambda^{(1)}(v;\bm u^{(1...n)}) \big\}^v$ as required.
Note that, owing to Corollary~\ref{C:spn-Psi_inv}, we have $\Lambda(v;\bm u^{(1...n)}) = \Lambda(v;\bm u^{(1...n)}_{\sigma^{(n)}})$ for any $\sigma^{(n)} \in \mf{S}_{m_n}$. 
Therefore, $\Psi(\bm u^{(1...n)})$ is an eigenvector of $\tau(v)$ with eigenvalue $\Lambda(v;\bm u^{(1...n)})$ provided $\Res{v \rightarrow u_k} \Lambda(v;\bm u^{(1...n)}) = 0$, or equivalently $\Res{v \rightarrow u^{(n)}_k} \Lambda(v+\tfrac{\ka}{2};\bm u^{(1...n)}) = 0$,  for $1 \leq k \leq m_n$. 
\end{proof}

\begin{exam} \label{E:BV-spn} 
The symplectic Bethe vector with $m$ top-level excitations and $m_1=\ldots=m_{n-1}=0$~is~given~by
\aln{
\Psi(\bm u^{(n)}) &= \big[B(u_1^{(n)}\!+\!\tfrac{\ka}{2})\big]_{n,1} \cdots \big[B(u_m^{(n)}\!+\!\tfrac{\ka}{2})\big]_{n,1} \cdot \eta 
\\
&=\big[S(u_1^{(n)}\!+\!\tfrac{\ka}{2})\big]_{n,n+1} \cdots \big[S(u_m^{(n)}\!+\!\tfrac{\ka}{2})\big]_{n,n+1} \cdot \eta .
}
For $m_1 = m_n = 1$ and $m_2 = \ldots = m_{n-1}=0$, the on-shell symplectic Bethe vector, that is, when the parameters satisfy the Bethe equations, takes the form
\ali{
\Psi(u^{(n)},u^{(n-1)}) &= \frac{(u^{(n)}-u^{(n-1)}-1)(u^{(n)}+u^{(n-1)}+\rho+1)}{(u^{(n)}-u^{(n-1)})(u^{(n)}+u^{(n-1)}+\rho)} \Bigg[ \big[B(u^{(n)}\!+\!\tfrac{\ka}{2})\big]_{n,1} \big[\hat{A}^{(n-1)}(u^{(n-1)}\!+\!\tfrac12)\big]_{12}  \el
& \qu - \frac{\wt \ga_{n-1}(u^{(n-1)}+\tfrac{n-1}{2})}{(u^{(n)}-u^{(n-1)}-1)(u^{(n)}+u^{(n-1)}+\rho+1)} \Big(\big[B(u^{(n)}\!+\!\tfrac{\ka}{2})\big]_{n,2}+\big[B(u^{(n)}\!+\!\tfrac{\ka}{2})\big]_{n-1,1}\Big) \Bigg] \cdot \eta,
}
where $\hat{A}^{(n-1)}(v)$ refers to the level-$(n\!-\!1)$ nested version of the $A$ operator of $S(v)$ obtained via \eqref{D-hat}.
\end{exam}


\subsection{Transfer matrix and Bethe vectors for a $X_\rho(\mfso_{2n},\mfso_{2n}^\theta)^{tw}$-chain}  \label{sec:NABA-SO}

We now focus on the orthogonal case. 
We define the transfer matrix $\tau(v)$ acting on the quantum space $M$ defined in \eqref{M} in the same way as we did in the symplectic case. 
However, the definition of the orthogonal Bethe vector will differ from its symplectic counterpart in Definition \ref{D:spn-Bethe}. 
Indeed, looking at Proposition~\ref{P:B(n,r)-rep}, the weights $\wt{\gamma}_n(v;\bm \tu; \bm u)$ do not have poles at $v=u_i$, and so making the same ansatz as in the symplectic case would yield Bethe equations that are trivially satisfied. 
Such an ansatz therefore must be identically equal to zero.
To remedy this we use a limiting procedure proposed in \cite{DVK}.
Recall that $\Phi^{(1)}(\bm u^{(1...n-1)};\bm w; \bm u)$ denotes the level-1 Bethe vector constructed from $S^{(1)}_{a}(v;\bm w;\bm u)$ according to Definition~\ref{D:Bnr-BV-k}.

\begin{defn} \label{D:son-lim-Bethe}
The level-1 orthogonal Bethe vector is defined by
\[ 
\Phi_{lim}^{(1)}(\bm u^{(1...n-1)},\bm \tu ;\bm \al, \bm \beta) := 
\lim_{\eps \rightarrow 0} \, \Phi^{(1)}(\bm u^{(1...n-2)},(\bm u^{(n-1)},\bm \tu -\tfrac{\ka}{2}- \eps);\bm \tu - \bm \beta\,\eps; \bm u + \bm \al\,\eps) 
. 
\]
\end{defn}

In the above definition, as well as parameters $\bm u^{(1...n-1)}$, the Bethe vector includes $m$ additional excitations at level-$(n\!-\!1)$, with parameters $\tu_i - \tfrac\ka2 -\eps = \tfrac\ka2-u_i-\rho-\eps$. 
The shift of $\tfrac\ka2 = \tfrac12(n-1)$ is simply to account for the parameter shifts in the nested Bethe ansatz for the $\tBnr$-chain.
Parameters $\bm \alpha$ and $\bm \beta$ have been introduced to control the limit as $\eps \rightarrow 0$.
These parameters should be thought of as additional Bethe parameters, which will eventually be determined by the Bethe equations. 
We obtain the same parameter symmetry as Lemma~\ref{L:spn-nBethe_symm}.
\begin{lemma} \label{L:son-nBethe_symm}
The level-1 orthogonal Bethe vector satisfies
\aln{
\check{R}_{\at_i \at_{i+1}}(u_{i+1}-u_{i}) \check{R}_{a_i a_{i+1}}(u_{i}-u_{i+1})\,\Phi_{lim}^{(1)}(\bm u^{(1...n-1)},\bm \tu ;\bm \al, \bm \beta)
=\, \Phi_{lim}^{(1)}(\bm u^{(1...n-1)},\bm \tu_{i \leftrightarrow i+1} ;\bm \al_{i \leftrightarrow i+1}, \bm \beta_{i \leftrightarrow i+1}).
}
\end{lemma}

\begin{proof}
We use Lemma~\ref{L:Bnr-bethe-symm} to write
\[
\Phi_{lim}^{(1)}(\bm u^{(1...n-1)},\bm \tu ;\bm \al, \bm \beta)=\Phi_{lim}^{(1)}(\bm u^{(1...n-1)},\bm \tu_{i \leftrightarrow i+1} ;\bm \al, \bm \beta).
\]
Then
\aln{
&\check{R}_{\at_i \at_{i+1}}(u_{i+1}-u_{i}) \check{R}_{a_i a_{i+1}}(u_{i}-u_{i+1})  \lim_{\eps \rightarrow 0} \, \Phi^{(1)}(\bm u^{(1...n-2)},(\bm u^{(n-1)},\bm \tu -\tfrac{\ka}{2}- \eps);\bm \tu - \bm \beta \eps; \bm u + \bm \al \eps) 
\\
&\qq =\lim_{\eps \rightarrow 0} \check{R}_{\at_i \at_{i+1}}(u_{i+1}-u_{i}+(\beta_i-\beta_{i+1})\, \eps) \check{R}_{a_i a_{i+1}}(u_{i}-u_{i+1}+(\al_i-\al_{i+1})\, \eps)   \\
&\hspace{6cm}  \times \, \Phi^{(1)}(\bm u^{(1...n-2)},(\bm u^{(n-1)},\bm \tu -\tfrac{\ka}{2}- \eps);\bm \tu - \bm \beta \eps; \bm u + \bm \al \eps) 
\\
&\qq = \lim_{\eps \rightarrow 0} \Phi^{(1)}(\bm u^{(1...n-2)},(\bm u^{(n-1)},\bm \tu -\tfrac{\ka}{2}- \eps);\bm \tu_{i \leftrightarrow i+1} - \bm \beta_{i \leftrightarrow i+1} \eps; \bm u_{i \leftrightarrow i+1} + \bm \al_{i \leftrightarrow i+1} \eps) ,
}
where the last equality follows from Lemma~\ref{L:RRs=sRR}, as in the symplectic case. 
Then, using Lemma~\ref{L:Bnr-bethe-symm} to exchange $\bm \tu -\tfrac{\ka}{2}- \eps$ with $\bm \tu_{i \leftrightarrow i+1}-\tfrac{\ka}{2}- \eps$, we obtain the desired result.
\end{proof}

\begin{crl} \label{C:son-nBethe-spec}
The level-1 orthogonal Bethe vector satisfies
\[
\tau^{(1)}(v;\bm \tu; \bm u) \, \Phi_{lim}^{(1)}(\bm u^{(1...n-1)},\bm \tu ;\bm \al, \bm \beta)
=
\Lambda^{(1)}\big(v;\bm u^{(1...n-2)},(\bm u^{(n-1)},\bm \tu-\tfrac{\ka}{2});\bm \tu,\bm u\big) \, \Phi_{lim}^{(1)}(\bm u^{(1...n-1)},\bm \tu ;\bm \al, \bm \beta)
\]
with 
\ali{ 
\label{son-Bnr-La1}
\Lambda^{(1)} & \big(v;\bm u^{(1...n-1)};\bm w,\bm u\big) = \frac{2v-n+\rho}{2v-1+\rho}\,\,\La^+\big(v-\tfrac 12,\bm u^{(1)}\big)\,\frac{\wt{\ga}_{1}(v)}{2v+\rho} \el
& \qu + \sum_{i=2}^{n-1} \frac{2v-n+\rho}{2v-i+\rho}\,\,\La^-\big(v-\tfrac{i-1}2,\bm u^{(i-1)}\big)\,\La^+\big(v-\tfrac i2,\bm u^{(i)}\big)\,\frac{\wt{\ga}_{i}(v)}{2v-i+1+\rho}
\el
&\qu + \La^-\big(v-\tfrac{n-1}2,\bm u^{(n-1)}\big) \La^+\big(v-\tfrac{n}2,\bm w-\tfrac{n}{2}\big) \La^+\big(v-\tfrac{n}2,\bm u-\tfrac{n}{2}\big) \, \frac{\wt{\ga}_n(v)}{2v-n+1+\rho} 
}
provided
\gat{
\Res{v\rightarrow u^{(i)}_j} \Lambda^{(1)}\big(v+\tfrac{i}2;\bm u^{(1...n-2)},(\bm u^{(n-1)},\bm \tu-\tfrac{\ka}{2});\bm \tu,\bm u\big) = 0 \qu \text{for} \qu 1 \leq j \leq m_i, \;\; 1 \leq i \leq n-1, \label{son-Bnr-BE}
\\
\lim_{\eps \rightarrow 0} \Res{v\rightarrow \tu_j-\eps}\Lambda^{(1)}\big(v;\bm u^{(1...n-2)},(\bm u^{(n-1)},\bm \tu-\tfrac{\ka}{2}-\eps);\bm \tu-\bm \beta \eps,\bm u + \bm \al \eps \big) = 0 \qu  \text{for} \qu 1 \leq j \leq m.\label{son-Bnr-BE+}
}
\end{crl}

\begin{proof}
By Theorem \ref{T:Bnr-spec} and Proposition \ref{P:B(n,r)-rep}, vector $\Phi^{(1)}(\bm u^{(1...n-2)},(\bm u^{(n-1)},\bm \tu -\tfrac{\ka}{2}- \eps);\bm \tu - \bm \beta\,\eps; \bm u + \bm \al\,\eps)$ is \linebreak an eigenvector of the nested transfer matrix $\tau^{(1)}(v;\bm \tu-\bm \beta \eps; \bm u + \bm \al \eps)$ with eigenvalue $\La^{(1)} = \La^{(1)}(v;\bm u^{(1...n-2)},$ $(\bm u^{(n-1)},\bm \tu-\tfrac{\ka}{2}-\eps);\bm \tu - \bm\beta \eps ; \bm u +\bm\al \eps)$ provided ${\rm Res}_{v\to u^{(i)}_j+\frac{i}{2}} \La^{(1)} = 0$ for $1\le j \le m_i$, $1\le j \le n-1$ and ${\rm Res}_{v\to \tl u_j - \frac\ka2-\eps+\frac{n-1}{2}} \,\La^{(1)} = 0$ for $1\le j\le m$. Taking the $\eps\to0$ limit gives the wanted result.
\end{proof}

Direct evaluation of the residue and the limit in \eqref{son-Bnr-BE+} yield the following Bethe equations for $1 \leq j \leq m$,
\aln{
\frac{2u_j-n+\rho+2}{2u_j-n+\rho}\frac{\wt\gamma_{n-1}(\tu_j)}{\wt\gamma_n(\tu_j)} &= 
-\frac{1-\al_j}{1-\beta_j} \frac1{\La^+(u_j-\tfrac{n}{2},\bm u^{(n-2)})} \frac{\La^+(u_j-\tfrac{n-1}{2},\bm u^{(n-1)})}{\La^-(u_j-\tfrac{n-1}{2},\bm u^{(n-1)})} .
}
For any collection of Bethe roots $\bm u^{(1...n-1)}$, the above equations can be thought to constrain $\bm \al$ in terms of $\bm \beta$.
With this perspective, for any $m$-tuple $\bm u$, as the equations depend on $ \al_j $ and $\beta_j$ only through the combination $(1-\al_j)/(1-\beta_j)$, there is a 1-parameter family of eigenvectors of the nested transfer matrix with the same eigenvalue. 
We conclude that any choice of $\bm \beta$ must give the same nested Bethe vector.
In particular, there will be two choices of interest:
\aln{
\al_j = 0, \qu \beta_j = \del_j
\qq \text{and} \qq 
\al_j = 1-\frac{1}{1-\del_j} =: \hat{\del}_j, \qu \beta_j = 0,
}
with eigenvector
\equ{ \label{nBethe-del} 
\Phi_{lim}^{(1)}(\bm u^{(1...n-1)},\bm \tu ;\bm 0, \bm \del) = 
\Phi_{lim}^{(1)}(\bm u^{(1...n-1)},\bm \tu ;\bm {\hat{ \delta }}, \bm 0).
}
Note that this equality has only been shown to hold ``on-shell'', i.e.\ when the $\bm u^{(1...n-1)}$ satisfy Bethe equations.
We are now ready to define the top-level orhtogonal Bethe vector.
In what follows, we will write $u^{(n)}_i:=u_i-\frac{\ka}{2}=u_i-\frac{n-1}{2}$, $\bar{v}:= -v-\rho$ and $m_n := m$.

\begin{defn} \label{D:son-Bethe}
The (top-level) orthogonal Bethe vector is
\aln{
\Psi(\bm u^{(1...n)};\bm \del) &:= \beta_{\bm \at \bm a}(\bm u) \cdot \Phi_{lim}^{(1)}(\bm u^{(1...n-1)},\bm \tu ;\bm 0, \bm \del)
\\
&\; = \beta_{\bm \at \bm a}(\bm u^{(n)}+\tfrac{\ka}{2}) \cdot \Phi_{lim}^{(1)}(\bm u^{(1...n-1)},\bm{\bar u}^{(n)}+\tfrac{\ka}{2};\bm 0, \bm \del)  
}
with $\del_j$ defined by, for $1 \leq j \leq m_n$,
\equ{ \label{son-del}
\del_j  := 
1+ \frac{2u^{(n)}_j+\rho-1}{2u^{(n)}_j+\rho+1} \frac{\wt\gamma_{n}(\bar{u}^{(n)}_j+\frac{\ka}{2})}{\wt\gamma_{n-1}(\bar{u}^{(n)}_j+\frac{\ka}{2})} \frac1{\La^+( u^{(n)}_j-\tfrac{1}{2},\bm u^{(n-2)})} \frac{\La^+( u^{(n)}_j,\bm u^{(n-1)})}{\La^-( u^{(n)}_j,\bm u^{(n-1)})}.
}
\end{defn}

We now have an $\mf{S}_{\bm m} := \mf{S}_{m_1} \times \dots \times \mf{S}_{m_{n-1}} \times \mf{S}_{m_n}$ action on the orthogonal Bethe vector by reordering parameters. 
The invariance of the Bethe vector under this action can then be shown by combining Lemma~\ref{L:B_symm} and Lemma~\ref{L:son-nBethe_symm}.

\begin{crl} \label{C:son-Psi_inv}
The orthogonal Bethe vector is invariant under the action of $\mf{S}_{\bm m}$. \qed
\end{crl}

The Theorem below is our second main result. Recall \eqref{son-Bnr-La1}.

\begin{thrm} \label{T:son-spec}
The orthogonal Bethe vector $\Psi(\bm u^{(1...n)};\bm \del)$ is an eigenvector of the transfer matrix $\tau(v)$ with eigenvalue 
\equ{
\Lambda(v;\bm u^{(1...n)}) := \big\{ p(v) \, \La^{(1)}\big(v;\bm u^{(1\dots n-2)},(\bm u^{(n-1)},\bm{\bar{u}}^{(n)});  \bm{\bar u}^{(n)}+\tfrac{\ka}{2}, \bm u^{(n)}+\tfrac{\ka}{2}\big) \big\}^v 
\label{son-La(v,u)}
}
provided
\equ{ \label{son-Bnr-BE2}
\Res{v \rightarrow u^{(i)}_j} \Lambda(v+\tfrac{i}2;\bm u^{(1...n)}) = 0 
}
for $1 \leq j \leq m_{i}$, $1 \leq i \leq n-2$, and
\equ{ \label{son-n-1-BE}
\frac{\wt \gamma_{n-1}(u^{(n-1)}_j+\tfrac{\ka}{2}) }{\wt \gamma_n(u^{(n-1)}_j+\tfrac{\ka}{2}) }
\prod_{\substack{i=1 \\ i\neq j}}^{m_{n-1}} \frac{(u^{(n-1)}_j-u^{(n-1)}_i+1)(u^{(n-1)}_j + u^{(n-1)}_i +\rho + 1)}{(u^{(n-1)}_j-u^{(n-1)}_i-1)(u^{(n-1)}_j+u^{(n-1)}_i+\rho-1)} 
= \frac{1}{\La^-(u^{(n-1)}_j+\tfrac{1}{2},\bm u^{(n-2)})} 
}
for $1 \leq j \leq m_{n-1}$, and
\equ{ \label{son-n-BE}
\frac{\wt \gamma_n(u^{(n)}_j+\tfrac{\ka}{2}) }{\wt \gamma_{n-1}(\bar{u}^{(n)}_j+\tfrac{\ka}{2}) }
\prod_{\substack{i=1 \\ i\neq j}}^{m_n} \frac{(u^{(n)}_j-u^{(n)}_i+1)(u^{(n)}_j + u^{(n)}_i +\rho + 1)}{(u^{(n)}_j-u^{(n)}_i-1)(u^{(n)}_j+u^{(n)}_i+\rho-1)}
= \frac{1}{\La^-(u^{(n)}_j+\tfrac{1}{2},\bm u^{(n-2)})}  
}
for $1 \leq j \leq m_n$.
\end{thrm}

\begin{rmk}
The equations (\ref{son-Bnr-BE2}--\ref{son-n-BE}) are Bethe equations for a $X_\rho(\mfso_{2n},\mfso_{2n}^\theta)^{tw}$-chain. 
Their explicit form for $u^{(i)}_j$ with $1 \leq i \leq n-3$ and $i=n-1$ is the same as in \eqref{Bnr-BE2}. 
For $i=n-2$ there is an additional factor, corresponding to the extra excitations at level $n-1$,
\ali{ \label{son-n-2-BE}
& \frac{\wt{\ga}_{n-2}\big(u^{(n-2)}_j+\tfrac{n-2}2\big)}{\wt{\ga}_{n-1}\big(u^{(n-2)}_j+\tfrac{n-2}{2}\big)} \prod_{\substack{i=1\\i\ne j}}^{m_{n-2}} \frac{(u^{(n-2)}_j-u^{(n-2)}_i+1)(u^{(n-2)}_j+u^{(n-2)}_i+1+\rho)}{(u^{(n-2)}_j-u^{(n-2)}_i-1)(u^{(n-2)}_j+u^{(n-2)}_i-1+\rho)} \el
& \qu = \prod_{i=1}^{m_{n-3}} \frac{(u^{(n-2)}_j-u^{(n-3)}_i+\tfrac12)(u^{(n-2)}_j+u^{(n-3)}_i+\tfrac12+\rho)}{(u^{(n-2)}_j-u^{(n-3)}_i-\tfrac12)(u^{(n-2)}_j+u^{(n-3)}_i-\tfrac12+\rho)} 
\el & \qq \times  \prod_{i=1}^{m_{n-1}} \frac{(u^{(n-2)}_j-u^{(n-1)}_i+\tfrac12)(u^{(n-2)}_j+u^{(n-1)}_i+\tfrac12+\rho)}{(u^{(n-2)}_j-u^{(n-1)}_i-\tfrac12)(u^{(n-2)}_j+u^{(n-1)}_i-\tfrac12+\rho)} 
 \el & \qq \times  \prod_{i=1}^{m_{0}} \frac{(u^{(n-2)}_j-u^{(n)}_i+\tfrac12)(u^{(n-2)}_j+u^{(n)}_i+\tfrac12+\rho)}{(u^{(n-2)}_j-u^{(n)}_i-\tfrac12)(u^{(n-2)}_j+u^{(n)}_i-\tfrac12+\rho)}.
} 
The sets of parameters $\bm u^{(n-1)}$ and $\bm u^{(n)}$ correspond to the two branching Dynkin nodes of $\mfso_{2n}$, and are often denoted $\bm u^{(+)}$ and $\bm u^{(-)}$ .
\end{rmk}

\begin{rmk}
For $n=2$, the Bethe equations \eqref{son-n-1-BE} and \eqref{son-n-BE} decouple into two sets of Bethe equations for open $\mfsl_2$ spin chains, and can be solved separately. 
This is consistent with the isomorphism $\mfso_{4} \cong \mfsl_2 \oplus \mfsl_2$.
Similarly, for $n=3$, the isomorphism  $\mfso_{6} \cong \mfsl_4$ is borne out in the Bethe equations \eqref{son-n-1-BE}, \eqref{son-n-BE} and \eqref{son-n-2-BE}.
\end{rmk}

\begin{proof}[Proof of Theorem~\ref{T:son-spec}]
The calculation of unwanted terms is identical to the symplectic case. In particular, using Lemma~\ref{L:son-nBethe_symm} we find
\aln{
\tau(v) \cdot \Psi(\bm u^{(1...n)};\bm \del) &= \beta_{\bm \at \bm a}(\bm u)\, \{p(v) \, \tau^{(1)}(v;\bm \tu; \bm u)\}^v \cdot \Phi_{lim}^{(1)}(\bm u^{(1...n-1)},\bm \tu ; \bm 0, \bm \del) 
\\
&\qu +\sum_{j=1}^{m_n} \frac1{p(u_j)} \bigg\{ \frac{p(v)}{u_j-v}\, \beta_{\at_1 a_1}(v) \bigg\}^v \prod_{i=2}^{m_n} \Big( \beta_{\at_i a_i}(u_{\sigma_j(i)}) \prod_{k=i-1}^1 R_{a_k \at_i}(-u_{\sigma_j(k)}{-}u_{\sigma_j(i)}{-}\rho) \Big) 
\\
& \qq  \times \Res{w \rightarrow u_j} \big\{p(w) \, \tau^{(1)}(w;\bm \tu_{\sigma_j}; \bm u_{\sigma_j})\big\}^w \cdot 
\Phi_{lim}^{(1)}(\bm u^{(1...n-1)},\bm \tu_{\sigma_j} ;\bm 0, \bm \del_{\sigma_j})  . 
}
Recall notation $u^{(n)}_j = u -\tfrac{\ka}{2}$. Corollary~\ref{C:son-nBethe-spec} applied to the wanted term together with the identity 
\[
\Res{v\to u^{(i)}_j} \La^{(1)}\big(\tv-\tfrac{i}2;\bm u^{(1\dots n-2)},(\bm u^{(n-1)},\bm{\bar{u}}^{(n)}); \bm{\bar u}^{(n)}+\tfrac{\ka}{2}, \bm u^{(n)}+\tfrac{\ka}{2}\big)=0 
\]
for $1\le j \le m_i$ and $1\le i \le n-2$ yields \eqref{son-La(v,u)} and \eqref{son-Bnr-BE2}. 
The above identity does not hold for $i=n-1$. Thus the Bethe equations \eqref{son-n-1-BE} for $u_j^{(n-1)}$ are obtained by evaluating directly
\[
\Res{v\to u^{(n-1)}_j} \La^{(1)}\big(v-\tfrac{n-1}2;\bm u^{(1\dots n-2)},(\bm u^{(n-1)},\bm{\bar{u}}^{(n)}); \bm{\bar u}^{(n)}+\tfrac{\ka}{2}, \bm u^{(n)}+\tfrac{\ka}{2}\big)=0 
\]
and the help of
\[
\La^-\big(v,\bm{\bar u}^{(n)}\big)\, \La^+\big(v-\tfrac12,\bm{\bar u}^{(n)}-\tfrac12\big) \La^+\big(v-\tfrac12,\bm u^{(n)}-\tfrac12\big) = \La^+\big(v,\bm{\bar u}^{(n)}\big) .
\]

The top-level Bethe equations \eqref{son-n-BE} for $u^{(n)}_j$ are obtained from equating to zero the unwanted terms.
However, some care must be taken so as not to exchange the order of the residue and limit.
Using the same arguments as in the proof of Corollary~\ref{C:son-nBethe-spec} and assuming \eqref{son-Bnr-BE2} and \eqref{son-n-1-BE}, so that \eqref{nBethe-del} holds, we write
\aln{
& \Res{w \rightarrow u_j} \big\{p(w) \, \tau^{(1)}(w;\bm \tu_{\si_j}; \bm u_{\si_j})\big\}^w \cdot 
\Phi_{lim}^{(1)}(\bm u^{(1...n-1)},\bm \tu_{\si_j} ;\bm 0, \bm \del_{\si_j}) 
\\
& \qu = \lim_{\eps \rightarrow 0} \Res{w \rightarrow u_j}\bigg( p(w)\, \tau^{(1)}(w;\bm \tu_{\si_j}-\bm \del_{\si_j} \eps; \bm u_{\si_j}) \cdot  \Phi^{(1)}(\bm u^{(1...n-2)},(\bm u^{(n-1)},\bm \tu_{\si_j} -\tfrac{\ka}{2}- \eps);\bm \tu_{\si_j}-\bm \del_{\si_j} \eps ; \bm u_{\si_j}) \\
& \hspace{3cm} + p(\tw)\,\tau^{(1)}(\tw;\bm \tu_{\si_j}; \bm u_{\si_j}+\bm{\hat \del}_{\si_j}\eps) \cdot  \Phi^{(1)}(\bm u^{(1...n-2)},(\bm u^{(n-1)},\bm \tu_{\si_j} -\tfrac{\ka}{2}- \eps);\bm \tu_{\si_j} ; \bm u_{\si_j}+\bm{\hat \del}_{\si_j}\eps ) 
\bigg) .
}
This expression equates to zero if
\aln{
& \lim_{\eps \rightarrow 0} 
\Res{w\rightarrow u_j} 
\Big( 
p(w) \, \Lambda^{(1)}\big(w;\bm u^{(1...n-2)},(\bm u^{(n-1)},\bm \tu -\tfrac{\ka}{2}-\eps);\bm \tu-\bm \del \eps ,\bm u\big)
\\
& \hspace{2.3cm} +
p(\tw) \, \Lambda^{(1)}\big(\tw;\bm u^{(1...n-2)},(\bm u^{(n-1)},\bm \tu -\tfrac{\ka}{2}-\eps);\bm \tu,\bm u+ \bm{\hat{\del}}\eps \big)
\Big)
\\
& = \lim_{\eps \rightarrow 0} \Res{w \rightarrow u_j} \bigg( p(w)\,\La^-\big(w-\tfrac{n-1}2,\bm u^{(n-1)}\big)\, \La^-\big(w-\tfrac{n-1}2,\bm \tw - \tfrac{\ka}{2} - \eps \big) \\ & \hspace{3cm}\times \La^+\big(w-\tfrac{n}2,\bm \tw - \bm \del \eps -\tfrac{n}{2}\big) \, \La^+\big(w-\tfrac{n}2,\bm u-\tfrac{n}{2}\big) \, \frac{\wt{\ga}_n(w)}{2w-n+1+\rho} 
\\
& \hspace{2.3cm} + p(\tw)\, \La^-\big(\tw-\tfrac{n-1}2,\bm u^{(n-1)}\big) \La^-\big(\tw-\tfrac{n-1}2,\bm \tu - \tfrac{\ka}{2} - \eps \big) \\ & \hspace{3cm}\times \La^+\big(\tw-\tfrac{n}2,\bm \tu -\tfrac{n}{2}\big) \, \La^+\big(\tw-\tfrac{n}2,\bm u+\bm{\hat{\del}} \eps-\tfrac{n}{2}\big) \, \frac{\wt{\ga}_n(\tw)}{2\tw-n+1+\rho} \bigg) = 0 .
}
Note that the terms that contain a pole at $u_j$ are $\La^+\big(w-\tfrac{n}2,\bm u-\tfrac{n}{2}\big)$ and $\La^+\big(\tw-\tfrac{n}2,\bm \tu -\tfrac{n}{2}\big)$.
Now evaluate the residue and use identities $\Lambda^{\pm}(\bar{v},\bm w) = \Lambda^{\mp}(v,\bm w)$, $\Lambda^{\pm}(v,\bm{\bar{w}}) = \Lambda^{\pm}(v,\bm w)$ and $p(w) = -p(\tw)$. Then, upon rewriting $u_i$'s in terms of $u^{(n)}_i$'s, we obtain 
\ali{
\lim_{\eps \rightarrow 0} \Bigg(& \La^-\big(u^{(n)}_j,\bm u^{(n-1)}\big) \La^-\big(u_j^{(n)},\bm{u}^{(n)} +\eps \big) 
\La^+\big(u^{(n)}_j-\tfrac{1}2,\bm{u}^{(n)} + \bm{\del} \eps  +\tfrac{1}{2}\big) 
\el
& \qq  \times \frac{\wt{\ga}_n(u^{(n)}_j+\tfrac{\ka}{2})}{2u^{(n)}_j+\rho-1} \prod_{\substack{i=1 \\ i\neq j}}^{m_n} \frac{u^{(n)}_j-u^{(n)}_i+1}{u^{(n)}_j-u^{(n)}_i} \frac{u^{(n)}_j+u^{(n)}_i+\rho}{u^{(n)}_j+u^{(n)}_i+\rho-1} 
\nn\\[.2em]
&- \La^+\big(u^{(n)}_j,\bm u^{(n-1)}\big) \La^+\big(u^{(n)}_j,\bm{u}^{(n)} + \eps \big) 
\La^-\big(u^{(n)}_j+\tfrac{1}2,\bm{u}^{(n)}+ \bm{\hat{\del}} \eps -\tfrac{1}{2}\big)
\nn \\[0.5em]
& \qq  \times
\frac{\wt{\ga}_n(\bar{u}^{(n)}_j+\tfrac{\ka}{2})}{2u^{(n)}_j+\rho+1} \prod_{\substack{i=1 \\ i\neq j}}^{m_n} \frac{u^{(n)}_j-u^{(n)}_i-1}{u^{(n)}_j-u^{(n)}_i} \frac{u^{(n)}_j+u^{(n)}_i+\rho}{u^{(n)}_j+u^{(n)}_i+\rho+1} \Bigg) = 0.
\label{son-n-BE-pre}
}
Observe that
\aln{
&\lim_{\eps \rightarrow 0}
\La^-\big(u_j^{(n)},\bm{u}^{(n)} +\eps \big) 
\La^+\big(u^{(n)}_j-\tfrac{1}2,\bm{u}^{(n)} + \bm{\del} \eps  +\tfrac{1}{2}\big) = {\del_j} \frac{2u_j^{(n)} + \rho-1}{2u_j^{(n)} + \rho} \frac{2u_j^{(n)} + \rho+1}{2u_j^{(n)} + \rho} \\
&
\hspace{3cm} \times \prod_{\substack{i=1 \\ i\neq j}}^{m_n} \frac{u^{(n)}_j-u^{(n)}_i-1}{u^{(n)}_j-u^{(n)}_i} \frac{u^{(n)}_j+u^{(n)}_i+\rho-1}{u^{(n)}_j+u^{(n)}_i+\rho}
 \frac{u^{(n)}_j-u^{(n)}_i}{u^{(n)}_j-u^{(n)}_i-1} \frac{u^{(n)}_j+u^{(n)}_i+\rho+1}{u^{(n)}_j+u^{(n)}_i+\rho}
}
and
\aln{
&\lim_{\eps \rightarrow 0}
\La^+\big(u_j^{(n)},\bm{u}^{(n)} +\eps \big) 
\La^-\big(u^{(n)}_j+\tfrac{1}2,\bm{u}^{(n)} + \bm{\hat\del} \eps  -\tfrac{1}{2}\big) = {\hat{\del}_j} \frac{2u_j^{(n)} + \rho-1}{2u_j^{(n)} + \rho} \frac{2u_j^{(n)} + \rho+1}{2u_j^{(n)} + \rho} \\
&
\hspace{3cm} \times \prod_{\substack{i=1 \\ i\neq j}}^{m_n} \frac{u^{(n)}_j-u^{(n)}_i+1}{u^{(n)}_j-u^{(n)}_i} \frac{u^{(n)}_j+u^{(n)}_i+\rho+1}{u^{(n)}_j+u^{(n)}_i+\rho}
 \frac{u^{(n)}_j-u^{(n)}_i}{u^{(n)}_j-u^{(n)}_i+1} \frac{u^{(n)}_j+u^{(n)}_i+\rho-1}{u^{(n)}_j+u^{(n)}_i+\rho} \,.
}
Hence taking the $\eps\to0$ limit in \eqref{son-n-BE-pre} gives
\aln{
& \La^-\big(u^{(n)}_j,\bm u^{(n-1)}\big)\, \del_j\, \frac{\wt{\ga}_n(u^{(n)}_j+\tfrac{\ka}{2})}{2u^{(n)}_j+\rho-1} \prod_{\substack{i=1 \\ i\neq j}}^{m_n} \frac{u^{(n)}_j-u^{(n)}_i+1}{u^{(n)}_j-u^{(n)}_i} \frac{u^{(n)}_j+u^{(n)}_i+\rho}{u^{(n)}_j+u^{(n)}_i+\rho-1} 
\\
& - \La^+\big(u^{(n)}_j,\bm u^{(n-1)}\big)\, \hat{\del}_j \,
\frac{\wt{\ga}_n(\bar{u}^{(n)}_j+\tfrac{\ka}{2})}{2u^{(n)}_j+\rho+1} \prod_{\substack{i=1 \\ i\neq j}}^{m_n} \frac{u^{(n)}_j-u^{(n)}_i-1}{u^{(n)}_j-u^{(n)}_i} \frac{u^{(n)}_j+u^{(n)}_i+\rho}{u^{(n)}_j+u^{(n)}_i+\rho+1} 
 = 0.
}
Recall that $\hat \del_j = -\del_j/(1-\del_j)$. We may thus rewrite the equality above as
\[
\frac{2u^{(n)}_j +\rho+1}{2u^{(n)}_j +\rho-1} \frac{\wt \gamma_n(u^{(n)}_j+\tfrac{\ka}{2}) }{\wt \gamma_n(\tfrac{\ka}{2}-u^{(n)}_j-\rho) }
\prod_{\substack{i=1 \\ i\neq j}}^{m_n} \frac{(u^{(n)}_j-u^{(n)}_i+1)(u^{(n)}_j + u^{(n)}_i +\rho + 1)}{(u^{(n)}_j-u^{(n)}_i-1)(u^{(n)}_j+u^{(n)}_i+\rho-1)}
= - \frac{1}{1-\del_j} 
\frac{\La^+(u^{(n)}_j,\bm u^{(n-1)})}{\La^-(u^{(n)}_j,\bm u^{(n-1)})}
.
\]
Substituting the definition of $\del_j$ from \eqref{son-del} and using $\La^+(u^{(n)}_j-\tfrac{1}{2},\bm u^{(n-2)})\,\La^-(u^{(n)}_j+\tfrac{1}{2},\bm u^{(n-2)})=1$ we obtain \eqref{son-n-BE}, as required.
\end{proof}

\begin{exam} \label{E:BV-son}
The orthogonal Bethe vector with a single top-level excitation and $m_1 = \ldots = m_{n-1} = 0$ is given by
\aln{
\Psi(u^{(n)}) 
&=
\bigg[-\frac{2u^{(n)}+\rho-1}{2u^{(n)}+\rho}  \big[B(u^{(n)}+\tfrac{\ka}2)\big]_{n-1,1} \big[\hat{A}^{(n-1)}(\bar{u}^{(n)}+\tfrac12)\big]_{11}
\\
& \qq + \frac{1}{1-\del} \frac{1}{2u^{(n)}+\rho+1} \big[B(u^{(n)}+\tfrac{\ka}2)\big]_{n-1,1}  \bigg(\big[\hat{A}^{(n-1)}(\bar{u}^{(n)}+\tfrac12)\big]_{22} - \frac{\big[\hat{A}^{(n-1)}(\bar{u}^{(n)}+\tfrac12)\big]_{11}}{2u^{(n)}+\rho} \bigg)
\\
&\qq 
-\frac{1}{1-\del} \frac{2u^{(n)}+\rho}{2u^{(n)}+\rho+1}  \big[B(u^{(n)}+\tfrac{\ka}2)\big]_{n,2} \, \bigg(\big[\hat{A}^{(n-1)}(\bar{u}^{(n)}+\tfrac12)\big]_{22} - \frac{\big[\hat{A}^{(n-1)}(\bar{u}^{(n)}+\tfrac12)\big]_{11}}{2u^{(n)}+\rho} \bigg)\bigg] \cdot \eta ,
}
where $\hat{A}^{(n-1)}(v)$ refers to the level-$(n\!-\!1)$ nested version of the $A$ operator of $S(v)$ obtained via \eqref{D-hat}.
Note that the level-$(n\!-\!1)$ excitations contribute only diagonal elements, which do not modify the vacuum vector. Hence the expression above may be simplified by using \eqref{son-del} and
\[
\big[\hat{A}^{(n)}(\bar{u}^{(n)})\big]_{11} \cdot \eta = -\frac{\wt\ga_{n}(\bar{u}^{(n)}+\tfrac{\ka}{2})}{2u^{(n)}+\rho} \,\, \eta ,
\qq
\big[\hat{A}^{(n-1)}(\bar{u}^{(n)}+\tfrac12)\big]_{11} \cdot \eta = -\frac{\wt\ga_{n-1}(\bar{u}^{(n)}+\tfrac{\ka}{2})}{2u^{(n)}+\rho-1} \,\, \eta ,x
\]
resulting in
\[
\Psi(u^{(n)}) =
\frac{\wt\ga_{n-1}(\bar{u}^{(n)}+\tfrac{\ka}{2})}{2u^{(n)}+\rho-1} \bigg( \big[S(u^{(n)}+\tfrac{\ka}2)\big]_{n-1,n+1}  - \big[S(u^{(n)}+\tfrac{\ka}2)\big]_{n,n+2} \bigg) \cdot \eta. 
\]
\end{exam}


\subsection{$SO_{2n}/(U_n \times U_n)$ and $SP_{2n}/(U_n \times U_n)$ magnets}  \label{sec:spin}

In this section we focus on orthogonal open spin chains with the bulk quantum space being an $\ell$--fold tensor product of the highest weight $\mfso_{2n}$ representations $L(\la)$ of weight $\la=(k,\dots,k,\pm k)$ with $k\in\tfrac12\Z_+$, and symplectic open spin chains with the bulk quantum space being an $\ell$--fold tensor product of the highest weight $\mfsp_{2n}$ representation $L(\la)$ of weight $\la=(k,\dots,k)$ with $k\in\Z_+$. 
Such spin chains were called $SO_{2n}/(U_n\times U_n)$ and $SO_{2n}/(U_n\times U_n)$ magnets in \cite{Rs1}.
We will consider both chains simultaneously.

Define operators $\mc{F}_{ij}$ by $F_{ij} L(\la) = \mc{F}_{ij} L(\la)$ and set $\mc{F}= \sum_{i,j=1}^{2n} e_{ij}\ot \mc{F}_{ij}$. 
The generating matrix $\mc{F}$ satisfies the quadratic identity 
\equ{
\mc{F}^2 = k(k+\ka)I + \ka \mc{F} . \label{FF=F}
}

\begin{rmk}
For $\mfg_{2n}=\mfso_{2n}$ and $k=\frac12$ an explicit realisation of the operators $\mc{F}_{ij}$ in terms of anti-commuting variables is given in the proof of Theorem 5.16 in \cite{AMR}. 
\end{rmk}

Define a Lax operator
\equ{
\mc{L}(u) = I + \frac{\mc{F}}{u-\frac\ka2} . \label{spin-lax}
}

\begin{lemma} \label{L:spin}
We have
\gat{
R(u-v) \mc{L}_1(u) \mc{L}_2(v) = \mc{L}_2(v) \mc{L}_1(u) R(u-v) , \label{spin-YBE}
\\
\mc{L}(u)\,\mc{L}(-u) = \mc{L}(u)\,\mc{L}^t(u+\ka) = \mc{L}^t(u+\ka)\,\mc{L}(u) = \frac{u^2-(k+\frac\ka2)^2}{u^2-\frac{\ka^2}4}\, I . \label{spin-unit}
}
\end{lemma}

\begin{proof}
Relation \eqref{spin-YBE} follows from Lemma 3.4 in \cite{IsMo}. We only need to show that the quantity $U$ in that Lemma equates to zero. Indeed, using symmetry $\mc{F}^t=-\mc{F}$ and identity \eqref{FF=F} we find 
\[
U = Q (\mc{F}_1 + \ka) \mc{F}_2 - \mc{F}_2(\mc{F}_1+\ka)Q = Q (\mc{F}_2\ka-\mc{F}^2_2) - (\mc{F}_2\ka-\mc{F}^2_2)Q = Q \,k(k+\ka) - k(k+\ka)Q = 0 .
\]
Relation \eqref{spin-unit} follows by a direct computation using similar arguments.
\end{proof}

Consider space $M$ defined in \eqref{M} and replace each $L(\la^{(i)})_{c_i}$ with a $c_i$-shifted $X(\mfso_{2n})$-module of weight $\la^{(i)}=(k_i,\dots,k_i,\pm k_i)$ with $2k_i\in\Z_+$, or a $c_i$-shifted $X(\mfsp_{2n})$-module of weight $\la^{(i)}=(k_i,\dots,k_i)$ with $k_i\in\Z_+$. (Recall \eqref{L0} and note that $L^0(\lambda^{(i)})_{c_{i}}\cong \C$ unless $\mfg_{2n}=\mfso_{2n}$ and $\la^{(i)}=(k_i,\dots,k_i,-k_i)$, in which case $L^0(\lambda^{(i)})_{c_{i}}\cong \C^n$.)
Then Proposition \ref{P:B(n,r)-rep} holds except \eqref{la_i(v)} should be replaced with
\equ{
\la^{(i)}_j(v) = 1+\frac{\la^{(i)}_j}{v-c_i-\tfrac\ka2}, \qq 
\bar\la^{(i)}_j(v) = 1-\frac{\la^{(i)}_j}{v-c_i-\tfrac\ka2} .
}
Finally, Theorems \ref{T:spn-spec} and \ref{T:son-spec} also hold.


\subsection{Hamiltonian for the fundamental open spin chain}

In this section, we discuss the case in which each bulk quantum space is the fundamental representation of $\mfg_{2n}$ and each $c_i=-\ka/2$, i.e., $M\cong (\C^{2n})^{\ot\ell}$. 
Additionally, set $\rho=0$.
Let $K(u)$ denote the $K$-matrix associated to a one-dimensional representation of $X_{\rho}(\mfg_{2n}, \mfg_{2n}^\theta)^{tw}$, as listed in Proposition~\ref{P:1-dim}.
Additionally, let $K^*(u)$ denote a solution of the dual reflection (obtained by substituting $u\to \tu$ and $v \to \tv$ in the reflection equation, so that $K^*(u) = K(\tu)$).
Note that in the above Sections we have taken $K^{*}(u) = I$.
For such an open spin chain, the transfer matrix given in Definition~\ref{D:tm} takes the form
\[
\tau(u) = \frac{g(u)}{2u-2\ka-\rho} \tr_a \left( K^*_a(u) \left(\prod_{i=1}^\ell R_{ai}(u) \right) K_{a}(u) \left( \prod_{i=\ell}^1 R_{ai}(u) \right) \right) .
\]  
Prior to extracting a Hamiltonian, we may cancel the poles at $u=0$ and $u=\ka$ by multiplying by a certain rational function in $u$ to obtain
\equ{ \label{poly_trans}
t(u) = \tr_a \left( \bbi{K}^*_a(u) \left(\prod_{i=1}^\ell \bbi{R}_{ai}(u) \right) \bbi{K}_{a}(u) \left( \prod_{i=\ell}^1 \bbi{R}_{ai}(u) \right) \right),
}
where 
\[
\bbi{R}(u) = -\frac{u(\ka-u)}{\ka} R(u) = -\frac{u(\ka-u)}{\ka} + \frac{\ka-u}{\ka} P +\frac{u}{\ka} Q \in \End(\C^{2n} \ot \C^{2n})[u],
\]
and $\bbi{K}(u), \bbi{K}^*(u)$ are normalised such that $\bbi{K}(0) = \bbi{K}^*(\ka) = I$, with  $\Tr \bbi{K}(\ka)$ and $\Tr \bbi{K}^*(0)$ both non-zero. 

\begin{prop}
The following Hamiltonian commutes with $\tau(u)$:
\equ{
H^0 = H^0_L + \sum_{i=1}^{\ell-1} H^0_{i,i+1} + H^0_R, \label{H0}
}
where
\[
H^0_L = \frac{\tr_a \big(\bbi{K}_a^*(0) H^0_{a1}\big)}{\tr \bbi{K}^*(0)}, 
\qq H^0_R = \tfrac12 \bbi{K}'_{\ell}(0), 
\qq H^0_{i,i+1} = P_{i,i+1} \pm \frac{Q_{i,i+1}}{\ka}.
\]
\end{prop}

\begin{proof}
Observe that $\bbi{R}(0) = P$, and $\bbi{K}(0)=I$, so Proposition 4 in \cite{Sk} allows us to extract a nearest neighbour interaction Hamiltonian for the system.
\end{proof}

The Hamiltonian \eqref{H0} is equivalent to the one considered in \cite{GKR}. The two-site interaction term $H_{i,i+1}$ is equivalent to that given in \cite{Rs1}%
\footnote{We believe there is a sign typo in (5.2) of \cite{Rs1}.}.

An additional Hamiltonian may be extracted from $t(u)$ by looking instead at the point $u=\ka$.
At this point, $\bbi{R}(\ka)$ is equal to $Q$, rather than $P$.
Nevertheless, the following procedure allows a nearest neighbour interaction Hamiltonian to be extracted.

\begin{prop}
The following Hamiltonian commutes with $\tau(u)$:
\equ{
H^{\ka} = H^{\ka}_L + \sum_{i=1}^{\ell-1} H^{\ka}_{i,i+1} + H^{\ka}_R, \label{Hk}
}
with
\[
H^{\ka}_L = \tfrac12 \big(\bbi{K}_1^{*\prime}(\ka)\big)^t, \qq H^{\ka}_R = \frac{\tr_a \big(H^{\ka}_{a \ell}\bbi{K}^t_a(\ka)\big)}{\tr \bbi{K}^t(\ka)}, \qq H^{\ka}_{i,i+1} = P_{i,i+1} \mp \frac{Q_{i,i+1}}{\ka}.
\]
\end{prop}

\begin{proof}
We begin by differentiating $t(u)$ at $u=\ka$ to obtain
\aln{
t'(\ka) = \tr_a \bigg(
&\bbi{K}^{*\prime}_a(\ka) \, Q_{a1} \cdots Q_{a \ell} \, \bbi{K}_a(\ka) \, Q_{a \ell} \cdots Q_{a1} +Q_{a1} \cdots Q_{a \ell} \, \bbi{K}'_a(\ka) \, Q_{a \ell} \cdots Q_{a1} \\
&+ \sum_{j=1}^\ell  \, Q_{a1} \cdots \bbi{R}'_{aj}(\ka) \cdots Q_{a \ell} \, \bbi{K}_a(\ka) \, Q_{a \ell} \cdots Q_{a1} 
+ \sum_{j=1}^\ell  \, Q_{a1} \cdots Q_{a \ell} \, \bbi{K}_a(\ka) \, Q_{a \ell} \cdots \bbi{R}'_{aj}(\ka) \cdots Q_{a1}
\bigg).
}
Repeated applications of $Q_{ai} M_{a} Q_{ai} = Q_{ai} \tr M$ and $\tr_a Q_{ai} = I$ allow us to reduce this to:
\aln{
t'(\ka) &= \tr_a \big(\bbi{K}_a^{*\prime}(\ka) Q_{a1} \big) \tr \bbi{K}(\ka) +\tr \bbi{K}'(\ka)
\\
&\qu + \sum_{j=1}^{\ell-1} \tr_a \big( \bbi{R}'_{aj}(\ka)Q_{a,j+1} Q_{aj}\big) \tr \bbi{K}(\ka) 
+ \tr_a \big( \bbi{R}'_{a\ell}(\ka) \bbi{K}_a(\ka) Q_{a \ell} \big)
\\
&\qu + \sum_{j=1}^{\ell-1} \tr_a \big( Q_{aj}Q_{a,j+1} \bbi{R}'_{aj}(\ka) \big) \tr \bbi{K}(\ka) 
+ \tr_a \big( Q_{a \ell} \bbi{K}_a(\ka) \bbi{R}'_{a\ell}(\ka) \big).
}
Since $\bbi{R}'(\ka) = I-P/\ka+Q/\ka$, it commutes with $Q$ acting on the same spaces, allowing us to apply the cyclicity of the partial trace. 
With this, and the identity $Q_{ai} M_a = Q_{ai} M^t_i$, we obtain
\[
t'(\ka) = \big(\bbi{K}_1^{*\prime}(\ka)\big)^t \tr \bbi{K}(\ka) +
\tr \bbi{K}'(\ka)+ 2\tr \bbi{K}(\ka) \sum_{j=1}^{\ell-1} \bbi{R}'_{j, j+1}(\ka)^t P_{j,j+1}
+ 2\tr_a \big(\bbi{K}_a(\ka) \bbi{R}'_{\ell a}(\ka) Q_{\ell a} \big).
\]
From here we divide by $\tr \bbi{K}(\ka)$ and subtract appropriate constants to extract the Hamiltonian.
\end{proof}

\begin{rmk}
Note that in the case where both conditions on $\bbi{K}$ and $\bbi{K}^*$ hold, the Hamiltonian $H^0 + H^{\ka}$ has nearest neighbour interaction in the bulk given by $P_{i,i+1}$.
\end{rmk}



\begin{thebibliography}{AAAI}


\bibitem[AACDFR1]{AACDFR1}
	D. Arnaudon, J. Avan, N. Cramp\'e, A. Doikou, L. Frappat, E. Ragoucy
	\emph{Classification of reflection matrices related to (super-)Yangians and application to open spin chain models},
	Nucl. Phys. B {\bf 668} (2003) 459--505, {\tt arXiv:math/0304150}.
	
\bibitem[AACDFR2]{AACDFR2}
	\bysame, 
	\emph{Bethe Ansatz equations and exact S matrices for the osp(M$|$2n) open super spin chain},
	Nucl. Phys. B {\bf 687} (2004) 257--278, {\tt arXiv:math-ph/0310042}.
	

\bibitem[AMR]{AMR} 
	D. Arnaudon, A. Molev, E. Ragoucy, 
	\emph{On the $R$-matrix realization of Yangians and their representations},  
	Ann. Henri Poincar\'e \textbf{7} (2006), no. 7-8, 1269--1325.
	{\tt arXiv:math/0511481}.

\bibitem[BeRa1]{BeRa1}
	S. Belliard and E. Ragoucy,
	{\it Nested Bethe ansatz for `all' open spin chains with diagonal boundary conditions},
	J. Phys. A {\bf 42} (2009) 205203,
	{\tt arXiv:arXiv:0902.0321}.
	
\bibitem[BeRa2]{BeRa2}	
	\bysame, 
	\emph{Nested Bethe ansatz for `all' open spin chains with diagonal boundary conditions}, 
	J. Phys. A {\bf 42} (2009) 205203,
	{\tt arXiv:0902.0321}.

\bibitem[BKWK]{BKWK}
	B. Berg, M. Karowski, P. Weisz, V.Kurak,
	{\it Factorized $U(n)$ symmetric $S$-matrices in two dimensions},
	Nucl. Phys. B {\bf 134} (1978), no.~1, 125--132

\bibitem[DMS]{DMS}
	G. W. Delius, N. J. MacKay, B. J. Short,
	{\it Boundary remnant of Yangian symmetry and the structure of rational reflection matrices}
	Phys. Lett. B {\bf 522} (2001) 335--344; Erratum-ibid. B {\bf 524} (2002) 401, {\tt arXiv:hep-th/0109115}.

\bibitem[DVK]{DVK}
	H.J. De Vega, M. Karowski,
	{\it Exact Bethe ansatz solution of O(2N) symmetric theories},
	Nuc. Phys. B {\bf 280} (1987) 225--254.
	
\bibitem[FST]{FST}
	L. D. Faddeev, E. K. Sklyanin and L. A. Takhtajan, 
	{\it Quantum Inverse Problem. I}, 
	Theor. Math. Phys. {\bf 40} (1979) 688--706.
	
\bibitem[FT]{FT}
	L. D. Faddeev and L. A. Takhtajan, 
	{\it The quantum method of the inverse problem and the Heisenberg XYZ model}, 
	Usp. Math. Nauk {\bf 34} (1979) 13; Russian Math. Surveys {\bf 34} (1979) 11 (Engl. transl.).

\bibitem[GKR]{GKR}
	Li Guang-Liang, Shi Kang-Jie, Yue Rui-Hong, 
	\emph{Algebraic Bethe Ansatz Solution to $C_N$ Vertex Model with Open Boundary Conditions}, 
	Commun. Theor. Phys. {\bf 44}, no. 1 (2005) 89--98.

\bibitem[GMR]{GMR}
    N. MacKay, A. Gerrard and V. Regelskis
    {\it Nested algebraic Bethe ansatz for open spin chains with even twisted Yangian symmetry},
    Ann. Henri Poincar\'e {\bf 20}, issue 2 (2019) 339-392, {\tt arXiv:1710.08409}.


\bibitem[GoPa]{GoPa}
	T.~Gombor, L.~Palla,
	{\it Algebraic Bethe Ansatz for O(2N) sigma models with integrable diagonal boundaries},
	JHEP {\bf 02} (2016) 158,
	{\tt arXiv:1511.03107}.
	
\bibitem[Go]{Go}
	T. Gombor,
	{\it Nonstandard Bethe Ansatz equations for open O(N) spin chains},
	Nucl. Phys. B {\bf 935} (2018) 310--343, {\tt arXiv:1712.03753}.
		
\bibitem[GmWa]{GmWa}
	R. Goodman, N. R. Wallach,
	{\it Symmetry, Representations and Invariants},
	eds. S. Axler, K. A. Ribet, Springer, 2009.

\bibitem[GR]{GR}
	N. Guay, V. Regelskis,
	{\it Twisted Yangians for symmetric pairs of types B, C, D},
	Math. Z. (2016) 284:131, {\tt arXiv:1407.5247}.

\bibitem[GRW1]{GRW1}
	N. Guay, V. Regelskis, C. Wendlandt,
	{\it Representations of twisted Yangians types B, C, D: I},
	Sel. Math. New Ser. {\bf 23} (2017), no. 3, 2071--2156, {\tt arXiv:1605.06733}. 

\bibitem[GRW2]{GRW2}
	\bysame,
	{\it Representations of twisted Yangians types B, C, D: II},
	Tran. Groups (2019), doi:10.1007/s00031-019-09514-x, {\tt arXiv:1708.00968}.
	
\bibitem[GRW3]{GRW3}
	\bysame,
	{\it Representations of twisted Yangians types B, C, D: III},
	in preparation.
	
\bibitem[GRW4]{GRW4}
	\bysame,
	{\it R-matrix presentation of orthogonal and symplectic quantum loop algebras and their representations},
	in preparation.
	
\bibitem[HLPRS1]{HLPRS1}	
	A. Hutsalyuk, A. Liashyk, S. Z. Pakuliak, E. Ragoucy, N. A. Slavnov, 
	\emph{Scalar products of Bethe vectors in the models with $\mfgl(m|n)$ symmetry}, 
	Nucl. Phys. B{\bf 923} (2017) 277--311,
	{\tt arXiv:1704.08173}.

\bibitem[HLPRS2]{HLPRS2}	
	\bysame
	\emph{Scalar products and norm of Bethe vectors for integrable models based on $U_q(\hat\mfgl_n)$}, SciPost Phys.~{\bf 4} (2018) 006,
	{\tt arXiv:1711.03867}.	
	

\bibitem[IsMo]{IsMo}
	A. P. Isaev, A. I. Molev,
	{\it Fusion procedure for the Brauer algebra},
	St. Petersburg Math. J. {\bf 22} (2011), 437--446, 
	{\tt arXiv:0812.4113}.
	
\bibitem[IMO]{IMO}
	A. P. Isaev, A. I. Molev, O. V. Ogievetsky,
	{\it A new fusion procedure for the Brauer algebra and evaluation homomorphisms},
	IMRN {\bf 11} (2012), 2571--2606, 
	{\tt arXiv:1101.1336}.	


\bibitem[IzKo]{IzKo}
	A. G. Izergin, V. E. Korepin, 
	\emph{The quantum inverse scattering method approach to correlation functions}, 
	Comm. Math. Phys. {\bf 94} (1984) 67--92.
	
	
\bibitem[KKMST1]{KKMST1}
	N. Kitanine, K. Kozlowski, J. M. Maillet, N. A. Slavnov, V. Terras,
	\emph{A form factor approach to the asymptotic behavior of correlation functions}, 
	J. Stat. Mech. (2011) P12010,
	{\tt arXiv:hep-th/1110.0803}.

\bibitem[KKMST2]{KKMST2}
	\bysame,
	\emph{Form factor approach to dynamical correlation functions in critical models},
	J. Stat. Mech. (2012) P09001,
	{\tt arXiv:1206.2630}.
	

\bibitem[KMST]{KMST}
	N. Kitanine, J. M. Maillet, N. A. Slavnov, V. Terras,
	\emph{Master equation for spin-spin correlation functions of the XXZ chain},
	Nucl. Phys. B{\bf712} (2005) 600--622,
	{\tt arXiv:hep-th/0406190}.
		
\bibitem[KMT]{KMT}
	N. Kitanine, J. M. Maillet, V. Terras, 
	\emph{Form factors of the XXZ Heisenberg spin-$1/2$ finite chain},
	Nucl. Phys. B{\bf 554} (1999) 647--678, 
	{\tt arXiv:math-ph/9807020}.	
		
\bibitem[Ko]{Ko}
	V. E. Korepin,
	\emph{Calculation of norms of Bethe wave functions},
	Comm. Math. Phys. {\bf 86}, no. 3 (1982) 391--418.

\bibitem[KrRs]{KrRs}
	A. N. Kirillov, N. Yu. Reshetikhin,
	{\it Representations of Yangians and multiplicities of occurrence of the irreducible components of the tensor product of representations of simple Lie algebras},
	J. of Sov. Math. {\bf 52} (1990) iss. 3, 3156--3164.
	

\bibitem[KuRs]{KuRs}
	P. P. Kulish, N. Yu. Reshetikhin,
	{\it Diagonalisation of GL(N) invariant transfer matrices and quantum N-wave system (Lee model)},
	J. Phys. A: Math. Gen. 16 (1983) L591-L596.
	
\bibitem[KuSk]{KuSk}
	P.~P. Kulish, E.~K. Sklyanin,
	{\it Solutions of the Yang-Baxter equation}, 
	J. Sov. Math. {\bf 19} (1982), 1596--1620. 	


\bibitem[MaRa]{MaRa}
	M.J. Martins and P.B. Ramos, 
	{\it The algebraic Bethe ansatz for rational braid-monoid lattice models},
	Nucl. Phys. B {\bf 500} (1997) 579,
	{\tt arXiv:hep-th/9703023}.


\bibitem[Mo1]{Mo1}
	A. I. Molev,
	{\it Finite-dimensional irreducible representations of twisted Yangians},
	J. Math. Phys. {\bf 39}, 5559-5600 (1998), {\tt arXiv:q-alg/9711022}

\bibitem[Mo2]{Mo2}
	\bysame,
	{\it Irreducibility criterion for tensor products of Yangian evaluation modules},
	Duke Math. J. {\bf 112}, 307--341  (2002).
	{\tt arXiv:math/0009183}.

\bibitem[Mo3]{Mo3}
	\bysame,
	{\it Yangians and Classical Lie Algebras}, 
	Mathematical Surveys and Monographs {\bf 143}, American Mathematical Society, Providence, RI, 2007, xviii+400 pp.
	
\bibitem[Mo4]{Mo4}
	\bysame,
	{\it Sugawara Operators for Classical Lie Algebras}, 
	Mathematical Surveys and Monographs {\bf 229}, American Mathematical Society, Providence, RI, 2018, xiv+304 pp.

\bibitem[MoMu]{MoMu}
	A. I. Molev, E. E. Mukhin,
	{\it Yangian characters and classical W-algebras},
	in {\it Conformal Field Theory, Automorphic Forms and Related Topics}, W. Kohnen, R. Weissauer, Eds, Springer, 2014, pp. 287--334.
	{\tt arXiv:1212.4032}.
	
\bibitem[MoRa]{MoRa} 
	A. I. Molev, E. Ragoucy, 
	\emph{Representations of reflection algebras}, 
	Rev. Math. Phys. \textbf{14} (2002), no. 3, 317--342.
	{\tt arXiv:math/0107213}.
	
	
\bibitem[PRS1]{PRS1}
	S. Pakuliak, E. Ragoucy, N. Slavnov,
	\emph{Bethe vectors of quantum integrable models based on $U_q(\hat\mfgl_n)$}, 
	J. Phys. A{\bf 47} (2014) 105202,
	{\tt arXiv:1310.3253}.

\bibitem[PRS2]{PRS2}
	\bysame,
	\emph{Bethe vectors for models based on the super-Yangian $Y(\mfgl(m|n))$}, 
	J. Integrable Systems {\bf 2} (2017) 1--31, 
	{\tt arXiv:1604.02311}.
	
\bibitem[PRS3]{PRS3}
	\bysame,
	\emph{Nested Algebraic Bethe Ansatz in integrable models: recent results},  	SciPost Phys. Lect. Notes {\bf 6} (2018),
	{\tt arXiv:1803.00103}.
	

\bibitem[Rs1]{Rs1}
	N. Yu. Reshetikhin,
	{\it Integrable models of quantum one-dimensional magnets with O(n) and Sp(2k) symmetry},
	Teor. Math. Fiz. {\bf 63} (1985) 347. 


\bibitem[Rs2]{Rs2}
	\bysame,
	{\it Algebraic Bethe ansatz for SO(N)-invariant transfer matrices},
	J. of Sov. Math. {\bf 54} (1991) 940.

\bibitem[Sk]{Sk}
	E. K. Sklyanin,
	{\it Boundary conditions for integrable quantum systems},
	J. Phys. A {\bf21} (1988) 2375.
	
	
\bibitem[Sl1]{Sl1}
	N. A. Slavnov,
	\emph{Calculation of scalar products of wave functions and form factors in
the framework of the algebraic Bethe ansatz},
	Theor. Math. Phys. {\bf 79} (1989) 502--508.
	
\bibitem[Sl2]{Sl2}
	\bysame,
	\emph{The algebraic Bethe ansatz and quantum integrable systems},
	Russ. Math. Surv. {\bf 62} (2007) 727.
	
\bibitem[Vl]{Vl}
	B. Vlaar,
	{\it Boundary transfer matrices and boundary quantum KZ equations},
	 	J. Math. Phys. {\bf 56} (2015) 071705. {\tt arXiv:1408.3364}.
	
\bibitem[ZaZa]{ZaZa}
	Al.~B. Zamolodchikov, Al.~B. Zamolodchikov,
	{\it Relativistic factorized $S$-matrix in two dimensions having $O(N)$ isotropic symmetry},
	Nucl. Phys. B {\bf 133} (1978), no.~3, 525--535. 



\end{thebibliography}
\end{document}